\documentclass[11pt]{article}
\usepackage{amsmath,amssymb,amsfonts}
\usepackage{geometry}
\usepackage{graphicx}
\usepackage{float}
\usepackage{algorithm}
\usepackage{algpseudocode}
\usepackage{cite}
\usepackage{authblk}
\usepackage{pdflscape}
\usepackage{booktabs}
\usepackage{array}
\usepackage{ragged2e}
\usepackage{afterpage}
\usepackage{titletoc}

\usepackage{xcolor}
\definecolor{mycolor}{RGB}{0,70,140}

\newcommand{\bfm}[1]{\mathbf{#1}}

\usepackage{sectsty}
\colorlet{sectioncolor}{black}      
\sectionfont{\color{sectioncolor}}
\subsectionfont{\color{sectioncolor}}

\colorlet{eqnumcolor}{black}      
\makeatletter
\newcommand{\eqtag@color}{\color{eqnumcolor}}
\def\tagform@#1{\maketag@@@{{\eqtag@color(#1)}}}
\DeclareRobustCommand{\eqref}[1]{%
    \textup{{\def\eqtag@color{\color{mycolor}}\tagform@{\ref{#1}}}}}
\makeatother

\usepackage[colorlinks=true,
            linkcolor=mycolor,
            citecolor=mycolor,
            urlcolor=mycolor]{hyperref}

\usepackage{amsthm}
\usepackage[framemethod=TikZ]{mdframed}
\definecolor{boxframe}{RGB}{0,70,140}
\definecolor{boxfill}{RGB}{246,249,253}
\mdfdefinestyle{boxedstyle}{%
    linecolor=boxframe,
    backgroundcolor=boxfill,
    linewidth=0.5pt,
    roundcorner=2pt,
    innertopmargin=8pt,
    innerbottommargin=8pt,
    innerleftmargin=10pt,
    innerrightmargin=10pt,
    skipabove=8pt,
    skipbelow=8pt,
}
\theoremstyle{plain}
\newmdtheoremenv[style=boxedstyle]{theorem}{Theorem}[section]
\theoremstyle{definition}
\newmdtheoremenv[style=boxedstyle]{problem}[theorem]{Problem}
\newmdtheoremenv[style=boxedstyle]{definition}[theorem]{Definition}

\definecolor{physframe}{RGB}{0,70,140}
\definecolor{physfill}{RGB}{246,249,253}
\mdfdefinestyle{physstyle}{%
    linecolor=physframe,
    backgroundcolor=physfill,
    linewidth=0.5pt,
    roundcorner=2pt,
    innertopmargin=8pt,
    innerbottommargin=8pt,
    innerleftmargin=10pt,
    innerrightmargin=10pt,
    skipabove=8pt,
    skipbelow=8pt,
}
\theoremstyle{plain}
\newmdtheoremenv[style=physstyle]{phystheorem}[theorem]{Theorem}

\definecolor{resultframe}{RGB}{160,35,35}
\definecolor{resultfill}{RGB}{255,238,238}
\mdfdefinestyle{resultstyle}{%
    linecolor=resultframe,
    backgroundcolor=resultfill,
    linewidth=0.6pt,
    roundcorner=2pt,
    innertopmargin=8pt,
    innerbottommargin=8pt,
    innerleftmargin=10pt,
    innerrightmargin=10pt,
    skipabove=8pt,
    skipbelow=8pt,
}

\usepackage{etoolbox}
\definecolor{algoframe}{gray}{0.5}
\definecolor{algofill}{gray}{0.96}
\mdfdefinestyle{algostyle}{%
    linecolor=algoframe,
    backgroundcolor=algofill,
    linewidth=0.5pt,
    roundcorner=2pt,
    innertopmargin=6pt,
    innerbottommargin=6pt,
    innerleftmargin=10pt,
    innerrightmargin=10pt,
    skipabove=8pt,
    skipbelow=8pt,
}

\BeforeBeginEnvironment{algorithm}{\begin{mdframed}[style=algostyle]}
\AfterEndEnvironment{algorithm}{\end{mdframed}}

\title{Unifying Physical Backpropagation}
\author[1]{Cyrill B\"osch\thanks{cyrillboesch@gmail.com}}
\author[2]{Yi\u{g}ithan Gediz}
\author[3]{Clara C. Wanjura}
\author[2]{Hakan E. T\"ureci}
\affil[1]{Department of Computer Science, Princeton University, Princeton, NJ 08540, USA}
\affil[2]{Department of Electrical and Computer Engineering, Princeton University, Princeton, NJ 08540, USA}
\affil[3]{Max Planck Institute for the Science of Light, Staudtstra\ss{}e 2, 91058 Erlangen, Germany}

\date{September 22, 2026}

\begin{document}
\maketitle

\begin{abstract}
Physical computing systems exploit device dynamics for computation, but their gradient-based optimization is challenging: backpropagation through a digital twin suffers from a model--reality gap. On-device gradient computation could resolve this issue, and a handful of theoretical and experimental studies have proposed ways to achieve it. Yet a unifying theory identifying when a physical system can compute the gradient of its own performance has been missing. Here we develop such a unification based on the adjoint method: we identify sufficient conditions under which the adjoint field required for formally exact gradients can be generated on the same hardware that performs the computation. Linear and nonlinear systems obey fundamentally different conditions: for linear systems, damping or gain is admissible provided reciprocity is preserved. For nonlinear trajectory systems, the sufficient conditions are reciprocity of the linearized system and the existence of a time-reversal mirror.
Algorithmically, the nonlinear case requires infinitesimal nudging, whereas linear systems admit a finite-amplitude experiment.
We recover (quantum) Equilibrium Propagation, Hamiltonian echo backpropagation, fully forward mode training and in situ gradient methods in integrated-photonic and free-space-optical systems.
Finally, we show that reciprocity is a special case of more general intertwining conditions.
For linear systems, these permit exact on-device gradients in a class of non-Hermitian, non-reciprocal systems.
For nonlinear trajectories, they combine with generalized time-reversal mirrors to cover, e.g., $\mathcal{PT}$-symmetric equations.
The framework also includes time-dependent parameters and Onsager-reciprocal dynamics, providing a unified basis for formally exact physical learning.
\end{abstract}

\startcontents[main]
\section*{\contentsname}
\begingroup\hypersetup{linkcolor=sectioncolor}
\printcontents[main]{}{0}{\setcounter{tocdepth}{1}}
\endgroup

\section{Introduction}\label{s:introduction}
Physical computing uses the controlled dynamics or steady-state response of a material or device for information processing~\cite{jaeger2023toward, hughes2019wave, zangeneh2021analogue}. Inputs may be encoded in initial or boundary conditions, external drives, or programmable parameters, and outputs are measured from selected physical degrees of freedom. Representative platforms include mechanical and elastic
networks~\cite{dubvcek2024sensor}, analog electrical circuits~\cite{dillavou2022demonstration}, spintronic and memristive devices~\cite{sebastian2020memory, markovic2020physics},
integrated photonic circuits~\cite{shastri2021photonics}, nonlinear multimode fibers~\cite{teugin2021scalable}, free-space optical systems~\cite{wetzstein2020inference,lin2018all, mcmahon2023physics} and programmable waveguide systems~\cite{onodera2024scaling, yanagimoto2026programmable}. These platforms span transient, fixed-point, and frequency-domain computations.

A common special case is Physical Reservoir Computing~\cite{tanaka2019recent}, in which the internal physical dynamics is fixed and only a linear readout layer is trained. Here, instead, we address the optimization of parameters that enter the physical dynamics itself, which remains a major challenge~\cite{momeni2025training}.
Training only a digital twin~\cite{nikkhah2024inverse} is vulnerable to a model--reality gap.
Physics-aware training, in which the forward pass is performed in the physical
device but differentiation is performed using a digital model, can reduce this
gap, but still requires a differentiable surrogate for the backward pass and
communication between the physical and digital domains~\cite{wright2022deep,oguz2024optical}.
Over the past decade, in situ (equivalently, on-device) training methods, in which the
physical system participates in both the forward computation and the
generation of the learning signal, have emerged as a promising
alternative~\cite{hermans2015trainable}. These approaches
range from heuristic or approximate rules to methods that compute formally
exact gradients. Approximate rules can be simple and local, but their updates
need not coincide with the true gradient, and their convergence and performance
can therefore be system dependent~\cite{momeni2025training}. This motivates
formally exact physical-gradient constructions, which have been developed for
free-space optics~\cite{zhou2020situ,mididoddi2025threading}, integrated
photonics~\cite{pai2023experimentally}, mechanics~\cite{li2024training},
resistive electrical networks~\cite{lin2026train}, and
Hamiltonian systems~\cite{lopez2023self}, among other platforms.

In this article, we unify these formally exact physical backpropagation methods using the adjoint framework of PDE-constrained optimization~\cite{pontryagin2018mathematical,fichtner2010full} and extend them to broader classes of physical systems. The main contribution of this work is twofold. First, we show that seemingly disparate physical backpropagation algorithms are instances of a single adjoint construction: this common formulation lets us compare and extend these algorithms. Second, we identify sufficient conditions under which the required adjoint field, and hence the cost function gradient, can be obtained from a physical experiment that runs in the \emph{same} system. The only allowed modifications between the forward and adjoint experiments are changes of initial conditions, external forces, and reversal of explicit time dependences.

We show that linear and nonlinear systems give rise to different physical experiments and different conditions. Linear reciprocal systems admit direct, finite-amplitude adjoint experiments, and damping or gain is compatible with exact physical gradients as long as reciprocity is preserved.
For example, trajectory gradients can be obtained physically even in linear reciprocal overdamped systems~\cite{hermans2015trainable}, whereas the corresponding nonlinear trajectory construction is obstructed. Under suitable conditions, Equilibrium Propagation~\cite{scellier2017equilibrium} emerges for fixed-point engineering.

Nonlinear trajectory systems are more restrictive: they require infinitesimal nudging, a particular form of time-reversal symmetry and reciprocity of the linearized dynamics. The latter two are independent conditions. In the nonlinear regime, ordinary damping or gain is generically obstructive, and even the ability to experimentally swap damping and gain does not in general restore a same-hardware adjoint construction.

The adjoint framework recovers a broad range of formally exact physical backpropagation algorithms: Equilibrium Propagation~\cite{scellier2017equilibrium}, the exact limit of complex-valued Near-Equilibrium Propagation~\cite{sajnok2025near}, Quantum Equilibrium Propagation~\cite{Wanjura2025QuantumEquilibrium,massar2024quantumthermal,scellier2024quantum}, Hamiltonian echo backpropagation and its recurrent variants~\cite{lopez2023self,pourcel2025learning}, early reciprocity-based recurrent learning~\cite{hermans2015trainable}, and physical gradient algorithms from integrated photonics~\cite{pai2023experimentally,hughes2018training}, free-space and diffractive optics~\cite{zhou2020situ,mididoddi2025threading}, linear wave scattering~\cite{wanjura2024fully,guillamon2025situ}, static and topological mechanics~\cite{li2024training,li2025topological}, and the projector-based analytical gradient for passive resistor networks~\cite{lin2026train}. For optical schemes with approximate gradients through the nonlinear subsystems~\cite{wagner1987multilayer,guo2021backpropagation,spall2025training}, the framework recovers only the exact backpropagation through the linear layers.
We then show that ordinary reciprocity can be generalized to an intertwining condition, in which the forward operator and its transpose are related by a constant invertible transformation. This construction yields exact physical gradients in a class of non-Hermitian and non-reciprocal systems. In the linear case, we call such systems ``twisted reciprocal'' and illustrate them with a Hatano--Nelson chain~\cite{hatano1996localization}. Fully Forward Mode (FFM) training~\cite{xue2024fully} is a special case of this construction. For nonlinear systems, we combine intertwining with a suitable time-reversal mirror and give examples involving nonlinear Hatano--Nelson chains and $\mathcal{PT}$-symmetric dynamics. Onsager reciprocity provides another special case, with the intertwining transformation supplied by a physical mobility (Sec.~\ref{s:onsager-reciprocity}). We further establish sufficient conditions for physical backpropagation in non-autonomous systems, making the temporal modulation of internal forces a trainable resource.

\paragraph{Disclaimer.}
In this article we do not cover coupled~\cite{stern2021supervised,stern2023learning, dillavou2024machine,stern2024training,dillavou2022demonstration, mcginnis2026coercivity} and heuristic~\cite{cheng2024multimodal,de2026multimodal, stern2026contrastive} physical learning rules. We also do not discuss forward-forward and related local-learning schemes for physical systems, which forgo gradients of a global loss function in favor of layer-local objectives~\cite{oguz2023forward, momeni2023backpropagation, wang2026training}. Furthermore, we do not cover probabilistic physical computing machines. Boltzmann machines~\cite{bohm2022noise} and score-based generative models~\cite{bosch2025local} have been shown to give rise to exact (unbiased) local learning rules, although these do not derive from the adjoint method. We consider a setup where the gradient is calculated from physically observed forward and adjoint fields. We are agnostic to how the parameter updates are then applied and assume for simplicity that an external teacher does so. This means we do not cover the problem of designing a physical system that adapts by itself~\cite{wagner1987multilayer, lopez2023self, dillavou2022demonstration, dillavou2024machine}.

\paragraph{Structure.}
In Sec.~\ref{s:problem-statement} we formally define the optimization problem and introduce the adjoint method using a second-order real-valued system. In Sec.~\ref{s:physical-question} we precisely define which part of the gradient computation is performed by the physical system and provide a summary of our results across various systems. In Sec.~\ref{s:local-learning-rules} we introduce the notion of local learning rules and discuss the difference between knowing the model and knowing the parameters. In Sections~\ref{s:physical-adjoints}--\ref{s:stationary} we present the main results: sufficiency conditions for obtaining the adjoint fields physically in second- (Sec.~\ref{s:physical-adjoints}) and first-order overdamped (Sec.~\ref{s:first-order}) real-valued systems, in Schr\"odinger-type complex-valued systems (Sec.~\ref{s:schrodinger}), and in stationary problems (Sec.~\ref{s:stationary}), including Equilibrium Propagation, Quantum Equilibrium Propagation (Sec.~\ref{ss:quantum-EP}), and Helmholtz-type constructions.
At the end of each section we identify which existing physical learning frameworks are recovered by the adjoint formalism.
Section~\ref{s:symmetry-adjoints} treats general intertwiners and time-reversal mirrors, including fully forward mode training and twisted-reciprocal Hatano--Nelson chains. Section~\ref{s:onsager-reciprocity} discusses Onsager reciprocity within this construction.

A self-contained appendix derives and discusses each case in more detail. We also provide a notebook for each section, with an implementation of the corresponding physical adjoint algorithm~\cite{physicalAdjointNotebooks}.


\section{Introducing the adjoint method}\label{s:problem-statement}
We consider $N$-dimensional discretized systems governed by ODEs, though the theory also applies to fields governed by PDEs. If the ODE resulted from discretizing a continuous PDE through, e.g., the finite element method, then, in the language of PDE-constrained optimization, this would correspond to the principle of ``discretize-then-optimize''~\cite{fichtner2010full}.
We denote the physical state by $\bfm{u}(t)\in \mathbb{R}^N$ and the trainable parameters by $\bfm{p}\in \mathbb{R}^M$. Inputs may be encoded in initial conditions, external drives~\cite{richardson2025nonlinear}, or directly in parameters. The latter endows linear systems with nonlinear computing capabilities~\cite{wanjura2024fully,xia2024nonlinear,yildirim2024nonlinear}. Suppressing explicit encoders, decoders, and batch averages, we write the cost as
\begin{equation}\label{eq:cost-trajectory-terminal}
    \chi[\bfm{u}]
    =
    \int_0^T \theta\bigl[\mathbb{P}\bfm{u}(t)\bigr]\,dt
    +
    \psi\bigl[\mathbb{P}\bfm{u}(T),\mathbb{P}\dot{\bfm{u}}(T)\bigr],
\end{equation}
where $\mathbb{P} = \mathbb{P}^T = \mathbb{P}^2$ projects onto the measured or output degrees of freedom (DOFs) which encode the solution to the computation. The first term is a trajectory loss and the second a final-state loss.

We now introduce the adjoint method~\cite{pontryagin2018mathematical,Giles2000, Plessix2006,VeronisDutton2004,Werschnik2007, Aronna2019, Brandwood1983, Hjorungnes2007, Sorber2012} that we will use throughout the text. The constructions used here are rooted in PDE-constrained optimization, where an objective is minimized subject to a governing equation enforced as a constraint through Lagrange multipliers~\cite{hinze2009optimization}. A full, self-contained derivation is given in Appendix~\ref{apA}.

\begin{figure}[t!]
    \centering
    \includegraphics[width=1\textwidth]{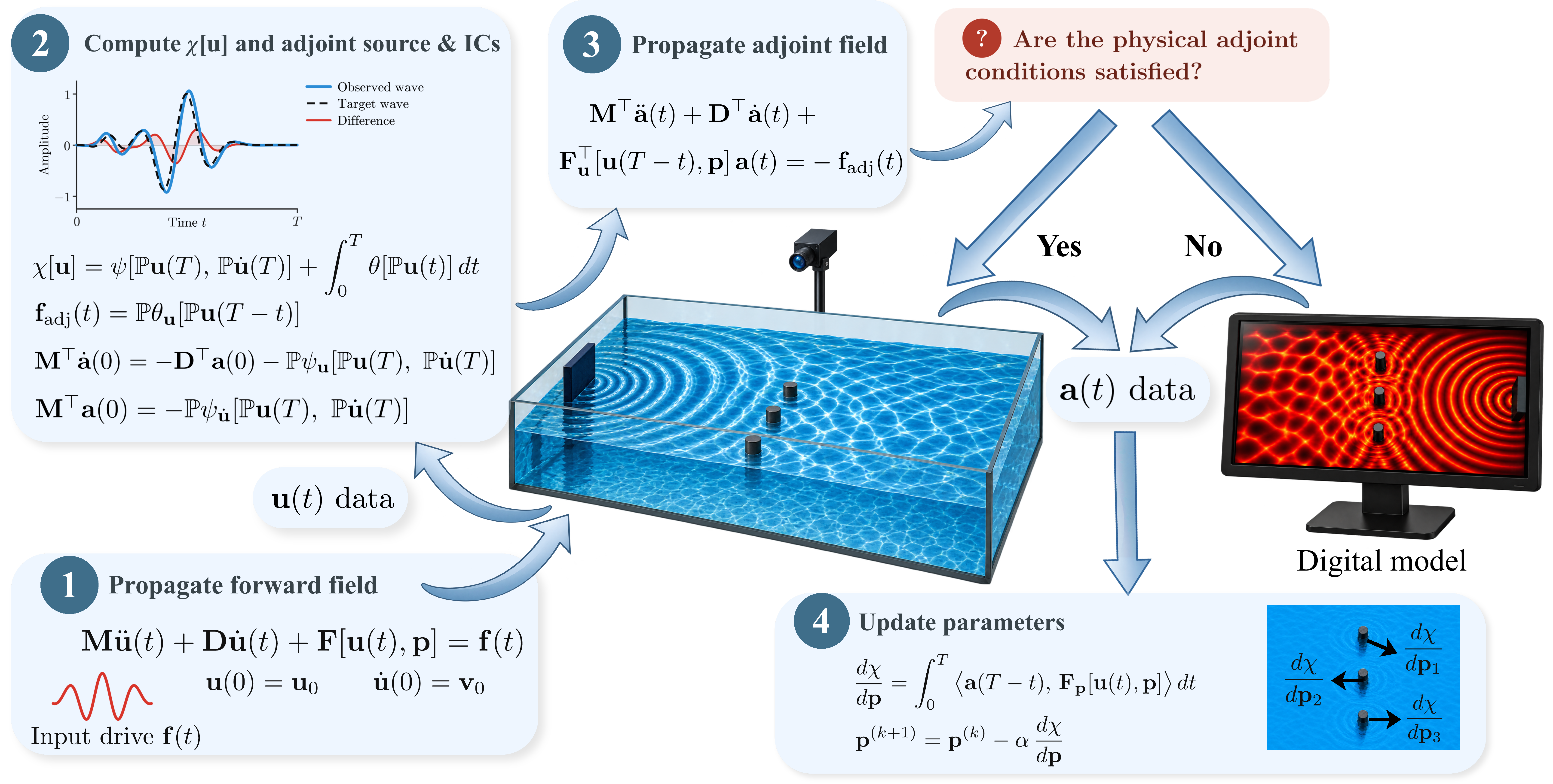}
    \caption{Conceptual overview of optimizing a physical system with the adjoint method. (1) The physical state $\bfm{u}(t)$ evolves under the forward dynamics $\bfm{M}\ddot{\bfm{u}}+\bfm{D}\dot{\bfm{u}}+\bfm{F}[\bfm{u},\bfm{p}]=\bfm{f}(t)$, with mass matrix $\bfm{M}$, damping matrix $\bfm{D}$, internal forces $\bfm{F}$ carrying the trainable parameters $\bfm{p}$, external drive $\bfm{f}$, and initial conditions $\bfm{u}(0)=\bfm{u}_0$, $\dot{\bfm{u}}(0)=\bfm{v}_0$.
    (2) The recorded trajectory and terminal state determine the cost $\chi[\bfm{u}]=\psi[\mathbb{P}\bfm{u}(T),\mathbb{P}\dot{\bfm{u}}(T)]+\int_0^T\theta[\mathbb{P}\bfm{u}(t)]\,dt$, where $\mathbb{P}$ projects onto the measured output degrees of freedom and $\psi$ and $\theta$ are the terminal and trajectory losses. Evaluating the derivatives of the losses provides the drive and initial conditions (ICs) for the adjoint field $\bfm{a}(t)$.
    (3) The adjoint dynamics involve the transposed mass and damping matrices and the linearized internal forces $\bfm{F}_{\bfm{u}}^T$, evaluated along the time-reversed forward trajectory $\bfm{u}(T-t)$. They therefore generally differ from the forward dynamics.
    A physical forward pass combined with a backward pass through a differentiable digital model is known as physics-aware training~\cite{wright2022deep,oguz2024optical}.
    Here we ask under what conditions the same physical hardware can generate the adjoint field, without a digital model for the backward pass.
    Finally, (4) the gradient $d\chi/d\bfm{p}$ is evaluated from overlap integrals of the recorded forward and adjoint fields, and the physical parameters are updated.}
    \label{fig:general_adjoint_method}
\end{figure}

For concreteness, consider a second-order real-valued system
\begin{equation}\label{eq:adjoint-method-second-order-state}
    \bfm{M}(\bfm{p}_M)\ddot{\bfm{u}}(t)
    +
    \bfm{D}(\bfm{p}_D)\dot{\bfm{u}}(t)
    +
    \bfm{F}\bigl[\bfm{u}(t),\bfm{p}_F,t\bigr]
    -
    \bfm{f}(t,\bfm{p}_f)
    =
    \bfm{0},
    \qquad
    \bfm{u}(0)=\bfm{u}_0(\bfm{p}_{u_0}),
    \qquad
    \dot{\bfm{u}}(0)=\bfm{v}_0(\bfm{p}_{v_0}).
\end{equation}
Here $\bfm{M}$ is the invertible mass matrix, $\bfm{D}$ the damping or gain matrix, $\bfm{F}$ the internal force, which carries the constitutive law of the system and is in general nonlinear in $\bfm{u}$ and explicitly time dependent, and $\bfm{f}$ the external drive; $\bfm{u}_0$ and $\bfm{v}_0$ are the initial state and its rate of change. We use mechanical language throughout, but the same three operators appear as the inertial, dissipative, and restoring terms of electrical, optical, and other physical platforms. Each operator carries its own block of trainable parameters, collected into $\bfm{p}=(\bfm{p}_M,\bfm{p}_D,\bfm{p}_F,\bfm{p}_f,\bfm{p}_{u_0},\bfm{p}_{v_0})$. The constrained optimization problem is
\begin{equation}\label{eq:abstract-opt}
    \min_{\bfm{p}}\;\chi[\bfm{u}]
    \qquad
    \text{subject to}
    \qquad
    \eqref{eq:adjoint-method-second-order-state}.
\end{equation}
The adjoint method replaces the differentiation of the full trajectory $\bfm{u}(t;\bfm{p})$ with respect to every parameter by a single adjoint solve. Written in forward time, the adjoint field $\bfm{a}(t)$ satisfies (see Appendix~\ref{apA})
\begin{equation}\label{eq:adjoint-method-second-order-adjoint}
    \bfm{M}^T\ddot{\bfm{a}}(t)
    +
    \bfm{D}^T\dot{\bfm{a}}(t)
    +
    \bfm{F}_{\bfm{u}}^T
    \bigl[\bfm{u}(T-t),\bfm{p}_F,T-t\bigr]\bfm{a}(t)
    +\bfm{f}_{\text{adj}}(t)
    =
    \bfm{0},
\end{equation}
where $(\cdot)_\bfm{x}$ denotes a partial derivative, i.e. $\bfm{F}_{\bfm{u}} := \partial \bfm{F}/\partial\bfm{u}$ (in the parameters, $\bfm{p}_X$, the subscripts are labels).
The adjoint source is given by
\begin{equation}\label{eq:adjoint-source}
    \bfm{f}_{\text{adj}}(t)=\mathbb{P}\theta_{\bfm{u}}
    \bigl[\mathbb{P}\bfm{u}(T-t)\bigr]
\end{equation}
and the adjoint initial conditions are
\begin{align}
    \bfm{M}^T\bfm{a}(0)
    &=
    -\,\mathbb{P}\psi_{\dot{\bfm{u}}}
    \bigl[\mathbb{P}\bfm{u}(T),\mathbb{P}\dot{\bfm{u}}(T)\bigr],
    \label{eq:adjoint-method-second-order-adjoint-ic1}\\
    \bfm{M}^T\dot{\bfm{a}}(0)
    &=
    -\,\bfm{D}^T\bfm{a}(0)
    -\,\mathbb{P}\psi_{\bfm{u}}
    \bigl[\mathbb{P}\bfm{u}(T),\mathbb{P}\dot{\bfm{u}}(T)\bigr].
    \label{eq:adjoint-method-second-order-adjoint-ic2}
\end{align}
The derivative of the trajectory loss, $\mathbb{P}\theta_{\bfm{u}}$, is therefore a forcing term in the adjoint experiment. The derivatives of the final-state loss, $\mathbb{P}\psi_{\bfm{u}}$ and $\mathbb{P}\psi_{\dot{\bfm{u}}}$, set the adjoint initial displacement and velocity. If there is no trajectory loss, the adjoint is unforced; if there is no terminal loss, the adjoint starts from zero.

Once $\bfm{a}(t)$ is known, the gradient is obtained from overlap integrals (see Appendix~\ref{apA:ss:gradient})
\begin{align}
    \frac{d\chi}{d\bfm{p}_M}
    &=
    \int_0^T
    \left\langle
        \bfm{a}(T-t),
        \bfm{M}_{\bfm{p}_M}\ddot{\bfm{u}}(t)
    \right\rangle dt,
    \label{eq:adjoint-method-gradient-pM}\\
    \frac{d\chi}{d\bfm{p}_D}
    &=
    \int_0^T
    \left\langle
        \bfm{a}(T-t),
        \bfm{D}_{\bfm{p}_D}\dot{\bfm{u}}(t)
    \right\rangle dt,
    \label{eq:adjoint-method-gradient-pD}\\
    \frac{d\chi}{d\bfm{p}_F}
    &=
    \int_0^T
    \left\langle
        \bfm{a}(T-t),
        \bfm{F}_{\bfm{p}_F}\bigl[\bfm{u}(t),\bfm{p}_F,t\bigr]
    \right\rangle dt,
    \label{eq:adjoint-method-gradient-pF}\\
    \frac{d\chi}{d\bfm{p}_f}
    &=
    -\int_0^T
    \left\langle
        \bfm{a}(T-t),
        \bfm{f}_{\bfm{p}_f}(t,\bfm{p}_f)
    \right\rangle dt,
    \label{eq:adjoint-method-gradient-pf}\\
    \frac{d\chi}{d\bfm{p}_{u_0}}
    &=
    -\,\left\langle
    \bigl[
        \bfm{M}^T\dot{\bfm{a}}(T)+\bfm{D}^T\bfm{a}(T)\bigr],\bfm{u}_{0,\bfm{p}_{u_0}} \right\rangle
    ,
    \label{eq:adjoint-method-gradient-u0}\\
    \frac{d\chi}{d\bfm{p}_{v_0}}
    &=
    -\,\left\langle\bfm{M}^T\bfm{a}(T),\bfm{v}_{0,\bfm{p}_{v_0}} \right\rangle,
    \label{eq:adjoint-method-gradient-v0}
\end{align}
where $ \left\langle \bfm{x},\bfm{y} \right\rangle = \bfm{x}^T\bfm{y}$ is the Euclidean inner product. 
Thus the gradient with respect to dynamical parameters is an overlap between the adjoint field and the parameter sensitivity of the respective forward operators. The gradient with respect to initial-condition parameters is an endpoint overlap. In particular, if the initial displacement and velocity themselves are trainable, then
\begin{equation}
    \frac{d\chi}{d\bfm{u}_0}
    =
    -\bfm{M}^T\dot{\bfm{a}}(T)-\bfm{D}^T\bfm{a}(T),
    \qquad
    \frac{d\chi}{d\bfm{v}_0}
    =
    -\bfm{M}^T\bfm{a}(T).
\end{equation}

\paragraph{Interpretation.} The overlap integrals show that the time-reversed adjoint field $\bfm{a}(T-t)$ encodes the sensitivity of the loss with respect to a force applied at time $t$ during the forward evolution. All parameters act in this way. Writing the state equation as $\bfm{M}\ddot{\bfm{u}}+\bfm{D}\dot{\bfm{u}}+\bfm{F}=\bfm{f}$, a parameter change $\delta\bfm{p}$ perturbs the net force acting on the system along the forward trajectory by
\begin{equation}\label{eq:effective-force-perturbation}
    \delta\bfm{f}(t)
    =
    \bfm{f}_{\bfm{p}_f}\delta\bfm{p}_f
    -\bfm{M}_{\bfm{p}_M}\ddot{\bfm{u}}(t)\,\delta\bfm{p}_M
    -\bfm{D}_{\bfm{p}_D}\dot{\bfm{u}}(t)\,\delta\bfm{p}_D
    -\bfm{F}_{\bfm{p}_F}\bigl[\bfm{u}(t),\bfm{p}_F,t\bigr]\,\delta\bfm{p}_F,
\end{equation}
where the internal contributions enter with the opposite sign because they act against the drive. The trajectory gradients~\eqref{eq:adjoint-method-gradient-pM}--\eqref{eq:adjoint-method-gradient-pf} then collapse into the single statement $\delta\chi=-\int_0^T\langle\bfm{a}(T-t),\delta\bfm{f}(t)\rangle\,dt$.
If the force change aligns with the adjoint field, the cost therefore decreases at the largest rate per unit of forcing, while a force orthogonal to it leaves the cost unchanged. The endpoint overlaps~\eqref{eq:adjoint-method-gradient-u0} and~\eqref{eq:adjoint-method-gradient-v0} have the same interpretation for the initial conditions.

\paragraph{The arrow of time.} As discussed in Appendix~\ref{apA:ss:forward-time}, the adjoint state as it emerges from the derivation evolves backward in time.
The equations above already incorporate the change of variable $t\mapsto T-t$ that recasts this evolution in forward time; this will become crucial for the physical realization of the adjoint field in Sec.~\ref{s:physical-adjoints}.

\paragraph{Algorithm.} The algorithm is summarized in Fig.~\ref{fig:general_adjoint_method}. One first solves the forward problem for the forward field $\bfm{u}(t)$, either digitally or, in physics-aware training, on the physical system itself. Next, the adjoint problem, which in general differs from the forward problem through the transposes and the linearization, is solved digitally for $\bfm{a}(t)$. The gradient then follows from the overlap integrals in~\eqref{eq:adjoint-method-gradient-pM}--\eqref{eq:adjoint-method-gradient-v0}, and the parameters are finally updated according to a chosen optimization algorithm.

\paragraph{A brief note on computational complexity.} The advantage of the adjoint method in digital optimization is that the cost of computing the gradient is essentially independent of the number of parameters~\cite{Giles2000,Plessix2006,fichtner2010full}. Finite-difference methods and forward sensitivity analysis require a separate solve for every entry of $\bfm{p}$, so their computational cost grows linearly with the parameter count $\dim(\bfm{p})$. The adjoint method, by contrast, requires only two trajectory solves regardless
of how many parameters are trained: one forward solve for $\bfm{u}(t)$ and one adjoint solve for $\bfm{a}(t)$. 
All gradient components are then read off from the overlap integrals~\eqref{eq:adjoint-method-gradient-pM}--\eqref{eq:adjoint-method-gradient-v0}, which are cheap because the only ingredients they require are the parameter sensitivities of the forward operators, such as $\bfm{M}_{\bfm{p}_M}$, $\bfm{D}_{\bfm{p}_D}$, and $\bfm{F}_{\bfm{p}_F}$, which are typically sparse, structured, and available in closed form from the model itself (see Sec.~\ref{s:local-learning-rules}).

\paragraph{How to test the gradient?} An adjoint gradient implementation can be checked with a Taylor test. We detail this procedure in Appendix~\ref{a:taylor-test}. We provide a notebook for each section, with an implementation of the corresponding physical adjoint algorithm and a Taylor test that numerically verifies the computed gradients~\cite{physicalAdjointNotebooks}.

\paragraph{Forward and adjoint systems are different.} The adjoint construction is exact and fully general: it applies to any physical system, linear or nonlinear, conservative or dissipative. The adjoint problem~\eqref{eq:adjoint-method-second-order-adjoint} involves transposes and derivatives of the internal forces, and is hence in general \emph{different} from the physical system that evolves with the forward state equation~\eqref{eq:adjoint-method-second-order-state}.

\section{What is computed physically, under what conditions and how: an overview}\label{s:physical-question}
In this article we ask the following question (see Fig.~\ref{fig:general_adjoint_method}):
\begin{mdframed}[style=boxedstyle,nobreak=true]
\noindent Under what conditions, and by which experimental protocol, can the adjoint field be obtained in the \emph{same physical system} that performs the forward computation?
\end{mdframed}

By the same system we mean the same hardware and trainable parameters; we may only change initial conditions and applied sources, and reverse explicit time dependences.

With one exception in the readout, all constructions below have the same form: run the forward dynamics and record $\bfm{u}(t)$; construct adjoint sources (Eq.~\eqref{eq:adjoint-source}) and initial conditions (Eqs.~\eqref{eq:adjoint-method-second-order-adjoint-ic1}, \eqref{eq:adjoint-method-second-order-adjoint-ic2}) from the loss derivatives; obtain the adjoint field $\bfm{a}(t)$ from a second physical experiment; then evaluate the gradient by an overlap between the measured forward/adjoint fields and known parameter sensitivities (Eqs.~\eqref{eq:adjoint-method-gradient-pM}--\eqref{eq:adjoint-method-gradient-v0}). Quantum Equilibrium Propagation (Sec.~\ref{ss:quantum-EP}) obtains the gradient directly from changes in operator expectation values without reading out the adjoint field.
The distinction between linear and nonlinear systems is the central organizing principle of the article. Across all systems we consider, the conditions and algorithms can be summarized as follows:

\begin{mdframed}[style=resultstyle]
\noindent\textbf{Sufficient physical adjoint conditions and constructions.}
\begin{itemize}
    \item \textbf{Linear systems:} Reciprocity is sufficient; Hermiticity is not required, so loss and gain are permitted. The adjoint field is obtained by a direct finite-amplitude run of the same physical system with adjoint sources and initial conditions and any explicit time dependence played in reverse.
    \item \textbf{Nonlinear trajectory systems:} The forward evolution must be recoverable by a time-reversal operation solely on the final state and the linearized dynamics around the reversed trajectory must be reciprocal. The adjoint field is then obtained as the infinitesimal response of a nudged time-reversed trajectory. Explicit time dependences are played in reverse.
    \item \textbf{Stationary linear problems:}
    In linear reciprocal fixed-point or Helmholtz-type systems the adjoint field can be obtained directly with a single finite-amplitude experiment.
    \item \textbf{Stationary nonlinear problems:}
    In nonlinear stationary systems, the exact adjoint is obtained in the infinitesimal limit from the difference between a free solution and a nudged solution, provided the system relaxes to the steady state and the linearized operator there is reciprocal. This is Equilibrium Propagation~\cite{scellier2017equilibrium}.
\end{itemize}
\noindent
Complex-valued systems require separate conditions on the two Wirtinger tangent operators.
Sections~\ref{s:physical-adjoints}--\ref{s:stationary} develop these basic constructions. Section~\ref{s:symmetry-adjoints} extends them through general intertwiners and time-reversal mirrors; Onsager reciprocity is the mobility-weighted special case discussed in Sec.~\ref{s:onsager-reciprocity}.

\end{mdframed}

Table~\ref{tab:physical-adjoint-map} organizes the cases treated in the remainder of the article along these lines. For each case, it records whether the adjoint field is available on the device and which experiment or proposal in the literature realizes it.

%
%

\definecolor{yescol}{RGB}{20,110,60}
\definecolor{nocol}{RGB}{160,35,35}
\definecolor{newcol}{RGB}{190,120,20}
\definecolor{expcol}{RGB}{35,80,155}
\newcommand{\adjyes}{\textcolor{yescol}{\textbf{yes}}}
\newcommand{\adjno}{\textcolor{nocol}{\textbf{no}}}
\newcommand{\adjnew}{\textcolor{newcol}{$\bigstar$}}
\newcommand{\adjexp}{\textcolor{expcol}{$\blacklozenge$}}
\newcommand{\adjbranch}[1]{%
  \multicolumn{6}{@{}>{\RaggedRight\arraybackslash}p{0.95\adjtw}@{}}{\hspace{1em}\textbf{#1}}\\[1pt]}
\newcommand{\adjcase}[1]{\leavevmode\hangindent=1em\hangafter=0\hspace{1em}#1}


\newlength{\adjtw}   
\newsavebox{\adjhalfbox}
\newcommand{\adjhalf}[2]{%
  \sbox{\adjhalfbox}{#2}%
  \mbox{%
    \fbox{\parbox[c][\dimexpr\ht\adjhalfbox+\dp\adjhalfbox-2\fboxsep-2\fboxrule\relax][c]{1.5em}{%
      \centering\rotatebox[origin=c]{90}{\footnotesize\bfseries #1}}}%
    \hspace{6pt}%
    \usebox{\adjhalfbox}%
  }%
}

\makeatletter
\newcommand{\adjsmallcaption}{%
  \renewcommand{\@makecaption}[2]{%
    \vskip\abovecaptionskip
    \footnotesize
    \sbox\@tempboxa{##1: ##2}%
    \ifdim \wd\@tempboxa >\hsize
      ##1: ##2\par
    \else
      \global\@minipagefalse
      \hb@xt@\hsize{\hfil\box\@tempboxa\hfil}%
    \fi
    \vskip\belowcaptionskip}}
\makeatother

\afterpage{%
\newgeometry{left=0.4in,right=0.4in,top=0.75in,bottom=0.75in}
\begin{landscape}
\begin{table}[p]
\centering
\makebox[\linewidth][c]{%
\begin{minipage}{\dimexpr\linewidth+0.5in\relax}
\centering
\fontsize{8}{8.5}\selectfont
\setlength{\tabcolsep}{4pt}
\setlength{\fboxsep}{3pt}
\setlength{\adjtw}{\dimexpr\linewidth-25pt\relax}
\adjsmallcaption
\renewcommand{\arraystretch}{1.0}

\caption{The cases treated in this article.
The trajectory constructions also cover non-autonomous dynamics with programmable internal forces: replay their schedules in reverse, with the stated conditions holding throughout the run.
``On device'' denotes a construction under the stated conditions, allowing known invertible field transformations.
A \adjexp\ marks an experimental realization; a \adjnew\ marks a case without a known realization.
Prescribed trajectory compensation is excluded; see the \hyperref[par:second-order-compensation]{second-order} and \hyperref[par:first-order-compensation]{first-order} discussions.}
\label{tab:physical-adjoint-map}

\adjhalf{Linear systems}{%
\begin{tabular}{@{}
  >{\RaggedRight\arraybackslash}p{0.160\adjtw}
  >{\RaggedRight\arraybackslash}p{0.255\adjtw}
  >{\RaggedRight\arraybackslash}p{0.225\adjtw}
  >{\centering\arraybackslash}p{2.1em}
  >{\centering\arraybackslash}p{3.4em}
  >{\RaggedRight\arraybackslash}p{0.225\adjtw}@{}}
\toprule
\textbf{System} &
\textbf{Sufficient condition} &
\textbf{Adjoint construction} &
\textbf{Sec.} &
\textbf{On device} &
\textbf{Representative realization} \\
\midrule

\adjbranch{Reciprocal}

\adjcase{Second order, real} &
$\bfm{M}^T=\bfm{M}$ and $\bfm{D}^T=\bfm{D}$, both constant in time, and $\bfm{K}^T=\bfm{K}$; damping is no obstacle &
Forward drive and initial conditions (ICs) swapped for
adjoint source and ICs and applied at full amplitude &
\ref{ss:second-order-linear} &
\adjyes &
\adjexp\ acoustic proof of concept~\cite{hermans2015trainable}; \adjexp\ coupled RLC network~\cite{kwon2026time}, written there in first-order form \\
\addlinespace[1pt]

\adjcase{First order, real} &
$\bfm{K}^T=\bfm{K}$ &
As above &
\ref{ss:first-order-linear} &
\adjyes &
Reciprocal dynamical construction~\cite{hermans2015trainable} \\
\addlinespace[1pt]

\adjcase{Schr\"odinger type} &
$\bfm{K}^T=\bfm{K}$; reciprocity and not Hermiticity, so gain and loss are admissible &
As above; what propagates is the conjugate adjoint $\bfm{c}=\overline{\bfm{a}}$ &
\ref{ss:schrodinger-linear} &
\adjyes &
In situ optical backpropagation~\cite{zhou2020situ}; linear layers of holographic networks~\cite{wagner1987multilayer} \\
\addlinespace[1pt]

\adjcase{Stationary, Helmholtz} &
$\bfm{K}^T=\bfm{K}$ &
Forward drive swapped for the adjoint source; applied at full amplitude &
\ref{ss:stationary-linear} &
\adjyes &
Photonic meshes~\cite{hughes2018training}; \adjexp\ Pai et al.~\cite{pai2023experimentally}. Wave scattering~\cite{wanjura2024fully}; \adjexp\ Guillamon et al.~\cite{guillamon2025situ} \\
\addlinespace[3pt]

\adjbranch{Intertwined reciprocity}

\adjcase{General intertwiner} &
Known, time-independent, invertible $\bfm{S}$ with $\bfm{S}\bfm{K}^T\bfm{S}^{-1}=\bfm{K}$; the same relation for $\bfm{M},\bfm{D}$ when present &
One run with $\bfm{S}$-weighted sources and ICs; recover $\bfm{c}=\overline{\bfm{a}}=\bfm{S}^{-1}\bfm{d}$ &
\ref{s:symmetry-adjoints} &
\adjyes &
Weighted stationary transpose access in resistor networks~\cite{lin2026train} \\
\addlinespace[1pt]

\adjcase{Hardware symmetry} &
$\bfm{K}^T=\bfm{K}$ and $\bfm{S}\bfm{K}\bfm{S}^{-1}=\bfm{K}$ with $\bfm{S}\mathbb{P}\bfm{S}^{-1}=\mathbb{P}^{\text{in}}$ &
The error signal enters where the data does and propagates forward; recover $\bfm{c}=\overline{\bfm{a}}=\bfm{S}^{-1}\bfm{d}$ &
\ref{s:symmetry-adjoints} &
\adjyes &
\adjexp\ Fully forward mode training~\cite{xue2024fully}, in static optics \\
\addlinespace[1pt]

\adjcase{Twisted reciprocity} &
Orthogonal $\bfm{S}$ with $\bfm{S}\bfm{K}\bfm{S}^{-1}=\bfm{K}^T$; e.g.\ a mirror-even Hatano--Nelson chain &
One run with relocated sources; recover $\bfm{c}=\overline{\bfm{a}}=\bfm{S}^{-1}\bfm{d}$ &
\ref{s:symmetry-adjoints} &
\adjyes &
\adjnew\ none \\
\bottomrule
\end{tabular}}

\par\vspace{0pt}

\adjhalf{Nonlinear systems}{%
\begin{tabular}{@{}
  >{\RaggedRight\arraybackslash}p{0.160\adjtw}
  >{\RaggedRight\arraybackslash}p{0.255\adjtw}
  >{\RaggedRight\arraybackslash}p{0.225\adjtw}
  >{\centering\arraybackslash}p{2.1em}
  >{\centering\arraybackslash}p{3.4em}
  >{\RaggedRight\arraybackslash}p{0.225\adjtw}@{}}
\toprule

\adjbranch{Trajectories}

\adjcase{Second order, real, undamped} &
$\bfm{D}=\bfm{0}$, so a time-reversal mirror exists, and $\bfm{M}^T=\bfm{M}$ with $\bfm{F}_{\bfm{u}}^T=\bfm{F}_{\bfm{u}}$ along the reversed trajectory &
Infinitesimal response of a nudged time-reversed trajectory &
\ref{ss:second-order-nonlinear} &
\adjyes &
Hamiltonian echo backpropagation~\cite{lopez2023self}; recurrent Hamiltonian echo learning~\cite{pourcel2025learning} \\
\addlinespace[1pt]

\adjcase{Second order, real, damped} &
None: time reversal turns damping into gain but the adjoint requires damping &
Obstructed &
\ref{ss:second-order-nonlinear} &
\adjno &
Contrastive trajectory rule~\cite{berneman2025equilibrium}, approximate \\
\addlinespace[1pt]

\adjcase{First order, real, overdamped} &
None: reversal requires $\bfm{F}\mapsto-\bfm{F}$, the adjoint $+\bfm{F}_{\bfm{u}}^T$ &
Obstructed; only the steady state survives &
\ref{ss:first-order-nonlinear} &
\adjno &
Rayleighian trajectory rule~\cite{berneman2026training}, approximate \\
\addlinespace[1pt]

\adjcase{Schr\"odinger} &
Conjugation is a TRM; $\bfm{F}_{\bfm{u}}$ is Hermitian and $\bfm{F}_{\overline{\bfm{u}}}$ symmetric along the trajectory &
Infinitesimal response of a nudged conjugation-reversed trajectory gives $\bfm{c}=\overline{\bfm{a}}$ &
\ref{ss:schrodinger-nonlinear} &
\adjyes &
Hamiltonian echo backpropagation in nonlinear waves~\cite{lopez2023self} \\
\addlinespace[3pt]

\adjbranch{Stationary problems}

\adjcase{Real fixed point (Equilibrium Propagation)} &
A stable fixed point, and $\bfm{F}_{\bfm{u}}^T=\bfm{F}_{\bfm{u}}$ there &
Difference between the free and the nudged equilibrium &
\ref{ss:stationary-nonlinear} &
\adjyes &
Equilibrium Propagation~\cite{scellier2017equilibrium}; \adjexp\ Ising-machine implementation~\cite{laydevant2024training} \\
\addlinespace[1pt]

\adjcase{Complex steady state (Complex-valued EP)} &
A stable steady state. For $\bfm{a}$: $\bfm{F}_{\bfm{u}}$ Hermitian and $\bfm{F}_{\overline{\bfm{u}}}$ symmetric. For $\bfm{c}=\overline{\bfm{a}}$: $\bfm{F}_{\bfm{u}}$ symmetric and $\bfm{F}_{\overline{\bfm{u}}}$ Hermitian &
Infinitesimal response; use the adjoint-source nudge for $\bfm{a}$ and its complex conjugate for $\bfm{c}$ &
\ref{ss:stationary-nonlinear} &
\adjyes &
Near-Equilibrium Propagation (NEP) recovers $\bfm{a}$ in the exact limit~\cite{sajnok2025near}; \adjexp\ experiment recovering $\bfm{c}=\overline{\bfm{a}}$ with one nonlinear element~\cite{guillamon2026nonlinear} \\
\addlinespace[1pt]

\adjcase{Quantum ground state (QEP)} &
Hermitian Hamiltonian with a nondegenerate, gapped ground state; tangent reciprocity holds after normalization and phase fixing &
Adjoint source enters as a Hamiltonian perturbation; measure operator expectations in the free and nudged states without recovering the adjoint field &
\ref{ss:quantum-EP} &
\adjyes &
Quantum Equilibrium Propagation~\cite{Wanjura2025QuantumEquilibrium,massar2024quantumthermal,scellier2024quantum} \\
\addlinespace[1pt]

\adjcase{Non-reciprocal tangent} &
Direct reciprocity fails, so the unweighted nudged response generally does not solve the adjoint equation &
Unweighted nudge obstructed &
\ref{ss:stationary-nonlinear} &
\adjno &
Approximate extensions~\cite{scellier2018generalization,laborieux2024improving}; \cite{scurria2026equilibrium} modifies the learning dynamics \\
\addlinespace[3pt]

\adjbranch{Tangent intertwiners}

\adjcase{Trajectories: second order and Schr\"odinger} &
Known, constant invertible maps $\mathcal{Q},\bfm{S}$: $\mathcal{Q}$ generates a TRM, and $\bfm{S}$ matches the adjoint to the pulled-back tangent, including inertial and velocity terms &
Infinitesimal response with transformed sources and ICs; undo $\mathcal{Q}\bfm{S}$. For complex fields the TRM is $\mathcal{Q}\mathcal{K}$ &
\ref{s:nonlinear-intertwiners}, \ref{a:onsager-second-order}, \ref{a:gyroscopic-pt} &
\adjyes &
\adjnew\ no known realization in this generality; related $\mathcal{PT}$-assisted Schr\"odinger protocol with gain--loss compensation~\cite{li2026situ} \\
\addlinespace[1pt]

\adjcase{Stationary states} &
A stable state with invertible full tangent; a known invertible $\bfm{S}$ maps the adjoint tangent onto the physical tangent &
Infinitesimal response to a transformed nudge on the same solution branch; undo $\bfm{S}$ &
\ref{s:onsager-reciprocity} &
\adjyes &
EP derivation with Onsager mobility~\cite{lin2026train}; ordinary EP is the case $\bfm{S}=\bfm{I}$ \\
\bottomrule
\end{tabular}}

\end{minipage}%
}
\end{table}
\end{landscape}%
\restoregeometry
}

\section{Local learning rules and ``knowing the model'' vs. ``knowing the parameters''}\label{s:local-learning-rules}

If the adjoint field can be observed physically and the parameter sensitivity is local, the remaining overlap gives a local learning rule. Consider two coupled Duffing oscillators with internal force
\begin{equation}\label{eq:intro-duffing-force}
    \bfm{F}(\bfm{u},\bfm{p})
    =
    \begin{bmatrix}
        k_1u_1+\beta_1u_1^3+c(u_1-u_2)\\
        k_2u_2+\beta_2u_2^3+c(u_2-u_1)
    \end{bmatrix},
    \qquad
    \bfm{p}=(k_1,k_2,\beta_1,\beta_2,c).
\end{equation}
For a parameter in the first oscillator, for example $\beta_1$, the parameter sensitivity is given by
\begin{equation}\label{eq:intro-duffing-beta-sensitivity}
    \bfm{F}_{\beta_1}(\bfm{u})
    =
    \begin{bmatrix}
        u_1^3\\
        0
    \end{bmatrix}.
\end{equation}
The corresponding gradient contribution is therefore
\begin{equation}\label{eq:intro-duffing-beta-gradient}
    \frac{d\chi}{d\beta_1}
    =
    \int_0^T
    a_1(T-t)\,u_1(t)^3\,dt .
\end{equation}
Only the forward and adjoint trajectories of the first DOF are needed. Similarly, $d\chi/dk_1=\int_0^T a_1(T-t)u_1(t)\,dt$. The coupling parameter requires both oscillators:
\begin{equation}\label{eq:intro-duffing-c-gradient}
    \frac{d\chi}{dc}
    =
    \int_0^T
    \bigl(a_1(T-t)-a_2(T-t)\bigr)
    \bigl(u_1(t)-u_2(t)\bigr)\,dt .
\end{equation}
These examples illustrate what we mean by a \emph{local learning rule} (e.g.~\cite{stern2023learning,dillavou2022demonstration,scellier2017equilibrium}): the gradient with respect to a given parameter depends only on the forward and adjoint dynamics of the degree or degrees of freedom to which this parameter couples. This locality follows because the parameter sensitivity of the internal force $\bfm{F}_{(\cdot)}$, and likewise of the mass, damping, and drive operators, is supported only on those degrees of freedom. The gradient overlap integral therefore reduces to those degrees of freedom.

\paragraph{Not knowing the parameters.} A further simplification arises when the parameters enter the forward operators linearly, as they do in Eq.~\eqref{eq:intro-duffing-force}. The sensitivity $\bfm{F}_{\bfm{p}_F}$ is then independent of $\bfm{p}_F$, so the overlap integrals~\eqref{eq:intro-duffing-beta-gradient} and~\eqref{eq:intro-duffing-c-gradient} are assembled entirely from the measured forward and adjoint trajectories and never reference the parameter values themselves. In an experiment, the gradient is thus evaluated at the device's current parameter values, so slow drift, aging, and unknown calibration offsets are absorbed automatically instead of biasing the gradient estimate.

\paragraph{But knowing the model.} Knowing which observable to record, and which overlap integral belongs to which parameter, is a statement about the structure of $\bfm{F}$ and therefore still requires a model, unlike self-learning schemes that avoid explicit knowledge of the detailed energy function or Hamiltonian~\cite{scellier2022agnostic,dillavou2022demonstration,lopez2023self}.

\section{Second-order real-valued systems}\label{s:physical-adjoints}

Consider a second-order, real-valued and non-autonomous system introduced in Sec.~\ref{s:problem-statement}. The explicit time dependence is assumed to be externally programmable, so that it can be replayed as $t\mapsto T-t$ during the adjoint experiment.
We wish to shape its trajectory such that it minimizes the cost function~\eqref{eq:cost-trajectory-terminal}. The adjoint field required to compute the loss function's gradients in Eqs.~\eqref{eq:adjoint-method-gradient-pM}--\eqref{eq:adjoint-method-gradient-v0} can be obtained on the same system under the following sufficient conditions:

\begin{phystheorem}[Sufficient conditions for physical adjoints for second-order real-valued systems]\label{thm:physical-adjoint}
The adjoint field can be obtained on the same hardware in the following cases:
\begin{itemize}
    \item \textbf{Linear: Reciprocity.} Suppose
    \[
        \bfm{F}[\bfm{u},\bfm{p}_F,t]=\bfm{K}(\bfm{p}_K,t)\bfm{u},
    \]
    and reciprocity holds, i.e.,
    \begin{equation}\label{eq:main-second-order-linear-V}
        \bfm{M}^T=\bfm{M},\qquad
        \bfm{D}^T=\bfm{D},\qquad
        \bfm{K}(\bfm{p}_K,t)^T=\bfm{K}(\bfm{p}_K,t),
        \qquad \forall\,t\in[0,T],
    \end{equation}
    then the adjoint is propagated by the same linear hardware, with the explicit time dependence replayed in reverse and with adjoint source~\eqref{eq:adjoint-source} and initial conditions~\eqref{eq:adjoint-method-second-order-adjoint-ic1}, \eqref{eq:adjoint-method-second-order-adjoint-ic2}.

\item \textbf{Nonlinear: Existence of a Time Reversal Mirror (TRM) and reciprocity of the
linearized dynamics.} Suppose the time-reversed trajectory can be generated on the same
hardware by an operation on the final state only (applying the TRM), together with the
reversal of explicit time dependences. This means initializing with $\bfm{u}(T)$ and
$-\dot{\bfm{u}}(T)$. If, in addition,
\begin{equation}\label{eq:main-second-order-nonlinear-V}
    \bfm{D}=\bfm{0},
    \qquad
    \bfm{M}^T=\bfm{M},
    \qquad
    \bfm{F}_{\bfm{u}}\bigl[\bfm{u}(t),\bfm{p}_F,t\bigr]^T
    =
    \bfm{F}_{\bfm{u}}\bigl[\bfm{u}(t),\bfm{p}_F,t\bigr]
    \qquad \forall\,t\in[0,T],
\end{equation}
then the adjoint is obtained as the infinitesimal response of a nudged time-reversed trajectory.

\end{itemize}
\end{phystheorem}

\paragraph{Non-autonomous dynamics and reciprocity.} In Theorem~\ref{thm:physical-adjoint} we allow the internal force $\bfm{F}$ to be a function of time, i.e. the parameters may now specify a time-dependent schedule rather than being static. The reciprocity condition is instantaneous, i.e. it must hold at every time step, and the explicit time dependence must be replayed in reverse during the adjoint experiment. On the other hand, $\bfm{M}$ and $\bfm{D}$ need to be constant in time because otherwise the adjoint state equation will include time derivatives of those matrices and so the forward and adjoint systems will be fundamentally different.


\subsection{Linear}\label{ss:second-order-linear}
The forward dynamics is
\begin{equation}\label{eq:second-order-linear-forward-simple}
    \bfm{M}\ddot{\bfm{u}}(t)
    +
    \bfm{D}\dot{\bfm{u}}(t)
    +
    \bfm{K}(t)\bfm{u}(t)
    -
    \bfm{f}(t)
    =
    \bfm{0}.
\end{equation}
In the reciprocal case where $\bfm{M}^T=\bfm{M}$, $\bfm{D}^T=\bfm{D}$, and $\bfm{K}(t)^T=\bfm{K}(t)$, the adjoint equation~\eqref{eq:adjoint-method-second-order-adjoint} becomes
\begin{equation}\label{eq:second-order-linear-adjoint-physical-form}
    \bfm{M}\ddot{\bfm{a}}(t)
    +
    \bfm{D}\dot{\bfm{a}}(t)
    +
    \bfm{K}(T-t)\bfm{a}(t)
    +
    \,\mathbb{P}\theta_{\bfm{u}}
    \bigl[\mathbb{P}\bfm{u}(T-t)\bigr] = \bfm{0}.
\end{equation}
Thus the adjoint experiment has exactly the same \emph{damped} second-order form as the forward experiment. This proves the linear part of Theorem~\ref{thm:physical-adjoint}. \qed

\paragraph{Physical adjoint algorithm.} One reuses the same hardware, sets the applied drive to $-\bfm{f}_{\mathrm{adj}}$, replaces the initial displacement and velocity by the cost-derived adjoint initial conditions (no nudging required), and replays any explicit time dependence as $t\mapsto T-t$. In this linear reciprocal regime the same damped hardware propagates the adjoint field (see Fig.~\ref{fig:linear_vs_nonlinear}, a).

\paragraph{Remark.} If the adjoint source or the initial conditions are too strong, they may excite the system into a large-amplitude regime where the linear equations no longer hold. In that case, the adjoint source and initial conditions can simply be rescaled by a constant, which rescales the norm of the gradient by the same constant.

\begin{figure}[t!]
    \centering
    \includegraphics[width=1\textwidth]{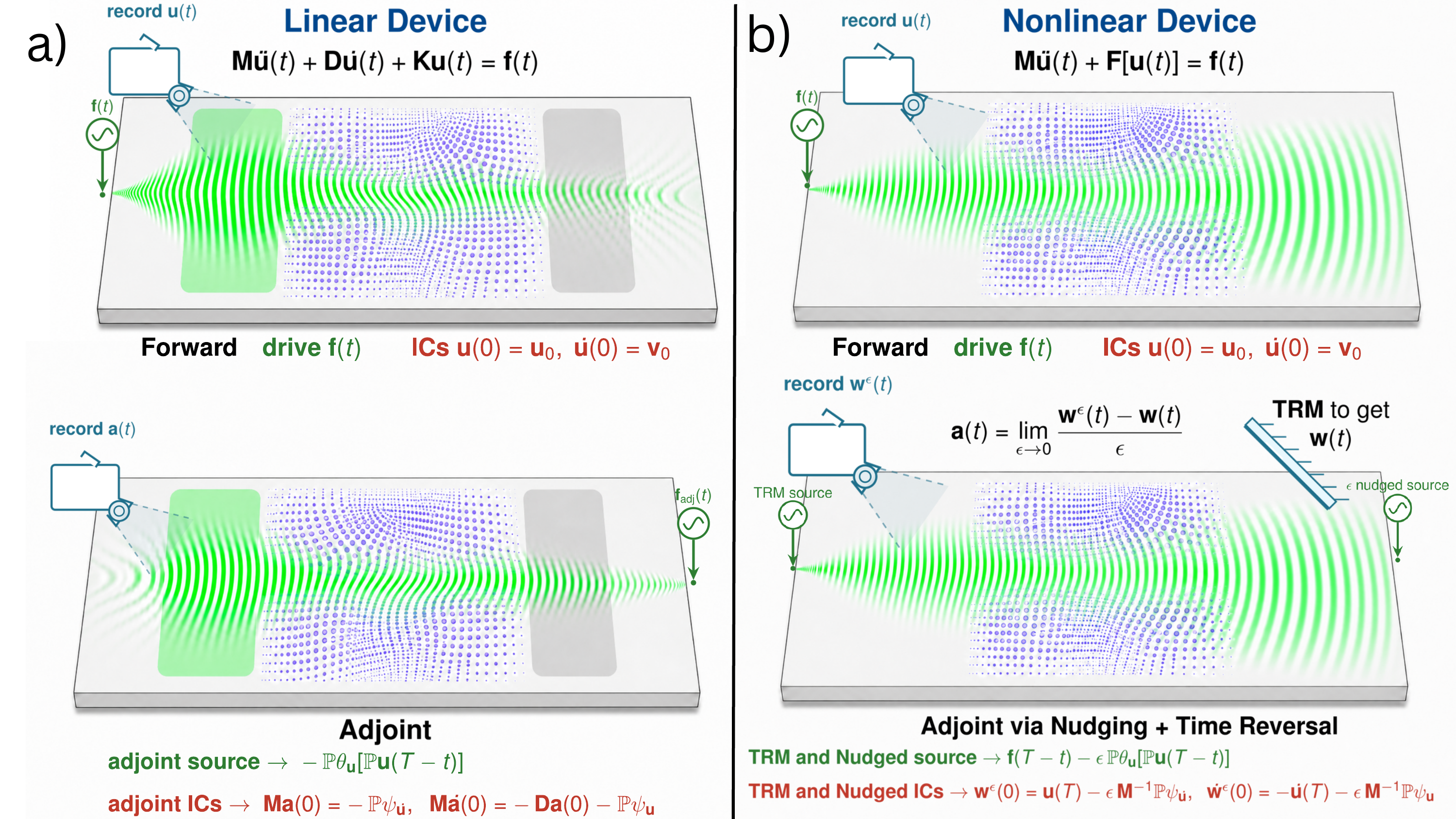}
    \caption{A schematic illustration of the experiments to obtain the adjoint field physically. For the linear devices (a), reciprocity allows the same device to be used, with the drive and ICs adjusted for the forward and adjoint field propagations. Damping and gain are allowed, so long as reciprocity is preserved.
    In the nonlinear case (b), the existence of a time reversal mirror (TRM)~\cite{fink1992time} together with reciprocity of the linearized system provides indirect access to the adjoint field in the same hardware. To obtain the adjoint field physically, simultaneously apply the TRM to the final state of the forward run and perturb the time-reversed trajectory with the adjoint drive and ICs. The adjoint field is then obtained as the (infinitesimal) differential response of the reverse trajectory to the nudge.}
    \label{fig:linear_vs_nonlinear}
\end{figure}

\paragraph{Literature.} In numerical wave-based inverse problems, the same operator governs the forward and adjoint problems when reciprocity holds~\cite{fichtner2010full}. In physical learning, reciprocity was used for backpropagation in early work by Hermans et al.~\cite{hermans2015trainable}. Berneman and Hexner~\cite{berneman2025equilibrium} study a contrastive trajectory rule for damped reciprocal systems that uses nudged trajectories and applies to periodic steady states or resting initial conditions. Our linear reciprocal adjoint construction requires no nudging and accommodates arbitrary initial conditions, including gradients with respect to them. Kwon et al.~\cite{kwon2026time} experimentally realize the time-domain adjoint in a coupled RLC network. They write the circuit in first-order voltage-current form; eliminating the currents gives the damped reciprocal second-order dynamics treated here. In Sec.~\ref{s:schrodinger} and Sec.~\ref{s:stationary} we discuss further references that use reciprocity for same-device gradient calculation in Schr\"odinger (including free-space optics) and stationary (including Helmholtz) problems, respectively.

\subsection{Nonlinear}\label{ss:second-order-nonlinear}
The nonlinear case is more restrictive because the experiment must reproduce the reversed forward trajectory before the infinitesimal nudge can reveal the adjoint. To see the obstruction, set $\bfm{w}(t)=\bfm{u}(T-t)$. Evaluating~\eqref{eq:adjoint-method-second-order-state} at $T-t$ gives
\begin{equation}\label{eq:main-second-order-reversed-state}
    \bfm{M}\ddot{\bfm{w}}(t)
    -
    \bfm{D}\dot{\bfm{w}}(t)
    +
    \bfm{F}\bigl[\bfm{w}(t),\bfm{p}_F,T-t\bigr]
    -
    \bfm{f}(T-t,\bfm{p}_f)
    =
    \bfm{0}, \quad \bfm{w}(0)=\bfm{u}(T),\quad \dot{\bfm{w}}(0)=-\dot{\bfm{u}}(T).
\end{equation}
The damping term has the wrong sign for the same hardware. Reproducing $\bfm{w}$ would therefore require either $\bfm{D}=\bfm{0}$ or an experimental replacement $\bfm{D}\mapsto-\bfm{D}$, i.e. turning damping into gain. Even granting such a gain-flipped run does not fix the adjoint problem. 
We perturb the time-reversed trajectory with the adjoint ICs~\eqref{eq:adjoint-method-second-order-adjoint-ic1}--\eqref{eq:adjoint-method-second-order-adjoint-ic2} and source~\eqref{eq:adjoint-source}:
\begin{align}
    \bfm{w}^\epsilon(0)
    &= \bfm{u}(T) - \epsilon\,\bfm{M}^{-T}\mathbb{P}\psi_{\dot{\bfm{u}}}
    \bigl[\mathbb{P}\bfm{u}(T),\mathbb{P}\dot{\bfm{u}}(T)\bigr],\\
    \dot{\bfm{w}}^\epsilon(0)&=-\dot{\bfm{u}}(T)
    -\epsilon\,\bfm{M}^{-T}\Big(\,\bfm{D}^T\bfm{a}(0)
    +\,\mathbb{P}\psi_{\bfm{u}}
    \bigl[\mathbb{P}\bfm{u}(T),\mathbb{P}\dot{\bfm{u}}(T)\bigr]\Big) \\
    \bfm{f}^\epsilon(t) &= \bfm{f}(T-t,\bfm{p}_f)
    -\epsilon\,\mathbb{P}\theta_{\bfm{u}}[\mathbb{P}\bfm{u}(T-t)],
\end{align}
where $|\epsilon|\ll 1$.
The perturbation $\delta\bfm{w} = \bfm{w}^\epsilon - \bfm{w}$ around the gain-flipped reversed trajectory due to this small nudge obeys (see Appendix~\ref{a:second-order-nonlinear-physical})
\begin{equation}\label{eq:main-second-order-reversed-perturbation-gain}
    \bfm{M}\delta\ddot{\bfm{w}}(t)
    -
    \bfm{D}\delta\dot{\bfm{w}}(t)
    +
    \bfm{F}_{\bfm{u}}\bigl[\bfm{w}(t),\bfm{p}_F,T-t\bigr]\delta\bfm{w}(t)
    +
    \epsilon\,\mathbb{P}\theta_{\bfm{u}}
    \bigl[\mathbb{P}\bfm{u}(T-t)\bigr] + \mathcal{O}(\epsilon^2)
    =
    \bfm{0}, \quad |\epsilon|\ll 1.
\end{equation}
The adjoint equation, however, propagates with $+\bfm{D}^T\dot{\bfm{a}}$. For reciprocal damping, $\bfm{D}^T=\bfm{D}$, the two tangent operators differ by $2\bfm{D}\dot{\bfm{a}}$. The same hardware cannot simultaneously reproduce the reversed nonlinear background trajectory and propagate the correct adjoint perturbation around it, except in the undamped case. 

For mass-proportional damping, if one can flip damping (gain) to gain (damping), then the adjoint field can be obtained physically and the gradient can be calculated by exponential time weighting with the known damping coefficient~\cite{pourcel2025lagrangian} (Appendix~\ref{a:mass-proportional-damping}).
This, however, falls outside our setting, since flipping the sign of the damping changes the hardware between the forward and the adjoint run.

If $\epsilon \rightarrow 0$ and $\bfm{D}=\bfm{0}$, the perturbation dynamics~\eqref{eq:main-second-order-reversed-perturbation-gain} recovers the adjoint state equation~\eqref{eq:adjoint-method-second-order-adjoint} provided that $\bfm{F}_\bfm{u}^T =\bfm{F}_\bfm{u}$ (reciprocity along the forward trajectory). Then
\begin{equation}
    \bfm{a}(t) = \lim_{\epsilon\to0}\frac{\delta\bfm{w}(t)}{\epsilon} = \lim_{\epsilon\to0}\frac{\bfm{w}^{\epsilon}(t)-\bfm{w}(t)}{\epsilon},
\end{equation}
which can therefore be extracted from the difference between the nudged and unnudged reversed trajectories (see Eq.~\eqref{eq:adjoint-extraction-nonlinear}).
This proves the nonlinear part of Theorem~\ref{thm:physical-adjoint}. \qed

\paragraph{Remark: only the tangent condition is used.}
We require $\bfm{F}_{\bfm{u}}^T=\bfm{F}_{\bfm{u}}$ only along the visited trajectory; no potential is needed for the construction. If this symmetry holds on an open neighborhood, the Poincar\'e lemma gives $\bfm{F}=E_{\bfm{u}}$ locally there~\cite{spivak1965calculus}; if it holds throughout a simply connected state space, a single global potential exists. The same distinction applies in the complex case of Sec.~\ref{ss:stationary-nonlinear}: on an open set, condition~\eqref{eq:stationary-nonlinear-cond} is equivalent to closedness of $2\operatorname{Re}\langle\bfm{F},d\bfm{u}\rangle_{\mathbb C}$ and hence locally to~\eqref{eq:stationary-complex-gradient-flow}.

\paragraph{Physical adjoint algorithm.} One first runs the forward dynamics and records $\bfm{u}(t)$; the time-reversed trajectory $\bfm{w}(t)=\bfm{u}(T-t)$ is then available directly from this recording and need not be regenerated physically. The adjoint requires only one further experiment, in which the TRM and the nudge are applied \emph{together}: On the same (undamped) hardware, initialize at the final state with the velocity flipped and both initial conditions perturbed in the directions set by the terminal-loss derivatives,
\begin{equation}
    \bfm{w}^{\epsilon}(0)=\bfm{u}(T)-\epsilon\,\bfm{M}^{-1}\mathbb{P}\psi_{\dot{\bfm{u}}},
    \qquad
    \dot{\bfm{w}}^{\epsilon}(0)=-\dot{\bfm{u}}(T)-\epsilon\,\bfm{M}^{-1}\mathbb{P}\psi_{\bfm{u}},
    \qquad |\epsilon|\ll1,
\end{equation}
drive the system with the reversed external drive superimposed with the adjoint source, $\bfm{f}(T-t)-\epsilon\,\mathbb{P}\theta_{\bfm{u}}\bigl[\mathbb{P}\bfm{u}(T-t)\bigr]$, and evolve it forward while replaying the explicit time dependence of the internal forces in reverse, $\bfm{F}[\,\cdot\,,\bfm{p}_F,t]\mapsto\bfm{F}[\,\cdot\,,\bfm{p}_F,T-t]$. From the recording of $\bfm{w}^{\epsilon}(t)$, the adjoint field is recovered as the infinitesimal differential response
\begin{equation}\label{eq:adjoint-extraction-nonlinear}
    \bfm{a}(t)=\lim_{\epsilon\to0}\frac{\bfm{w}^{\epsilon}(t)-\bfm{w}(t)}{\epsilon},
    \qquad \bfm{w}(t)=\bfm{u}(T-t),
\end{equation}
and the gradient is evaluated digitally from the overlap integrals~\eqref{eq:adjoint-method-gradient-pM}--\eqref{eq:adjoint-method-gradient-v0} (recall $\bfm{D}=\bfm{0}$; see Fig.~\ref{fig:linear_vs_nonlinear}, b).

\paragraph{Damping compensation with external forces.}\label{par:second-order-compensation}
The obstruction above can be bypassed if the damping matrix or gain matrix $\bfm{D}$ is known and one can apply additional forces at each damped DOF during the nudged evolution.
To do so, one applies an additional drive computed from $\bfm{D}$ and the recorded, time-reversed velocity:
\begin{equation}
    \bfm{f}^\epsilon(t) = \bfm{f}(T-t,\bfm{p}_f)
    -\epsilon\,\mathbb{P}\theta_{\bfm{u}}[\mathbb{P}\bfm{u}(T-t)] +2\bfm{D}\dot{\bfm{w}}(t).
\end{equation}
Then the perturbation evolves according to
\begin{equation}
\bfm{M}\delta\ddot{\bfm{w}}
+
\bfm{D}\delta\dot{\bfm{w}}
+
\bfm{F}_{\bfm{u}}[\bfm{w}]\delta\bfm{w}
+
\mathbb{P}\theta_{\bfm{u}}+ \mathcal{O}(\epsilon^2)
=
\bfm{0}.
\end{equation}
Thus, if \(\bfm{M}^T=\bfm{M}\), \(\bfm{D}^T=\bfm{D}\), and
\(\bfm{F}_{\bfm{u}}^T=\bfm{F}_{\bfm{u}}\) along the trajectory, the differential response recovers the forward-time adjoint even in the presence of damping.
The construction therefore remains a same-device adjoint in our sense because it leaves the intrinsic dynamics unchanged. The compensating force $2\bfm{D}\dot{\bfm{w}}(t)$ must be computed from the recorded trajectory and applied to every degree of freedom on which $\bfm{D}$ acts. This requires accurate knowledge of $\bfm{D}$, distributed actuation, and precise trajectory replay. Moreover, the compensation must be prescribed from the reference trajectory; applying it as feedback to the auxiliary trajectory would modify the tangent dynamics. The same damping compensation was used in the concurrent work by Li et al.~\cite{li2026situ} for the nonlinear Schr\"odinger equation in combination with $\mathcal{PT}$-symmetry (see Sec.~\ref{s:nonlinear-intertwiners}).
We exclude such compensation from Table~\ref{tab:physical-adjoint-map}, because it requires calibrated external force control at every damped DOF. It may nevertheless be practical when only a few DOFs are damped and their damping is known.

\paragraph{Symmetric nudging.} At finite $\epsilon$ the extraction~\eqref{eq:adjoint-extraction-nonlinear} is a forward difference, so the measured field is the adjoint up to an $\mathcal{O}(\epsilon)$ bias from the second-order term of the perturbation expansion. Running instead two nudged experiments with opposite signs and forming the central difference,
\begin{equation}\label{eq:adjoint-extraction-central}
    \bfm{a}(t)=\lim_{\epsilon\to0}\frac{\bfm{w}^{+\epsilon}(t)-\bfm{w}^{-\epsilon}(t)}{2\epsilon},
\end{equation}
leaves an $\mathcal{O}(\epsilon^2)$ bias, at the cost of one additional experiment~\cite{laborieux2021scaling, pourcel2025learning}.

\paragraph{Literature.} For time-reversible Hamiltonian dynamics, the construction above implements the echo principle of Hamiltonian echo backpropagation~\cite{lopez2023self} and its recurrent extension by Pourcel and Ernoult~\cite{pourcel2025learning}. The latter was shown to be equivalent to the continuous adjoint method; the self-adaptation of the physical parameters proposed by L\'opez-Pastor and Marquardt~\cite{lopez2023self}, in which the parameters themselves become dynamical, is outside our scope. Our formulation also covers time-dependent parameters of the internal force. A complementary, variational route recasts the learning signal through the classical action: Massar~\cite{massar2025equilibrium} and Pourcel et al.~\cite{pourcel2025lagrangian} generalize Equilibrium Propagation~\cite{scellier2017equilibrium} from energy-based fixed points to full trajectories (Lagrangian Equilibrium Propagation), replacing its two-state energy contrast with a contrast of Lagrangian parameter-derivatives integrated over the trajectory. Because this contrast follows from an integration by parts, it carries boundary terms. Depending on the imposed boundary conditions, these terms either vanish by construction (e.g.\ fixed or periodic endpoints) or persist as generally intractable residuals. For time-reversible systems, Pourcel et al.~\cite{pourcel2025lagrangian} identify a Parametric Final Value Problem boundary choice that renders these residuals tractable and recovers the recurrent Hamiltonian echo rule~\cite{pourcel2025learning} as the physically practical instantiation of this variational construction.

\subsection{Summary: linear and nonlinear systems behave fundamentally differently}
We have shown that linear systems require weaker conditions and allow simpler finite-amplitude physical adjoint field constructions (Fig.~\ref{fig:linear_vs_nonlinear}). The simplicity of physical adjoint experiments for linear systems makes a strong case for parameter encoding (structural nonlinearities)~\cite{momeni2023backpropagation,xia2024nonlinear,yildirim2024nonlinear} and source-position encoding~\cite{richardson2025nonlinear}, which give the linear system nonlinear computational capabilities. Yet physical gradients are readily accessible with finite-amplitude constructions and in the realistic damped setting~\cite{wanjura2024fully}. Damping and gain parameters, whose gradients are available once the adjoint field is obtained (Eq.~\eqref{eq:adjoint-method-gradient-pD}), can also be used to increase the expressivity of the physical device.

\section{First-order real-valued systems}\label{s:first-order}
We next consider first-order, real-valued\footnote{Here, we consider overdamped first-order systems. First-order reformulations of second-order systems that retain both position and velocity are covered in Sec.~\ref{s:physical-adjoints}.} and non-autonomous dynamics
\begin{equation}\label{eq:main-first-order-state}
    \dot{\bfm{u}}(t)+\bfm{F}\bigl[\bfm{u}(t),\bfm{p}_F,t\bigr]-\bfm{f}(t,\bfm{p}_f)=\bfm{0}.
\end{equation}
The corresponding adjoint equation, derived in Appendix~\ref{a:first-order-adjoint}, is
\begin{equation}\label{eq:main-first-order-adjoint}
    \dot{\bfm{a}}(t)
    +
    \bfm{F}_{\bfm{u}}^T
    \bigl[\bfm{u}(T-t),\bfm{p}_F,T-t\bigr]\bfm{a}(t)
    +
    \bfm{f}_{\text{adj}}(t)
    =
    \bfm{0},
\end{equation}
with adjoint source and initial conditions given by
\begin{equation}\label{eq:adj-src-IC-first-order}
\bfm{f}_{\text{adj}}(t) = \mathbb{P}\theta_{\bfm{u}}
    \bigl[\mathbb{P}\bfm{u}(T-t)\bigr], \qquad \bfm{a}(0)
    =
    -\,\mathbb{P}\psi_{\bfm{u}}
    \bigl[\mathbb{P}\bfm{u}(T)\bigr].
\end{equation}
Again, we assume external control of the explicit time dependences of $\bfm{F}$ and $\bfm{f}$, so that they can be replayed as $t\mapsto T-t$ during the adjoint experiment.

\begin{phystheorem}[Sufficient conditions for physical adjoints for first-order real-valued systems]\label{thm:first-order-physical}
The adjoint field can be obtained on the same hardware in the following cases:
\begin{itemize}
    \item \textbf{Linear: Reciprocity.} Suppose
    \[
        \bfm{F}[\bfm{u},\bfm{p}_F,t]=\bfm{K}(\bfm{p}_K,t)\bfm{u}
    \]
    and reciprocity holds, i.e.,
    \begin{equation}\label{eq:main-first-order-linear-V}
        \bfm{K}(\bfm{p}_K,t)^T=\bfm{K}(\bfm{p}_K,t),
        \qquad \forall\,t\in[0,T],
    \end{equation}
    then the adjoint is propagated by the same linear hardware, with the explicit time dependence replayed in reverse and with the adjoint source and initial conditions~\eqref{eq:adj-src-IC-first-order}.

    \item \textbf{Nonlinear, overdamped: trajectories are obstructed without prescribed compensation, the steady state is not.}
\end{itemize}
\end{phystheorem}
 
\subsection{Linear}\label{ss:first-order-linear}
The forward dynamics is
\begin{equation}\label{eq:first-order-linear-forward-simple}
    \dot{\bfm{u}}(t)+\bfm{K}(t)\bfm{u}(t)-\bfm{f}(t)=\bfm{0}.
\end{equation}
If the operator is reciprocal, $\bfm{K}(t)^T=\bfm{K}(t)$, the adjoint equation~\eqref{eq:main-first-order-adjoint} becomes
\begin{equation}\label{eq:first-order-linear-adjoint-physical-form}
    \dot{\bfm{a}}(t)
    +
    \bfm{K}(T-t)\bfm{a}(t)
    +
    \,\mathbb{P}\theta_{\bfm{u}}
    \bigl[\mathbb{P}\bfm{u}(T-t)\bigr]
    =
    \bfm{0}.
\end{equation}
Thus the adjoint experiment has exactly the same overdamped first-order form as the forward experiment and can be performed on the same hardware. \qed

\paragraph{Physical adjoint algorithm.} One reuses the same hardware, sets the applied drive to $-\bfm{f}_{\mathrm{adj}}$, replaces the initial condition by the loss-derived adjoint initial condition, and replays any explicit time dependence as $t\mapsto T-t$. 

\paragraph{Literature.}
Hermans et al.~\cite{hermans2015trainable} formulated physical backpropagation using reciprocal impulse responses, explicitly including linear heat and diffusion equations. Here we recover this construction in the first-order adjoint formulation and extend it to explicitly time-dependent coefficients. Berneman and Hexner~\cite{berneman2026training} proposed a Rayleighian learning rule that is approximate for general nonlinear trajectories.

\subsection{Nonlinear}\label{ss:first-order-nonlinear}

The first-order ODE of this section is the large-damping limit of the second-order ODE (Sec.~\ref{s:physical-adjoints}). The nonlinear adjoint construction is therefore obstructed for the same reason discussed there: time reversal requires flipping damping into gain (here $\bfm{F}\rightarrow -\bfm{F}$) while the adjoint requires damping. Trajectory gradients of nonlinear overdamped dynamics are therefore not accessible through this reversed-drive construction.

\paragraph{Trajectory compensation with external forces.}\label{par:first-order-compensation}
For $\bfm{w}(t)=\bfm{u}(T-t)$, the prescribed drive $\bfm{f}^\epsilon(t)=\bfm{f}(T-t)+2\dot{\bfm{w}}(t)-\epsilon\bfm{f}_{\mathrm{adj}}(t)$ and initial state $\bfm{w}^\epsilon(0)=\bfm{u}(T)+\epsilon\bfm{a}(0)$ yield the adjoint as an infinitesimal response when $\bfm{F}_{\bfm{u}}^T=\bfm{F}_{\bfm{u}}$ along the trajectory. This requires accurate velocity recording and external force control at each DOF with nonzero reference velocity, and is excluded from Table~\ref{tab:physical-adjoint-map}.

Without trajectory compensation, the steady state remains accessible. If $\bfm{F}$ and $\bfm{f}$ are time-independent and the dynamics converges to a linearly exponentially stable fixed point, the loss depends on the parameters only through that fixed point and the problem becomes stationary. That case is Equilibrium Propagation, treated in Sec.~\ref{ss:stationary-nonlinear}; its relation to the trajectory adjoint of this section is derived in Appendix~\ref{a:static-adjoint}.

\section{Schr\"odinger-type complex systems}\label{s:schrodinger}
We now consider complex fields governed by
\begin{equation}\label{eq:main-schrodinger-state}
    i\dot{\bfm{u}}(t)
    +
    \bfm{F}\bigl[\bfm{u}(t),\overline{\bfm{u}}(t),\bfm{p}_F,t\bigr]
    -
    \bfm{f}(t,\bfm{p}_f)
    =
    \bfm{0},\qquad
    \bfm{u}(0)=\bfm{u}_0(\bfm{p}_{u_0}).
\end{equation}
Throughout this section derivatives are Wirtinger derivatives~\cite{remmert1991theory,koor2023short}: for $z=x+iy$, $\partial_z=\frac12(\partial_x-i\partial_y)$ and $\partial_{\overline z}=\frac12(\partial_x+i\partial_y)$, and for vector fields we treat $\bfm{u}$ and its complex conjugate $\overline{\bfm{u}}$ as formally independent variables. 
Thus $\bfm{F}_{\bfm{u}}$ and $\bfm{F}_{\overline{\bfm{u}}}$ denote the two Wirtinger Jacobians, and for the real-valued cost functions the corresponding derivatives satisfy $\theta_{\bfm{u}}=\overline{\theta_{\overline{\bfm{u}}}}$ and $\psi_{\bfm{u}}=\overline{\psi_{\overline{\bfm{u}}}}$ and similarly for $\chi$. 
As we will see, the physical adjoint construction delivers the complex conjugate of the adjoint field.
Writing $\bfm{c}(t)=\overline{\bfm{b}(T-t)}$, the forward-time complex conjugate adjoint equation derived in Appendix~\ref{a:schrodinger-wirtinger} is
\begin{equation}\label{eq:main-schrodinger-adjoint}
    i\dot{\bfm{c}}(t)
    +
    \bfm{F}_{\bfm{u}}^T
    \bigl[\bfm{u}(T-t),\overline{\bfm{u}}(T-t),\bfm{p}_F,T-t\bigr]\bfm{c}(t)
    +
    \overline{
    \bfm{F}_{\overline{\bfm{u}}}
    \bigl[\bfm{u}(T-t),\overline{\bfm{u}}(T-t),\bfm{p}_F,T-t\bigr]}^T
    \overline{\bfm{c}}(t)
    +
   \bfm{f}_{\text{adj}}(t)
    =
    \bfm{0},
\end{equation}
where $\overline{(\cdot)}^{\,T}$ denotes the conjugate transpose. The source and initial conditions are
\begin{equation}\label{eq:main-schrodinger-adjoint-ic}
    \bfm{f}_{\text{adj}}(t) = \mathbb{P}\theta_{\bfm{u}}
    \bigl[
    \mathbb{P}\bfm{u}(T-t),
    \mathbb{P}\overline{\bfm{u}}(T-t)
    \bigr],
    \qquad
    \bfm{c}(0)
    =
    i\,\mathbb{P}\psi_{\bfm{u}}
    \bigl[
    \mathbb{P}\bfm{u}(T),
    \mathbb{P}\overline{\bfm{u}}(T)
    \bigr].
\end{equation}
Taking the complex conjugate of Eq.~\eqref{eq:main-schrodinger-adjoint} flips the signs in front of the Wirtinger Jacobians, which is why $\bfm{c}$, rather than $\bfm{b}$, can be obtained physically in the same hardware.

Inner products now include complex conjugation, $\langle \bfm{x},\bfm{y}\rangle_{\mathbb{C}}=\overline{\bfm{x}}^T\bfm{y}$.
The total variation of $\chi$ is therefore $\delta \chi = \langle d\chi/d\overline{\bfm{u}}, \delta \bfm{u} \rangle_{\mathbb{C}}+ \langle d\chi/d\bfm{u}, \delta \overline{\bfm{u}}\rangle_{\mathbb{C}} = 2\operatorname{Re} \langle d\chi/d\overline{\bfm{u}}, \delta \bfm{u} \rangle_{\mathbb{C}}$.
The gradient overlap integrals expressed in terms of $\bfm{c}$ then become (Appendix~\ref{a:schrodinger-wirtinger}):
\begin{align}
\frac{d\chi}{d\bfm{p}_F}
&=
2\operatorname{Re}
\int_0^T
\left\langle
\overline{\mathbf c(T-t)},
\bfm{F}_{\bfm{p}_F}\bigl[\bfm{u}(t),\overline{\bfm{u}}(t),\bfm{p}_F,t\bigr]
\right\rangle_{\mathbb C}
\,dt,\label{eq:gradient-schrodinger-int_force}
\\
\frac{d\chi}{d\bfm{p}_f}
&=
-2\operatorname{Re}
\int_0^T
\left\langle
\overline{\mathbf c(T-t)},
\bfm{f}_{\bfm{p}_f}(t)
\right\rangle_{\mathbb C}
\,dt,\label{eq:gradient-schrodinger-ext_force}
\\
\frac{d\chi}{d\bfm{p}_{u_0}}
    &=
    -\,2\operatorname{Re}\left\langle\overline{i \mathbf c(T)},\bfm{u}_{0,\bfm{p}_{u_0}}\right\rangle_{\mathbb C}.\label{eq:gradient-schrodinger-IC}
\end{align}

The sufficient conditions under which $\bfm{c}$ can be obtained in the same hardware are as follows:

\begin{phystheorem}[Sufficient conditions for physical adjoints for Schr\"odinger-type systems]\label{thm:schrodinger}
The complex conjugate of the adjoint field can be obtained on the same hardware in the following cases:
\begin{itemize}
    \item \textbf{Linear: Reciprocity.} Suppose
    \[
        \bfm{F}[\bfm{u},\overline{\bfm{u}},\bfm{p}_F,t]
        =
        \bfm{K}(\bfm{p}_K,t)\bfm{u}
    \]
    and reciprocity holds, i.e.,
    \[
        \bfm{K}(\bfm{p}_K,t)^T
        =
        \bfm{K}(\bfm{p}_K,t),
        \qquad \forall\,t\in[0,T],
    \]
    then the complex conjugate of the adjoint is propagated by the same linear hardware, with the explicit time dependence replayed in reverse, the applied drive $-\bfm{f}_{\mathrm{adj}}$, and the initial condition in Eq.~\eqref{eq:main-schrodinger-adjoint-ic}.
\item \textbf{Nonlinear: Conjugation TRM and reciprocity of the linearized dynamics.}
Let $\mathcal{K}$ denote complex conjugation. Suppose $\bfm{F}$ is continuously
Wirtinger differentiable on a conjugation-invariant operating domain and obeys
\begin{equation}\label{eq:main-schrodinger-tr}
    \bfm{F}\bigl[
        \overline{\bfm{u}},\bfm{u},\bfm{p}_F,t
    \bigr]
    =
    \overline{
        \bfm{F}\bigl[\bfm{u},\overline{\bfm{u}},\bfm{p}_F,t\bigr]
    },
    \qquad \forall\,t\in[0,T].
\end{equation}
Then $\mathcal{K}$ is a TRM: the reversed background
$\bfm{w}(t)=\overline{\bfm{u}(T-t)}$ is generated on the same hardware by
conjugating the final state, reversing the explicit time dependences, and applying
the drive $\overline{\bfm{f}(T-t)}$. Suppose, independently, that along the
forward trajectory
\begin{align}
    \bfm{F}_{\bfm{u}}^T
    &=
    \overline{\bfm{F}_{\bfm{u}}},
    \label{eq:main-schrodinger-Fu-reciprocity}\\
    \overline{\bfm{F}_{\overline{\bfm{u}}}}^{\,T}
    &=
    \overline{\bfm{F}_{\overline{\bfm{u}}}},
    \qquad \forall\,t\in[0,T],
    \label{eq:main-schrodinger-Fubar-reciprocity}
\end{align}
where both Jacobians are evaluated at
$[\bfm{u}(t),\overline{\bfm{u}(t)},\bfm{p}_F,t]$. Then $\bfm{c}$ is the
infinitesimal response of a nudged conjugation-reversed trajectory.

\end{itemize}
\end{phystheorem}

\subsection{Linear}\label{ss:schrodinger-linear}
The linear case requires symmetry in the reciprocal sense, not Hermiticity: loss and gain are compatible with the physical adjoint as long as reciprocity is preserved.
For a complex-linear system, the forward dynamics is
\begin{equation}
    i\dot{\bfm{u}}(t)
    +
    \bfm{K}(\bfm{p}_K, t)\bfm{u}(t)
    -
    \bfm{f}(t,\bfm{p}_f)
    =
    \bfm{0},
\end{equation}
where $\bfm{F}_{\overline{\bfm{u}}}$ = 0 by complex linearity and $\bfm{F}_{\bfm{u}} = \bfm{K}(\bfm{p}_K, t)$ is independent of the forward state. Then Eq.~\eqref{eq:main-schrodinger-adjoint} reads
\begin{equation}
    i\dot{\bfm{c}}(t) + \bfm{K}(\bfm{p}_K, T-t)^T \bfm{c}(t) +  \mathbb{P}\theta_{\bfm{u}}
    \bigl[
    \mathbb{P}\bfm{u}(T-t),
    \mathbb{P}\overline{\bfm{u}}(T-t)
    \bigr] = \bfm{0}.
\end{equation}
Hence, if $\bfm{K}^T = \bfm{K}$ for all $t$, one can propagate both fields in the same device by setting the applied drive to $-\bfm{f}_{\mathrm{adj}}$, imposing the adjoint initial condition, and replaying the explicit time dependence in reverse. \qed

\paragraph{Free-space diffractive optical systems.}
In Appendix~\ref{app:free-space-optics} we show that replacing time $t$ by the propagation coordinate $z$ recovers free-space diffractive optical systems.

\paragraph{Holographic vs. intensity measurements.} The physical algorithm above requires measuring both $\bfm{u}(t)$ and $\bfm{c}(t)$ holographically (phase and amplitude), so that the gradient overlap integrals~\eqref{eq:gradient-schrodinger-int_force}--\eqref{eq:gradient-schrodinger-IC} can be evaluated. Extracting the gradient from intensity measurements alone would be experimentally simpler. In~\cite{pai2023experimentally,hughes2018training} this was achieved for Helmholtz systems by exploiting time reversal symmetry in addition to reciprocity. The same construction applies to linear Schr\"odinger problems:

For two complex numbers $x$ and $y$, $2\operatorname{Re}(\overline{x}y) = |x+y|^2-|x|^2-|y|^2$, so $\operatorname{Re}(\overline{x}y)$ follows from three intensity measurements. Eqs.~\eqref{eq:gradient-schrodinger-int_force}--\eqref{eq:gradient-schrodinger-IC} are instead of the form $\operatorname{Re}(xy)$, with $x$ not conjugated, which intensities do not give; one needs $\bfm{b}$ physically in the system rather than $\bfm{c}(t) = \overline{\bfm{b}(T-t)}$.
This can be achieved if the system is time reversal symmetric and reciprocal:
(i) evolve the forward state $\bfm{u}$ from $0$ to $T$ and measure the intensity $I_\bfm{u}(t) = |\bfm{u}(t)|^2$  and (ii) physically evolve $\bfm{c}$ from $0$ to $T$ and measure the intensity $I_\bfm{c}(t) = |\bfm{c}(t)|^2$. Then, (iii), apply the TRM of complex conjugating $\bfm{c}(T)\rightarrow \overline{\bfm{c}(T)}$, time reverse the adjoint source and add the forward sources and initial conditions so that the system now evolves the forward state plus the actual adjoint state $\bfm{u}(t)+\overline{\bfm{c}(T-t)}=\bfm{u}(t)+\bfm{b}(t)$; and record $I_{\bfm{u}+\bfm{b}}(t)=|\bfm{u}(t)+\bfm{b}(t)|^2$. The instantaneous overlap can then be constructed as $I_{\bfm{u}+\bfm{b}}(t)-I_\bfm{u}(t)-I_\bfm{c}(T-t)$. The construction requires time reversal symmetry, so that complex conjugating $\bfm{c}(T)$ regenerates the adjoint state evolution, and linearity, since step (iii) uses the superposition principle.

This gives the site-wise overlap $2\operatorname{Re}[\overline{b_i}(t)u_i(t)]$, and time reversal symmetry makes $\bfm{K}$, and with it $\bfm{K}_{\bfm{p}}$, real. Eq.~\eqref{eq:gradient-schrodinger-int_force} therefore follows when $\bfm{K}_{\bfm{p}}$ is diagonal in the measurement basis, as for an on-site potential, but not for a coupling coefficient, which needs the cross-site overlaps $2\operatorname{Re}[\overline{b_i}u_j]$. The source and initial-condition gradients~\eqref{eq:gradient-schrodinger-ext_force}--\eqref{eq:gradient-schrodinger-IC} pair $\bfm{b}$ with a known vector rather than with $\bfm{u}$ and still require holographic measurements.

\paragraph{Literature.} For phase-only spatial light modulators (SLMs), the framework reduces to the in situ optical backpropagation rule theorized by Zhou et al.~\cite{zhou2020situ}. 
Wagner and Psaltis proposed the same layerwise adjoint structure in a holographic free-space optical network: reciprocal volume holograms propagated the transpose of each linear interconnection, while a weak counterpropagating probe through nonlinear Fabry-P\'erot neurons was designed to approximate multiplication by the local activation derivative~\cite{wagner1987multilayer}. This is a direct small-signal implementation of the nonlinear tangent rather than a contrastive nudging construction; because the probe transmission only approximates the required derivative, the resulting physical gradient is approximate through the nonlinear units.
The same reciprocal linear-adjoint construction covers the linear optical layers proposed in Guo et al.~\cite{guo2021backpropagation} and experimentally realized in Spall et al.~\cite{spall2025training}; in both cases, the saturable-absorber response only approximates the nonlinear tangent required by backpropagation, so the construction is exact for the ideal linear layers but approximate through the nonlinearity. Mididoddi et al.~\cite{mididoddi2025threading} obtained the adjoint field experimentally to optimize the input field and suppress output fluctuations in a dynamically varying medium. In the linear case, the adjoint source can be applied at full amplitude and no nudging or difference of trajectories is required.

\subsection{Nonlinear}\label{ss:schrodinger-nonlinear}

Nonlinear Schr\"odinger systems require two independent ingredients. A TRM must regenerate the reversed background, and the tangent dynamics about that background must match the adjoint dynamics. Here the TRM is ordinary spinless time reversal, implemented by complex conjugation.

Let
\begin{equation}\label{eq:main-schrodinger-tr-background}
    \bfm{w}(t):=\overline{\bfm{u}(T-t)}.
\end{equation}
Under condition~\eqref{eq:main-schrodinger-tr}, this field obeys
\begin{equation}
\begin{split}
    i\dot{\bfm{w}}(t)
    &+
    \bfm{F}\bigl[
        \bfm{w}(t),\overline{\bfm{w}(t)},\bfm{p}_F,T-t
    \bigr]
    -
    \overline{\bfm{f}(T-t,\bfm{p}_f)}
    =
    \bfm{0},\\
    \bfm{w}(0)&=\overline{\bfm{u}(T)}.
\end{split}
\end{equation}
Thus conjugating the final state and the drive, while reversing every explicit time dependence, regenerates the reversed background on the same hardware.

Let
\[
    \bfm{q}(t):=\bfm{f}_{\mathrm{adj}}(t),
    \qquad
    \bfm{q}_T:=
    \mathbb{P}\psi_{\bfm{u}}
    \bigl[
        \mathbb{P}\bfm{u}(T),
        \mathbb{P}\overline{\bfm{u}(T)}
    \bigr].
\]
Run a nudged copy with
\begin{equation}\label{eq:main-schrodinger-conjugation-nudge}
\begin{split}
    i\dot{\bfm{w}}^\epsilon(t)
    &+
    \bfm{F}\bigl[
        \bfm{w}^\epsilon(t),\overline{\bfm{w}^\epsilon(t)},\bfm{p}_F,T-t
    \bigr]
    -
    \left[
        \overline{\bfm{f}(T-t,\bfm{p}_f)}
        -\epsilon\bfm{q}(t)
    \right]
    =
    \bfm{0},\\
    \bfm{w}^\epsilon(0)
    &=\bfm{w}(0)+\epsilon i\bfm{q}_T.
\end{split}
\end{equation}
The physical nudge is $-\epsilon\bfm{q}$ because the external drive enters the residual with a minus sign. Differentiating condition~\eqref{eq:main-schrodinger-tr} gives the tangent operators about $\bfm{w}$ in terms of those along the forward trajectory. Conditions~\eqref{eq:main-schrodinger-Fu-reciprocity} and~\eqref{eq:main-schrodinger-Fubar-reciprocity} then make the differential response obey Eq.~\eqref{eq:main-schrodinger-adjoint}. Therefore
\begin{equation}\label{eq:main-schrodinger-conjugation-recovery}
    \bfm{c}(t)
    =
    \lim_{\epsilon\to0}
    \frac{\bfm{w}^\epsilon(t)-\bfm{w}(t)}{\epsilon},
    \qquad
    \bfm{b}(t)=\overline{\bfm{c}(T-t)}.
\end{equation}
The construction is exact in the infinitesimal-nudge limit. For a sufficiently smooth response, a one-sided difference has an $\mathcal{O}(\epsilon)$ bias, while two opposite nudges reduce it to $\mathcal{O}(\epsilon^2)$. Appendix~\ref{a:nonlinear-intertwined-proof} gives the calculation in the more general setting used in Sec.~\ref{s:symmetry-adjoints}. \qed

\paragraph{Example: local Kerr nonlinearity.}
Consider
\begin{equation}\label{eq:schrodinger-kerr-force}
    \bfm{F}[\bfm{u},\overline{\bfm{u}}]
    =
    \bfm{K}\bfm{u}
    +
    g\,|\bfm{u}|^2\odot\bfm{u},
    \qquad
    g\in\mathbb{R}.
\end{equation}
Conjugation is a TRM when $\bfm{K}=\overline{\bfm{K}}$. The Wirtinger Jacobians are
\begin{equation}
    \bfm{F}_{\bfm{u}}
    =
    \bfm{K}
    +
    2g\,\operatorname{diag}(|\bfm{u}|^2),
    \qquad
    \bfm{F}_{\overline{\bfm{u}}}
    =
    g\,\operatorname{diag}(\bfm{u}^2).
\end{equation}
The local terms satisfy the tangent conditions automatically. The linear term also requires $\bfm{K}^T=\bfm{K}$. Thus ordinary conjugation time reversal gives the physical adjoint for a real symmetric $\bfm{K}$ and a real local Kerr coefficient. Generalized mirrors and tangent intertwiners are treated in Sec.~\ref{s:symmetry-adjoints}.

\paragraph{Literature.}
For autonomous Hamiltonian systems with a terminal-value cost, the nudged time-reversed echo dynamics of L\'opez-Pastor and Marquardt~\cite{lopez2023self} transports the output perturbation backward and yields the continuous-adjoint gradient. That work also constructs a self-adapting system in which the trainable parameters are physical degrees of freedom updated by the same dynamics, without external gradient computation.

\section{Stationary problems and (quantum) Equilibrium Propagation}\label{s:stationary}

In the preceding sections, the computation was carried out along a trajectory. We now turn to stationary problems, in which the result is read from a steady state, such as a fixed point of relaxational dynamics or the steady state of a harmonically driven system.
We allow complex amplitudes $\bfm{u}\in\mathbb{C}^N$. We consider the stationary problem
\begin{equation}\label{eq:stationary-state}
    \bfm{F}\bigl[\bfm{u}^*,\overline{\bfm{u}^*},\bfm{p}_F\bigr]-\bfm{f}(\bfm{p}_f)=\bfm{0},
\end{equation}
where $\bfm{u}^*$ denotes the steady state, and the cost function is defined on it as
\begin{equation}\label{eq:stationary-cost}
    \chi=\psi\bigl[\mathbb{P}\bfm{u}^*,\mathbb{P}\overline{\bfm{u}^*}\bigr].
\end{equation}

The adjoint field $\bfm{a}\in\mathbb{C}^N$ associated with Eq.~\eqref{eq:stationary-state} satisfies the following equation (see Appendix~\ref{app:helmholtz-scattering}), with both Wirtinger Jacobians evaluated at the steady state:
\begin{equation}\label{eq:stationary-adjoint}
    \overline{\bfm{F}_{\bfm{u}}}^{\,T}\bfm{a}
    +
    \bfm{F}_{\overline{\bfm{u}}}^{T}\overline{\bfm{a}}
    +
    \bfm{f}_{\text{adj}}
    =
    \bfm{0},
\end{equation}
where the adjoint source is given by
\begin{equation}\label{eq:stationary-adj-src}
    \bfm{f}_{\text{adj}} = \mathbb{P}\,\psi_{\overline{\bfm{u}}}\bigl[\mathbb{P}\bfm{u}^*,\mathbb{P}\overline{\bfm{u}^*}\bigr].
\end{equation}
Once $\bfm{a}$ is known, the gradients are a single overlap of the two static fields,
\begin{align}
    \frac{d\chi}{d\bfm{p}_F}
    &=2\operatorname{Re}\left\langle\bfm{a},\bfm{F}_{\bfm{p}_F}\bigl[\bfm{u}^*,\overline{\bfm{u}^*},\bfm{p}_F\bigr]\right\rangle_{\mathbb C},
    \label{eq:stationary-gradient-F}\\
    \frac{d\chi}{d\bfm{p}_f}
    &=-2\operatorname{Re}\left\langle\bfm{a},\bfm{f}_{\bfm{p}_f}(\bfm{p}_f)\right\rangle_{\mathbb C}.
    \label{eq:stationary-gradient-f}
\end{align}

The sufficient conditions mirror those of the trajectory case. The nonlinear construction nudges the steady state instead of a time-reversed trajectory:
\begin{phystheorem}[Sufficient conditions for physical adjoints in stationary problems]\label{thm:stationary}
The adjoint field, or its complex conjugate, can be obtained on the same hardware in the following cases:
\begin{itemize}
    \item \textbf{Linear: Reciprocity.} Suppose
    \[
        \bfm{F}\bigl[\bfm{u},\overline{\bfm{u}},\bfm{p}_F\bigr]=\bfm{K}(\bfm{p}_K)\bfm{u},
    \]
    and reciprocity holds, i.e.,
    \[
        \bfm{K}(\bfm{p}_\bfm{K})^T=\bfm{K}(\bfm{p}_\bfm{K}),
    \]
    then the complex conjugate of the adjoint field, $\bfm{c}=\overline{\bfm{a}}$, is generated by solving the same stationary problem with the forward source $\bfm{f}$ replaced by the adjoint source $-\mathbb{P}\psi_{\bfm{u}}$, in a single finite-amplitude experiment and with no nudging.

\item \textbf{Nonlinear: Reciprocity of the tangent operators at the steady state.} Suppose the steady state $\bfm{u}^*$ is a linearly exponentially stable fixed point of a relaxation realized by the same hardware, and that the full real-linear tangent operator of the stationary equation at $\bfm{u}^*$ is invertible. If the Wirtinger tangent operators there satisfy
\begin{equation}\label{eq:stationary-nonlinear-cond}
    \overline{\bfm{F}_{\bfm{u}}\bigl[\bfm{u}^*,\bfm{p}_F\bigr]}^{\,T}
    =
    \bfm{F}_{\bfm{u}}\bigl[\bfm{u}^*,\bfm{p}_F\bigr],
    \qquad
    \bfm{F}_{\overline{\bfm{u}}}\bigl[\bfm{u}^*,\bfm{p}_F\bigr]^T
    =
    \bfm{F}_{\overline{\bfm{u}}}\bigl[\bfm{u}^*,\bfm{p}_F\bigr],
\end{equation}
then the adjoint field $\bfm{a}$ is the infinitesimal response of the fixed point under an adjoint source nudge.
If instead 
\begin{equation}\label{eq:stationary-nonlinear-cond-c}
    \bfm{F}_{\bfm{u}}\bigl[\bfm{u}^*,\bfm{p}_F\bigr]^{\,T}
    =
    \bfm{F}_{\bfm{u}}\bigl[\bfm{u}^*,\bfm{p}_F\bigr],
    \qquad
    \overline{\bfm{F}_{\overline{\bfm{u}}}\bigl[\bfm{u}^*,\bfm{p}_F\bigr]}^{\,T}
    =
    \bfm{F}_{\overline{\bfm{u}}}\bigl[\bfm{u}^*,\bfm{p}_F\bigr],
\end{equation}
then $\bfm{c}=\overline{\bfm{a}}$ can be obtained on the same device using the complex conjugate adjoint-source nudge.
\end{itemize}
\end{phystheorem}

\subsection{Linear}\label{ss:stationary-linear}

For a linear stationary problem
\begin{equation}\label{eq:helmholtz-state}
    \bfm{K}(\bfm{p}_\bfm{K})\,\bfm{u}^* \;=\; \bfm{f}(\bfm{p}_\bfm{f}),
\end{equation}
with $\bfm{K}$ invertible, the internal force $\bfm{F}[\bfm{u},\overline{\bfm{u}},\bfm{p}_F]=\bfm{K}(\bfm{p}_\bfm{K})\bfm{u}$ is complex-linear, so its Wirtinger Jacobians are constant, $\bfm{F}_{\bfm{u}}=\bfm{K}(\bfm{p}_\bfm{K})$ and $\bfm{F}_{\overline{\bfm{u}}}=\bfm{0}$.
An example is a harmonically driven system at frequency $\omega$, for which $\bfm{K}=-\omega^2\bfm{M}+i\omega\bfm{D}+\bfm{K}_{\text{stat}}$.
The adjoint problem then becomes the Hermitian-adjoint system
\begin{equation}\label{eq:main-helmholtz-adj-state}
    \overline{\bfm{K}(\bfm{p}_\bfm{K})^{T}}\,\bfm{a}
    \;=\;
    -\,\mathbb{P}\,\psi_{\overline{\bfm{u}}}\bigl[\mathbb{P}\bfm{u}^*,\mathbb{P}\overline{\bfm{u}^*}\bigr],
\end{equation}
whose conjugate reads
\begin{equation}\label{eq:main-helmholtz-adj-state-cc}
    \bfm{K}(\bfm{p}_\bfm{K})^{T}\,\bfm{c}
    \;=\;
    -\,\mathbb{P}\,\psi_{\bfm{u}}\bigl[\mathbb{P}\bfm{u}^*,\mathbb{P}\overline{\bfm{u}^*}\bigr], \qquad \bfm{c} = \overline{\bfm{a}}.
\end{equation}
This is the same system if $\bfm{K}(\bfm{p}_\bfm{K})^{T}=\bfm{K}(\bfm{p}_\bfm{K})$. Hence, the adjoint field can be obtained by swapping the external drive for the adjoint source. \qed

\paragraph{Comparison with EP.}
This is the linear, stationary (and complex-valued) counterpart of Equilibrium Propagation (Sec.~\ref{ss:stationary-nonlinear}). Because the system is linear, the adjoint source is applied at finite amplitude in a single experiment, without nudging. Damping or gain are permissible: reciprocity is required, but Hermiticity is not.

\paragraph{Holographic vs.\ intensity measurements.}
For complex amplitudes, holographic measurements are required. Two distinct cases allow the gradients to be reconstructed from intensity measurements.

(i) If $\bfm{K}$ is real-valued (and hence reciprocity implies Hermiticity), injecting $-\mathbb{P}\,\psi_{\overline{\bfm{u}}}$ in place of $-\mathbb{P}\,\psi_{\bfm{u}}$ propagates $\bfm{a}$ itself rather than $\bfm{c}=\overline{\bfm{a}}$. Three intensity distributions, $I_{\bfm{u}^*}$, $I_{\bfm{a}}$ and $I_{\bfm{u}^*+\bfm{a}}$, the last with the applied drive $\bfm{f}-\mathbb{P}\psi_{\overline{\bfm{u}}}$, then deliver the site-wise overlaps $2\operatorname{Re}[\overline{a_i}\,u^*_i]=I_{\bfm{u}^*+\bfm{a},i}-I_{\bfm{u}^*,i}-I_{\bfm{a},i}$. These evaluate~\eqref{eq:stationary-gradient-F} whenever $\bfm{K}_{\bfm{p}}$ is real and diagonal in the measurement basis, as for a local permittivity; a non-local $\bfm{K}_{\bfm{p}}$ also needs the cross-site overlaps $2\operatorname{Re}[\overline{a_i}\,u^*_j]$, which intensity distributions alone do not provide.

(ii) If $\bfm{K}$ is complex-valued and propagation is reciprocal and lossless throughout the device, the analogous procedure of Sec.~\ref{s:schrodinger}, paragraph ``Holographic vs.\ intensity measurements,'' applies. It was developed by Hughes et al.~\cite{hughes2018training} and realized by Pai et al.~\cite{pai2023experimentally} in silicon photonic meshes of Mach--Zehnder interferometers and programmable phase shifters. The conjugated adjoint field $\bfm{c}$ is propagated from the outputs to the inputs, measured holographically, and re-encoded as complex-conjugated input sources, so that a further forward pass carries $\bfm{a}$ rather than $\bfm{c}$; superimposing it on the forward field gives $I_{\bfm{u}^*+\bfm{a}}$. Here $\bfm{K}_{\bfm{p}}$ may be complex on the diagonal, as for a phase shifter, which also requires $\operatorname{Im}[\overline{a_i}\,u^*_i]$. Re-encoding with an extra factor $i$ carries $i\bfm{a}$, for which $\operatorname{Re}[\overline{ia_i}\,u^*_i]=\operatorname{Im}[\overline{a_i}\,u^*_i]$ and $I_{i\bfm{a}}=I_{\bfm{a}}$, so a single further distribution $I_{\bfm{u}^*+i\bfm{a}}$ suffices.

\paragraph{Literature.}
Physical backpropagation in Helmholtz systems was theorized for feed-forward photonic neural networks by Hughes et al.~\cite{hughes2018training} and experimentally demonstrated by Pai et al.~\cite{pai2023experimentally}. In these optical--electronic--optical (OEO) architectures, the optical output of each layer is detected, transformed by the electronic nonlinearity, and re-encoded for the next layer. During backpropagation, the field returned by each optical layer is multiplied electronically by the derivative of the intervening nonlinearity and re-encoded as the adjoint input to the preceding layer, allowing the single-layer protocol to be chained across layers (see Appendix~\ref{app:helmholtz-layered-photonic}). Ashtiani et al.~\cite{ashtiani2026integrated} demonstrated a different on-chip backpropagation architecture with separate forward and backward blocks and approximate physical nonlinear derivatives. For general wave-scattering systems, the in situ physical adjoint computing protocol of Guillamon et al.~\cite{guillamon2025situ} is equivalent to the holographic formulation presented here; a mechanical equivalent was developed by Li and Mao~\cite{li2025topological,li2024training}. Closely related identities were derived by Wanjura and Marquardt~\cite{wanjura2024fully} for neuromorphic systems based on linear wave scattering. As shown in Appendix~\ref{app:helmholtz-green}, their scattering-matrix gradient formulas follow from the present adjoint identity.

\subsection{Nonlinear: Equilibrium Propagation}\label{ss:stationary-nonlinear}

Let $\bfm{u}^{*\epsilon}$ be the steady state reached under the drive $\bfm{f}-\epsilon\mathbb{P}\psi_{\overline{\bfm{u}}}$,
\begin{equation}\label{eq:stationary-nudged}
    \bfm{F}\bigl[\bfm{u}^{*\epsilon},\overline{\bfm{u}^{*\epsilon}},\bfm{p}_F\bigr]
    -\bfm{f}(\bfm{p}_f)
    +\epsilon\,\mathbb{P}\psi_{\overline{\bfm{u}}}
    \bigl[\mathbb{P}\bfm{u}^{*\epsilon},\mathbb{P}\overline{\bfm{u}^{*\epsilon}}\bigr]
    =\bfm{0},
    \qquad |\epsilon|\ll1 .
\end{equation}
Subtracting the free steady-state equation and linearizing in $\delta\bfm{u}=\bfm{u}^{*\epsilon}-\bfm{u}^*$ gives
\begin{equation}\label{eq:stationary-tangent}
    \bfm{F}_{\bfm{u}}\,\delta\bfm{u}
    +\bfm{F}_{\overline{\bfm{u}}}\,\overline{\delta\bfm{u}}
    +\epsilon\,\mathbb{P}\psi_{\overline{\bfm{u}}}
    =\mathcal{O}(\epsilon^2).
\end{equation}
Divide by $\epsilon$ and let $\epsilon\to0$. The rescaled response then obeys the adjoint equation~\eqref{eq:stationary-adjoint} precisely when the tangent operators satisfy condition~\eqref{eq:stationary-nonlinear-cond}, and by uniqueness
\begin{equation}\label{eq:stationary-ep}
    \bfm{a}=\lim_{\epsilon\to0}\frac{\bfm{u}^{*\epsilon}-\bfm{u}^*}{\epsilon}.
\end{equation}
This proves that conditions~\eqref{eq:stationary-nonlinear-cond} are sufficient to obtain $\bfm{a}$.
Taking the complex conjugate of the adjoint equation~\eqref{eq:stationary-adjoint} shows that conditions~\eqref{eq:stationary-nonlinear-cond-c} are sufficient to obtain $\bfm{c}$.
This proves the nonlinear part of Theorem~\ref{thm:stationary}.  \qed

\paragraph{Physical adjoint algorithm.} The device is run twice and the two steady states are subtracted. (i) \emph{Free phase}: let the system settle under the forward drive alone, record $\bfm{u}^*$, and evaluate the loss derivative from the measured output. (ii) \emph{Nudged phase}: let it settle again under the drive $\bfm{f}-\epsilon\mathbb{P}\psi_{\overline{\bfm{u}}}$, and record $\bfm{u}^{*\epsilon}$ to extract the adjoint field via Eq.~\eqref{eq:stationary-ep}. The free and nudged phases must settle onto the same branch: $\epsilon$ must be small enough that the nudge neither crosses a basin boundary nor drives the system through a bifurcation.

\paragraph{Complex-valued example.} Condition~\eqref{eq:stationary-nonlinear-cond} holds identically whenever the force is the complex gradient of a real energy,
\begin{equation}\label{eq:stationary-complex-gradient-flow}
    \bfm{F}=E_{\overline{\bfm{u}}},
    \qquad E\in\mathbb{R},
\end{equation}
the complex counterpart of $\bfm{F}=E_{\bfm{u}}$: reality of $E$ makes $\bfm{F}_{\bfm{u}}=E_{\overline{\bfm{u}}\bfm{u}}$ Hermitian, and commuting partial derivatives make $\bfm{F}_{\overline{\bfm{u}}}=E_{\overline{\bfm{u}}\,\overline{\bfm{u}}}$ symmetric. A concrete case is a lattice with a Kerr nonlinearity,
\begin{equation}\label{eq:stationary-kerr-example}
    \bfm{F}=\bfm{K}\bfm{u}+g\,|\bfm{u}|^2\odot\bfm{u},
    \qquad
    E=\langle\bfm{u},\bfm{K}\bfm{u}\rangle_{\mathbb C}+\tfrac{g}{2}\sum_i|u_i|^4,
\end{equation}
with $\bfm{K}$ Hermitian and $g$ real, for which $\bfm{F}_{\bfm{u}}=\bfm{K}+2g\operatorname{diag}(|\bfm{u}|^2)$ and $\bfm{F}_{\overline{\bfm{u}}}=g\operatorname{diag}(\bfm{u}^2)$. This example recovers the exact limit of Near-Equilibrium Propagation (NEP)~\cite{sajnok2025near}.

\paragraph{Comparison with the trajectory case.} A trajectory construction must also regenerate the reversed background. With the conjugation TRM of Sec.~\ref{s:schrodinger}, the Kerr example requires $\bfm{K}$ to be real, while tangent matching requires it to be Hermitian; together these conditions make $\bfm{K}$ real symmetric. With the generalized $\mathcal{PT}$ mirror of Sec.~\ref{s:nonlinear-intertwiners}, the two conditions are $\bfm{K}=\bfm{\Pi}\overline{\bfm{K}}\bfm{\Pi}$ and $\bfm{K}=\bfm{K}^{\dagger}$. The stationary construction needs no reversed background and therefore requires only Hermiticity.

\paragraph{The real-valued case.} For real fields, $\bfm{F}_{\overline{\bfm{u}}}$ is absent and Hermiticity is symmetry, so condition~\eqref{eq:stationary-nonlinear-cond} reduces to $\bfm{F}_{\bfm{u}}^T=\bfm{F}_{\bfm{u}}$ at the steady state. The adjoint and gradient are then real,
\begin{equation}\label{eq:ep-real}
    \bfm{F}_{\bfm{u}}^{T}\bfm{a}+\mathbb{P}\psi_{\bfm{u}}\bigl[\mathbb{P}\bfm{u}^*\bigr]=\bfm{0},
    \qquad
    \frac{d\chi}{d\bfm{p}}
    =
    \bigl\langle\bfm{a},\,\bfm{F}_{\bfm{p}_F}\bigl[\bfm{u}^*,\bfm{p}_F\bigr]-\bfm{f}_{\bfm{p}_f}(\bfm{p}_f)\bigr\rangle,
\end{equation}
with $\psi_{\bfm{u}}$ the ordinary real gradient; the factor $2\operatorname{Re}$ disappears together with the Wirtinger convention. The construction is then Equilibrium Propagation~\cite{scellier2017equilibrium}; Appendix~\ref{a:static-adjoint} derives this static adjoint and relates it to the trajectory adjoint of Sec.~\ref{s:first-order}.

\paragraph{Literature.}
EP was first introduced by Scellier and Bengio~\cite{scellier2017equilibrium}, theoretically applied to nonlinear resistive networks by Kendall et al.~\cite{kendall2020training} and experimentally implemented on an Ising machine by Laydevant et al.~\cite{laydevant2024training}. It was later extended by Laborieux and Zenke~\cite{laborieux2022holomorphic} to holomorphic networks, where exact gradients can be recovered from finite-amplitude complex oscillations; there it is the nudge amplitude that is complexified and the dynamics must be holomorphic, $\bfm{F}_{\overline{\bfm{u}}}=\bfm{0}$.
A broad literature has clarified the relation between EP and recurrent backpropagation~\cite{scellier2019equivalence,ernoult2019updates}, reduced finite-nudge bias via symmetric nudging and scaled EP to deep convolutional networks~\cite{laborieux2021scaling}, introduced frequency propagation as a related frequency-multiplexed learning rule~\cite{anisetti2024frequency}, studied EP in strongly multistable systems~\cite{wang2024training} and developed approximate extensions beyond reciprocal energy-based dynamics~\cite{scellier2018generalization,laborieux2024improving}. In particular, Scurria et al.~\cite{scurria2026equilibrium} generalize to non-reciprocal systems by modifying the learning dynamics. Lin et al.~\cite{lin2026train} propose a generalized Equilibrium Propagation that organizes two-phase rules by the perturbative order of the nudge at the free equilibrium, recovering EP at first order and Coupled Learning at second, and apply it to linear resistor networks.

Sajnok and Matuszewski~\cite{sajnok2025near} derived a complex-valued equilibrium-propagation rule for nonlinear networks and continuous fields. When their tangent conditions hold exactly, they coincide with Eq.~\eqref{eq:stationary-nonlinear-cond}, and their update rule agrees with the direct-adjoint construction presented here. Recovering $\bfm{a}$ in this way requires Hermiticity of $\bfm{F}_{\bfm{u}}$, which linear loss violates. Their NEP formulation treats this violation perturbatively in weakly dissipative systems.

Guillamon et al.~\cite{guillamon2026nonlinear} experimentally implement the conjugated-adjoint construction, $\bfm{c}=\overline{\bfm{a}}$, in a lossy reciprocal network with one localized nonlinear element. A global phase rotation makes its anomalous tangent coefficient real, satisfying conditions~\eqref{eq:stationary-nonlinear-cond-c}. Here the symmetry condition on $\bfm{F}_{\bfm{u}}$ is preserved by reciprocal linear loss, so the construction remains exact in the infinitesimal-nudge limit despite damping. For several local nonlinear elements, the same phase correction requires their nonzero anomalous coefficients to share a phase modulo $\pi$.

Dal Cin et al.~\cite{cin2025training} recover exact gradients in reciprocal linear systems, including dissipative systems, as the nudge vanishes. Their nonlinear optical extension is perturbative in the nonlinear response.

\subsection{Quantum Equilibrium Propagation}\label{ss:quantum-EP}

Quantum Equilibrium Propagation (QEP) trains the parameters of a quantum Hamiltonian through the response of its equilibrium state~\cite{Wanjura2025QuantumEquilibrium,massar2024quantumthermal,scellier2024quantum}. We consider the normalized, nondegenerate ground state $\bfm{u}^*$ of $\bfm{H}(\bfm{p})=\sum_j p_j\bfm{A}_j$, with real parameters $p_j$ and fixed Hermitian operators $\bfm{A}_j$. Inputs, when present, enter through a further parameter block held fixed during each training step. Although the Schr\"odinger equation is linear in the state, the energy expectation depends nonlinearly on the normalized state. QEP is therefore a constrained version of the nonlinear stationary construction in Theorem~\ref{thm:stationary}.

The outputs are expectation values $y_\ell=\langle\bfm{u}^*,\bfm{O}_\ell\bfm{u}^*\rangle_{\mathbb C}$ of Hermitian observables, and the loss is $\chi=\psi(\bfm{y})$. In the free phase, the device prepares the ground state and measures these outputs. 
The derivative, $\psi_{\bfm{y}}=d \psi/d \bfm{y}$ (with components $\psi_{y_l}$), typically referred to as the error signal and evaluated from the free-state readout, gives the adjoint source and its Hamiltonian implementation~\cite{Wanjura2025QuantumEquilibrium}:

\begin{equation}\label{eq:QEP-nudged-Hamiltonian}
    \begin{aligned}
        \bfm{f}_{\mathrm{adj}}
        &=\chi_{\overline{\bfm{u}}}
        =\sum_\ell\psi_{y_\ell}\bfm{O}_\ell\bfm{u}^*,\\
        \bfm{H}_\epsilon
        &=\bfm{H}+\epsilon\sum_\ell\psi_{y_\ell}\bfm{O}_\ell.
    \end{aligned}
\end{equation}
Because the loss depends on operator expectation values, the source is the action of a Hermitian operator on the state. It is therefore implemented as a Hamiltonian perturbation, whereas the preceding sections add an external force. The error signal components are held fixed during nudging. The device must support these Hamiltonian terms and prepare the ground state $\bfm{u}^{*,\epsilon}$ of the nudged Hamiltonian.

Here, state variations must preserve unit norm. We fix the arbitrary global phase so that these variations lie in the tangent space orthogonal to $\bfm{u}^*$. Appendix~\ref{app:QEP-adjoint} shows that, on this space, the Hermiticity of $\bfm{H}$ automatically gives both tangent reciprocity conditions of Theorem~\ref{thm:stationary}. The Hamiltonian nudge supplies the adjoint source projected onto this space, so the infinitesimal state response is the adjoint. A positive excitation gap makes this response unique. The appendix also relates this construction to the Onsager reciprocity used by Wanjura and Marquardt~\cite{Wanjura2025QuantumEquilibrium}.

Writing $\langle\bfm{A}_j\rangle_\epsilon=\langle\bfm{u}^{*,\epsilon},\bfm{A}_j\bfm{u}^{*,\epsilon}\rangle_{\mathbb C}$, their gradient formula is our stationary adjoint overlap:
\begin{equation}\label{eq:QEP-gradient-observable}
    \frac{d\chi}{dp_j}
    =\lim_{\epsilon\to0}\frac{\langle\bfm{A}_j\rangle_\epsilon-\langle\bfm{A}_j\rangle_0}{\epsilon}
    =2\operatorname{Re}\langle\bfm{a},\bfm{A}_j\bfm{u}^*\rangle_{\mathbb C}.
\end{equation}
The proof is given in Appendix~\ref{app:QEP-adjoint}.

QEP obtains this gradient without reading out the adjoint field itself. It measures $\langle\bfm{A}_j\rangle$ in the free and nudged states and divides their difference by $\epsilon$. This works because an expectation value is quadratic in the state: its first-order change pairs the free state with the state response, which is the adjoint. Since $\bfm{A}_j=\partial\bfm{H}/\partial p_j$, that pairing is exactly the overlap needed for the parameter gradient. No phase-resolved field measurement is required. When $\bfm{A}_j$ describes a local coupling, the readout is local in the sense of Sec.~\ref{s:local-learning-rules}; for a two-spin coupling it is the corresponding spin correlator.

At finite $\epsilon$, the one-sided estimate has an $O(\epsilon)$ bias. Opposite nudges reduce this to $O(\epsilon^2)$, while repeated state preparations are needed to estimate the expectation values. The nudge must stay within the linear-response regime and be large enough to resolve the observable differences above measurement noise; Appendix~\ref{app:QEP-finite-nudge} discusses this trade-off.

\paragraph{Literature.} Wanjura and Marquardt~\cite{Wanjura2025QuantumEquilibrium} derive QEP from Onsager reciprocity and demonstrate it numerically for many-body phase recognition, phase-diagram exploration, and sensitivity optimisation. Quantum versions of EP were developed independently by Massar and Mognetti~\cite{massar2024quantumthermal} and Scellier~\cite{scellier2024quantum}.

\section{General intertwiners and twisted reciprocity}\label{s:symmetry-adjoints}

Reciprocity of the tangent operators is a sufficient condition for access to the adjoint on the same hardware. More generally, an intertwiner can map the adjoint operators onto the forward operators. We state the construction for Schr\"odinger-type systems.
The forward-time conjugated Schr\"odinger adjoint $\bfm{c}$ of Eq.~\eqref{eq:main-schrodinger-adjoint} generally evolves under transposed Wirtinger tangent operators. A linear system needs only an intertwiner $\bfm{S}$, whereas a nonlinear system also needs a TRM $\mathcal{Q}\mathcal{K}$ to regenerate the reversed background.

For both constructions below, let $\tau$ denote time in the adjoint experiment and set
\begin{equation}\label{eq:intertwiner-reverse-time}
    r:=T-\tau.
\end{equation}
Write the distributed and terminal loss derivatives as
\begin{equation}\label{eq:intertwiner-loss-data}
    \bfm{q}(\tau)
    :=
    \mathbb{P}\theta_{\bfm{u}}
    \bigl[
        \mathbb{P}\bfm{u}(r),
        \mathbb{P}\overline{\bfm{u}(r)}
    \bigr],
    \qquad
    \bfm{q}_T
    :=
    \mathbb{P}\psi_{\bfm{u}}
    \bigl[
        \mathbb{P}\bfm{u}(T),
        \mathbb{P}\overline{\bfm{u}(T)}
    \bigr].
\end{equation}

\begin{phystheorem}[Sufficient conditions for intertwined physical adjoints in Schr\"odinger-type systems]\label{thm:intertwined-reciprocity}
Consider the Schr\"odinger system~\eqref{eq:main-schrodinger-state}. The adjoint can be obtained on the same hardware in the following cases:

\begin{itemize}
    \item \textbf{Linear.}
    Suppose
    \[
        \bfm{F}\bigl[
            \bfm{u},\overline{\bfm{u}},\bfm{p}_F,t
        \bigr]
        =
        \bfm{K}(\bfm{p}_K,t)\bfm{u}
    \]
    and, at the current parameter value, there is a known invertible matrix $\bfm{S}$ that is fixed over the full time schedule and satisfies
    \begin{equation}\label{eq:generalized-reciprocity-linear}
        \bfm{S}\bfm{K}(\bfm{p}_K,t)^T\bfm{S}^{-1}
        =
        \bfm{K}(\bfm{p}_K,t),
        \qquad \forall\,t\in[0,T].
    \end{equation}

    \item \textbf{Nonlinear.}
    Let $\mathcal{K}$ denote complex conjugation and
    \begin{equation}\label{eq:generalized-trm}
        \Theta_{\mathcal{Q}}:=\mathcal{Q}\mathcal{K},
    \end{equation}
    where $\mathcal{Q}$ is an invertible complex-linear map. Suppose $\bfm{F}$ is continuously Wirtinger differentiable on an operating domain invariant under $\bfm{u}\mapsto\mathcal{Q}\overline{\bfm{u}}$ and obeys
    \begin{equation}\label{eq:Q-background-condition}
        \bfm{F}\bigl[
            \mathcal{Q}\overline{\bfm{u}},
            \overline{\mathcal{Q}}\bfm{u},
            \bfm{p}_F,t
        \bigr]
        =
        \mathcal{Q}\,\overline{
            \bfm{F}\bigl[
                \bfm{u},\overline{\bfm{u}},\bfm{p}_F,t
            \bigr]
        },
        \qquad \forall\,t\in[0,T],
    \end{equation}
    throughout that domain. Suppose, independently, that an invertible matrix $\bfm{S}$ satisfies, along the forward trajectory,
    \begin{align}
        \bfm{S}\bfm{F}_{\bfm{u}}^T\bfm{S}^{-1}
        &=
        \overline{\bfm{F}_{\bfm{u}}},
        \label{eq:S-condition-1}\\
        \bfm{S}\overline{\bfm{F}_{\overline{\bfm{u}}}}^{\,T}
        \overline{\bfm{S}}^{-1}
        &=
        \overline{\bfm{F}_{\overline{\bfm{u}}}},
        \qquad \forall\,t\in[0,T],
        \label{eq:S-condition-2}
    \end{align}
    where both Jacobians are evaluated at
    $[\bfm{u}(t),\overline{\bfm{u}(t)},\bfm{p}_F,t]$.
\end{itemize}
\end{phystheorem}

\paragraph{Other equation classes.}
For real linear first-order systems, the same similarity condition gives $\bfm{d}=\bfm{S}\bfm{a}$. Real linear second-order systems also require $\bfm{S}\bfm{M}^T\bfm{S}^{-1}=\bfm{M}$ and $\bfm{S}\bfm{D}^T\bfm{S}^{-1}=\bfm{D}$. For example, an additional gyroscopic term $\bfm{G}\dot{\bfm{u}}$, with $\bfm{G}^T=-\bfm{G}$, is covered by $\bfm{S}=\bfm{\Pi}$ if a constant real mirror $\bfm{\Pi}=\bfm{\Pi}^T=\bfm{\Pi}^{-1}$ preserves the symmetric mass, damping and stiffness operators and satisfies $\bfm{\Pi}\bfm{G}\bfm{\Pi}=-\bfm{G}$. Then $\bfm{\Pi}(\bfm{D}+\bfm{G})^T\bfm{\Pi}=\bfm{D}+\bfm{G}$. Appendix~\ref{a:gyroscopic-pt} treats the nonlinear, undamped case, where a $\mathcal{PT}$-TRM must also regenerate the background. Nonlinear real first-order dynamics have no generic unchanged-device reverse trajectory.\footnote{Here we exclude prescribed trajectory compensation; see the paragraphs on \hyperref[par:second-order-compensation]{damping compensation} in Sec.~\ref{ss:second-order-nonlinear} and \hyperref[par:first-order-compensation]{trajectory compensation} in Sec.~\ref{ss:first-order-nonlinear}.} Stationary systems need no TRM or schedule reversal: the same linear intertwining condition applies, while nonlinear nudging requires a stable operating point and continuation on the same solution branch.

\subsection{Linear Schr\"odinger systems: spatial symmetries and twisted reciprocity}

Using the notation above, the complex conjugate of the linear Schr\"odinger adjoint is
\begin{equation}\label{eq:sym-adjoint}
    i\dot{\bfm{c}}(\tau)
    +
    \bfm{K}(\bfm{p}_K,r)^T\bfm{c}(\tau)
    +
    \bfm{q}(\tau)
    =
    \bfm{0},
    \qquad
    \bfm{c}(0)=i\bfm{q}_T.
\end{equation}
Set $\bfm{d}=\bfm{S}\bfm{c}$ and multiply this equation by $\bfm{S}$. Condition~\eqref{eq:generalized-reciprocity-linear} gives
\begin{equation}\label{eq:intertwined-linear-experiment}
    i\dot{\bfm{d}}(\tau)
    +
    \bfm{K}(\bfm{p}_K,r)\bfm{d}(\tau)
    +
    \bfm{S}\bfm{q}(\tau)
    =
    \bfm{0},
    \qquad
    \bfm{d}(0)=i\bfm{S}\bfm{q}_T.
\end{equation}
This has the same dynamical operator as the forward system and proves the linear part of Theorem~\ref{thm:intertwined-reciprocity}. \qed

\paragraph{Which $\bfm{S}$ qualify, and at what cost?}
At each parameter point, a known $\bfm{S}(\bfm{p}_K)$ must satisfy Eq.~\eqref{eq:generalized-reciprocity-linear} at all times during that experiment. A common $\bfm{S}$ for the full trainable family avoids recalibration. Otherwise $\bfm{S}(\bfm{p}_K)$ must be recalibrated after every update. A signed permutation relocates ports and may add sign flips. A unitary $\bfm{S}$ preserves amplitudes but can require coherent source combinations. A diagonal $\bfm{S}$ preserves source support but rescales its amplitude. A general invertible $\bfm{S}$ requires weighted source synthesis and inverse weighting during recovery; its condition number sets the required dynamic range. Parameter updates must remain within a family for which such a transformation is available.

\paragraph{Physical implementation.}
The experimenter must be able to apply the transformed drive and initial condition and measure the complex field $\bfm{d}$. The experimenter applies the drive $-\bfm{S}\bfm{q}$, prepares $\bfm{d}(0)=i\bfm{S}\bfm{q}_T$, and recovers the conjugate of the adjoint digitally as
\begin{equation}\label{eq:intertwined-linear-recovery}
    \bfm{c}=\bfm{S}^{-1}\bfm{d}.
\end{equation}
This is one additional finite-amplitude run. The transformation does not appear in the forward equation. If $\bfm{S}$ is a spatial symmetry, these operations relocate or coherently combine source and measurement ports.

\paragraph{Reciprocal systems and fully forward mode training.}
For reciprocal hardware, $\bfm{K}^T=\bfm{K}$, the identity $\bfm{S}=\bfm{I}$ satisfies Eq.~\eqref{eq:generalized-reciprocity-linear}. This recovers the linear result of Theorem~\ref{thm:schrodinger}. Suppose the hardware also has a spatial symmetry
\[
    \bfm{S}\bfm{K}(t)\bfm{S}^{-1}
    =
    \bfm{K}(t),
    \qquad
    \bfm{S}\mathbb{P}\bfm{S}^{-1}
    =
    \mathbb{P}^{\mathrm{in}},
\]
which maps the output DOFs onto the input DOFs. The same $\bfm{S}$ is then an intertwiner, so the adjoint source and initial condition can be injected where the data enter. This is fully forward mode training~\cite{xue2024fully}. Reciprocity supplies the transpose propagation, and the spatial symmetry routes the adjoint excitation. The parameter updates must preserve that symmetry. The original derivation uses Green's functions for static Helmholtz propagation; Eq.~\eqref{eq:intertwined-linear-experiment} extends the construction to time-dependent Schr\"odinger trajectories.

\paragraph{Non-reciprocal systems: twisted reciprocity.}
We call a non-reciprocal system \emph{twisted reciprocal} if an orthogonal transformation $\bfm{S}$ satisfies Eq.~\eqref{eq:generalized-reciprocity-linear}. Write $\bfm{K}=\bfm{K}_s+\bfm{K}_a$, where $\bfm{K}_s^T=\bfm{K}_s$ and $\bfm{K}_a^T=-\bfm{K}_a$. Then Eq.~\eqref{eq:generalized-reciprocity-linear} is equivalent to
\begin{equation}\label{eq:sym-structure}
    \bfm{S}\bfm{K}_s\bfm{S}^{-1}
    =
    \bfm{K}_s,
    \qquad
    \bfm{S}\bfm{K}_a\bfm{S}^{-1}
    =
    -\bfm{K}_a.
\end{equation}
The reciprocal part is even under the symmetry, while the non-reciprocal part is odd. Examples include a hopping bias~\cite{hatano1996localization,helbig2020generalized,weidemann2020topological}, a direction of rotation~\cite{nash2015topological,fleury2014sound,scheibner2020odd}, or an in-plane magnetization. In the last case the symmetry supplies the field reversal in the Onsager--Casimir relation $\bfm{K}(\bfm{B})^T=\bfm{K}(-\bfm{B})$~\cite{casimir1945onsager}. The operator remains non-reciprocal, $\bfm{K}^T\neq\bfm{K}$. Equation~\eqref{eq:generalized-reciprocity-linear} is related to the transpose-type symmetries in the non-Hermitian symmetry classification~\cite{kawabata2019symmetry}, used here as a trainability condition.

\paragraph{Example of twisted reciprocity: a Hatano--Nelson chain.}
Consider $N$ sites on an open chain with asymmetric nearest-neighbor hoppings and a real on-site potential~\cite{hatano1996localization},
\begin{equation}\label{eq:hn-chain}
    K_{n+1,n}=Je^{+h},
    \qquad
    K_{n,n+1}=Je^{-h},
    \qquad
    K_{nn}=\varepsilon_n,
    \qquad
    J,h,\varepsilon_n\in\mathbb{R}.
\end{equation}
\begin{figure}[t!]
    \centering
    \includegraphics[width=1\textwidth]{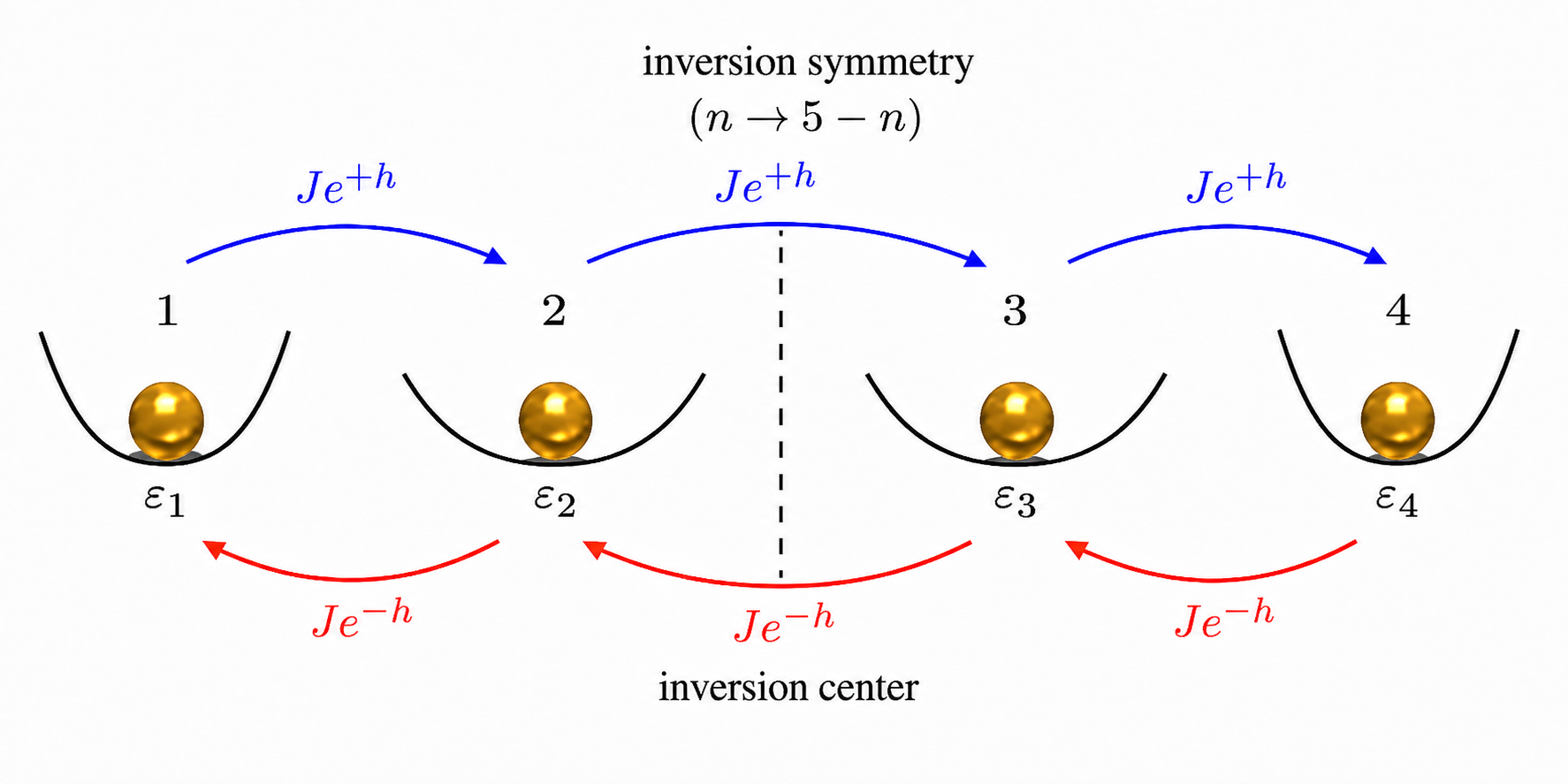}
    \caption{The Hatano--Nelson chain of Eq.~\eqref{eq:hn-chain} for $N=4$. The rightward hopping is amplified by $e^h$ and the leftward hopping is attenuated by $e^{-h}$, so $\bfm{K}^T\neq\bfm{K}$. Transposition reverses every hopping arrow, and spatial inversion reverses it again. For a mirror-symmetric on-site potential, $\varepsilon_n=\varepsilon_{N+1-n}$, the chain obeys $\bfm{S}\bfm{K}^T\bfm{S}^{-1}=\bfm{K}$. The transformed adjoint therefore propagates on the same chain with the same bias.}
    \label{fig:Hatano_Nelson_chain}
\end{figure}

For $N=4$, the operator and spatial inversion $S_{mn}=\delta_{m,N+1-n}$ are
\begin{equation}
    \bfm{K}
    =
    \begin{bmatrix}
        \varepsilon_1 & Je^{-h} & 0 & 0\\
        Je^{+h} & \varepsilon_2 & Je^{-h} & 0\\
        0 & Je^{+h} & \varepsilon_3 & Je^{-h}\\
        0 & 0 & Je^{+h} & \varepsilon_4
    \end{bmatrix},
    \qquad
    \bfm{S}
    =
    \begin{bmatrix}
        0 & 0 & 0 & 1\\
        0 & 0 & 1 & 0\\
        0 & 1 & 0 & 0\\
        1 & 0 & 0 & 0
    \end{bmatrix}.
\end{equation}  
The chain has an exponentially biased directional response and is a standard model of the non-Hermitian skin effect~\cite{kawabata2019symmetry}. Related platforms include feedback-biased active metamaterials~\cite{brandenbourger2019non,ghatak2020observation}, non-reciprocal topolectrical circuits~\cite{helbig2020generalized}, non-reciprocal photonic lattices~\cite{weidemann2020topological}, and acoustic circulators~\cite{fleury2014sound}. The untransformed adjoint in Eq.~\eqref{eq:sym-adjoint} propagates under $\bfm{K}^T$, which the hardware does not provide.

Spatial inversion gives Eq.~\eqref{eq:generalized-reciprocity-linear} when $\varepsilon_n=\varepsilon_{N+1-n}$. The transformed field is obtained in one additional finite-amplitude run with the same bias $h$. If the loss is measured at the right edge, the source $-\bfm{S}\bfm{q}$ and initial condition $i\bfm{S}\bfm{q}_T$ are applied at the left edge. Since $\bfm{S}^{-1}=\bfm{S}$, the adjoint is $\bfm{c}=\bfm{S}\bfm{d}$. Heterogeneous profiles are allowed when $J_n=J_{N-n}$, $h_n=h_{N-n}$, and $\varepsilon_n=\varepsilon_{N+1-n}$. These parameter ties must be preserved during training.

\paragraph{Literature.}
Fully Forward Mode training~\cite{xue2024fully} is an example of the reciprocal spatial-symmetry construction above. 
For linear resistor networks, Lin et al.~\cite{lin2026train} obtain exact gradients from the weighted-reciprocity identity $\bfm{R}\bfm{\Omega}^T\bfm{R}^{-1}=\bfm{\Omega}$, with $\bfm{S}=\bfm{R}$, accessible through voltage- and current-mode measurements.
To the best of the authors' knowledge, prior work has not exploited orthogonal transpose intertwiners to generate physical adjoints on the same non-reciprocal hardware.

\subsection{Nonlinear Schr\"odinger systems: time-reversal mirrors and intertwined reciprocity}\label{s:nonlinear-intertwiners}

The nonlinear part of Theorem~\ref{thm:intertwined-reciprocity} is proved in Appendix~\ref{a:nonlinear-intertwined-proof}.

\paragraph{Physical adjoint algorithm.}
The TRM generates the reversed background
\begin{equation}\label{eq:intertwined-nonlinear-background}
    \bfm{w}(\tau)
    =
    \mathcal{Q}\overline{\bfm{u}(r)}.
\end{equation}
The experimenter must be able to apply the transformed drive and initial condition and measure the resulting complex field. The transformations $\mathcal{Q}$ and $\bfm{S}$ are known and constant during each experiment. If either depends on the parameters, it must be recalibrated after an update. Run a nudged trajectory $\bfm{w}^{\epsilon}$ on the same hardware, with the explicit schedule evaluated at $r$ and the applied drive
\begin{equation}\label{eq:intertwined-nonlinear-drive}
    \bfm{f}^{\epsilon}(\tau)
    =
    \mathcal{Q}\overline{\bfm{f}(r,\bfm{p}_f)}
    -
    \epsilon\mathcal{Q}\bfm{S}\bfm{q}(\tau).
\end{equation}
The initial condition is
\begin{equation}\label{eq:intertwined-nonlinear-initial}
    \bfm{w}^{\epsilon}(0)
    =
    \mathcal{Q}\overline{\bfm{u}(T)}
    +
    \epsilon i\mathcal{Q}\bfm{S}\bfm{q}_T.
\end{equation}
The differential response gives the transformed adjoint,
\begin{equation}\label{eq:nonlinear-total-map}
    \bfm{d}(\tau)
    :=
    \lim_{\epsilon\to0}
    \frac{\bfm{w}^{\epsilon}(\tau)-\bfm{w}(\tau)}{\epsilon}
    =
    \mathcal{Q}\bfm{S}\bfm{c}(\tau),
    \qquad
    \bfm{c}(\tau)
    =
    \bfm{S}^{-1}\mathcal{Q}^{-1}\bfm{d}(\tau).
\end{equation}
This construction is exact in the infinitesimal-nudge limit.

\paragraph{Example: $\mathcal{PT}$ symmetry.}
Let $\bfm{\Pi}$ be a real signed-permutation parity involution. Parity enters through the TRM, while no additional tangent intertwiner is used:
\begin{equation}
    \mathcal{Q}=\bfm{\Pi},
    \qquad
    \bfm{S}=\bfm{I},
    \qquad
    \bfm{d}=\mathcal{Q}\bfm{S}\bfm{c}=\bfm{\Pi}\bfm{c}.
\end{equation}
Condition~\eqref{eq:Q-background-condition} is then the $\mathcal{PT}$ covariance
\begin{equation}
    \bfm{F}\bigl[
        \bfm{\Pi}\overline{\bfm{u}},
        \bfm{\Pi}\bfm{u},
        \bfm{p}_F,t
    \bigr]
    =
    \bfm{\Pi}\overline{
        \bfm{F}\bigl[
            \bfm{u},\overline{\bfm{u}},\bfm{p}_F,t
        \bigr]
    }.
\end{equation}
The direct matching conditions between the physical tangent at the reversed background and the adjoint tangent are
\begin{align}
    \bfm{\Pi}\bfm{F}_{\bfm{u}}
    [\bfm{\Pi}\overline{\bfm{u}},\bfm{\Pi}\bfm{u},\bfm{p}_F,t]\bfm{\Pi}
    &=
    \bfm{F}_{\bfm{u}}
    [\bfm{u},\overline{\bfm{u}},\bfm{p}_F,t]^T,\label{eq:pt-direct-tangent-u}\\
    \bfm{\Pi}\bfm{F}_{\overline{\bfm{u}}}
    [\bfm{\Pi}\overline{\bfm{u}},\bfm{\Pi}\bfm{u},\bfm{p}_F,t]\bfm{\Pi}
    &=
    \overline{\bfm{F}_{\overline{\bfm{u}}}
    [\bfm{u},\overline{\bfm{u}},\bfm{p}_F,t]}^{\,T}.
    \label{eq:pt-direct-tangent-ubar}
\end{align}
Differentiating the $\mathcal{PT}$ covariance shows that the left-hand sides are $\overline{\bfm{F}_{\bfm{u}}}$ and $\overline{\bfm{F}_{\overline{\bfm{u}}}}$, respectively. Equations~\eqref{eq:pt-direct-tangent-u} and~\eqref{eq:pt-direct-tangent-ubar} therefore reduce to conditions~\eqref{eq:S-condition-1} and~\eqref{eq:S-condition-2} with $\bfm{S}=\bfm{I}$. The covariance regenerates the background but does not by itself impose these tangent conditions.

For the Kerr force in Eq.~\eqref{eq:schrodinger-kerr-force},
the TRM condition requires
$\bfm{K}=\bfm{\Pi}\overline{\bfm{K}}\bfm{\Pi}$, while tangent matching requires Hermiticity, i.e. $\bfm{K}=\bar{\bfm{K}}^{T}$. The local Kerr terms satisfy both tangent conditions. The $\mathcal{PT}$ mirror thus permits a Hermitian coupling matrix whose complex coupling phases are reversed by the spatial inversion $\bfm{\Pi}$. In contrast, with ordinary conjugation and $\bfm{S}=\bfm{I}$, background reversal requires $\bfm{K}=\overline{\bfm{K}}$, while tangent matching requires $\bfm{K}^T=\bfm{K}$; together they restrict $\bfm{K}$ to be real symmetric. The recovered adjoint is $\bfm{c}=\bfm{\Pi}\bfm{d}$.

\paragraph{Example: a nonlinear Hatano--Nelson chain.}
Consider the open chain
\begin{equation}\label{eq:nonlinear-hn-chain}
    K_{n+1,n}=J_n e^{h_n},
    \qquad
    K_{n,n+1}=J_n e^{-h_n},
    \qquad
    K_{nn}=\varepsilon_n,
    \qquad
    F_n=(\bfm{K}\bfm{u})_n+g_n|u_n|^2u_n,
\end{equation}
with $J_n,h_n,\varepsilon_n,g_n\in\mathbb{R}$. Ordinary conjugation regenerates the nonlinear background, so $\mathcal{Q}=\bfm{I}$. Define
\begin{equation}\label{eq:nonlinear-hn-intertwiner}
    \bfm{S}_{\mathrm{HN}}
    :=
    \operatorname{diag}(s_1,\ldots,s_N),
    \qquad
    \frac{s_{n+1}}{s_n}=e^{2h_n},
    \qquad
    s_1>0.
\end{equation}
The spatial inversion used above routes sources without amplitude weighting but requires mirror-tied coefficients. In contrast, the diagonal choice satisfies $\bfm{S}_{\mathrm{HN}}\bfm{K}^T\bfm{S}_{\mathrm{HN}}^{-1}=\bfm{K}$ for arbitrary open-chain profiles and commutes with the local Kerr terms, at the cost of site-dependent amplitudes.
The Wirtinger Jacobians are
\begin{equation}
    \bfm{F}_{\bfm{u}}
    =
    \bfm{K}
    +
    2\operatorname{diag}(g_n|u_n|^2),
    \qquad
    \bfm{F}_{\overline{\bfm{u}}}
    =
    \operatorname{diag}(g_nu_n^2).
\end{equation}
The diagonal nonlinear terms commute with $\bfm{S}_{\mathrm{HN}}$, so Eqs.~\eqref{eq:S-condition-1} and~\eqref{eq:S-condition-2} hold with $\bfm{S}=\bfm{S}_{\mathrm{HN}}$. The measured perturbation is $\bfm{d}=\bfm{S}_{\mathrm{HN}}\bfm{c}$.

The diagonal intertwiner preserves the support of the nudge but rescales its amplitude at each site. The adjoint is recovered as $\bfm{c}=\bfm{S}_{\mathrm{HN}}^{-1}\bfm{d}$. If the hopping biases are trained, $\bfm{S}_{\mathrm{HN}}$ must be recalibrated after every update.

\paragraph{Literature.}
Dal Cin et al.~\cite{cin2025training} use an approximate intertwining relation for the scattering matrix linearized about a nonlinear optical steady state; their two-measurement gradient is perturbative in both the nonlinear response and the nudge amplitude. In concurrent numerical work, Li et al.~\cite{li2026situ} combine a $\mathcal{PT}$ mirror with prescribed gain--loss compensation in nonlinear resonator chains, analogous to the \hyperref[par:second-order-compensation]{damping compensation in Sec.~\ref{ss:second-order-nonlinear}}. Their gradients are exact in the infinitesimal-nudge limit.

\subsection{Onsager reciprocity}\label{s:onsager-reciprocity}
Onsager reciprocity~\cite{onsager1931reciprocal,onsager1953fluctuations} is the special case in which the intertwiner comes from a physical mobility. For real state variables, let the constant mobility satisfy $\bfm{V}=\bfm{V}^T\succ0$ and define the Onsager inner product
\begin{equation}
    \langle\bfm{x},\bfm{y}\rangle_{\bfm{V}^{-1}}
    =
    \bfm{x}^T\bfm{V}^{-1}\bfm{y}.
\end{equation}
An operator $\bfm{K}$ is $\bfm{V}$-reciprocal if it is self-adjoint in this metric,
\begin{equation}\label{eq:V-reciprocal-real}
    \bfm{V}\bfm{K}^T\bfm{V}^{-1}
    =
    \bfm{K}.
\end{equation}
Equivalently, $\bfm{K}=\bfm{V}\bfm{H}$ with $\bfm{H}=\bfm{H}^T$. The same transpose identity gives $\bfm{S}=\bfm{V}$ for linear Schr\"odinger systems, including complex symmetric $\bfm{H}$. Here $\bfm{V}$ is part of the dynamics. If the physical source enters as $\bfm{f}=\bfm{V}\bfm{j}$, setting $\bfm{j}=-\bfm{q}$ produces the required applied drive $-\bfm{V}\bfm{q}$. With $\bfm{d}=\bfm{V}\bfm{c}$, the adjoint is recovered as $\bfm{c}=\bfm{V}^{-1}\bfm{d}$. If the trainable parameters enter $\bfm{H}$ and $\bfm{V}$ is fixed, the weighting also cancels from the gradient overlap,
\begin{equation}\label{eq:sym-onsager-cancel}
    2\operatorname{Re}
    \left[
        \bfm{c}^T\bfm{K}_{\bfm{p}}\bfm{u}
    \right]
    =
    2\operatorname{Re}
    \left[
        \bfm{d}^T\bfm{H}_{\bfm{p}}\bfm{u}
    \right].
\end{equation}
Natural source variables can therefore supply the Onsager weighting, and the overlap above does not require explicit inverse weighting. For other parameter blocks, this cancellation need not hold. Symmetry and positivity distinguish an Onsager mobility from a general intertwiner; neither is required by the linear construction.

\paragraph{Nonlinear systems.}
A constant mobility also supplies a nonlinear tangent intertwiner. Let $E[\bfm{u},\overline{\bfm{u}}]\in\mathbb{R}$ be twice continuously differentiable and let $\bfm{V}=\bfm{V}^{\dagger}\succ0$ be constant. If
\begin{equation}\label{eq:nonlinear-onsager-force}
    \bfm{F}
    =
    \bfm{V}E_{\overline{\bfm{u}}},
\end{equation}
then $\bfm{F}_{\bfm{u}}=\bfm{V}E_{\overline{\bfm{u}}\bfm{u}}$ and $\bfm{F}_{\overline{\bfm{u}}}=\bfm{V}E_{\overline{\bfm{u}}\,\overline{\bfm{u}}}$. Reality of $E$ makes $E_{\overline{\bfm{u}}\bfm{u}}$ Hermitian, while commuting Wirtinger derivatives make $E_{\overline{\bfm{u}}\,\overline{\bfm{u}}}$ symmetric. Hence
\begin{align}
    \overline{\bfm{V}}\bfm{F}_{\bfm{u}}^T\overline{\bfm{V}}^{-1}
    &=
    \overline{\bfm{F}_{\bfm{u}}},\\
    \overline{\bfm{V}}\overline{\bfm{F}_{\overline{\bfm{u}}}}^{\,T}\bfm{V}^{-1}
    &=
    \overline{\bfm{F}_{\overline{\bfm{u}}}},
\end{align}
so Eqs.~\eqref{eq:S-condition-1} and~\eqref{eq:S-condition-2} hold with
\begin{equation}
    \bfm{S}=\overline{\bfm{V}}.
\end{equation}
For a real symmetric mobility, $\bfm{S}=\bfm{V}$. This supplies the tangent intertwiner; the TRM condition~\eqref{eq:Q-background-condition} remains independent. The physical perturbation and adjoint recovery are given by Eq.~\eqref{eq:nonlinear-total-map}. Constancy of $\bfm{V}$ is essential here, but positivity is not needed for the intertwining identities.

The nonlinear Hatano--Nelson chain~\eqref{eq:nonlinear-hn-chain} has the form~\eqref{eq:nonlinear-onsager-force} with $\bfm{V}=\bfm{S}_{\mathrm{HN}}$, while ordinary conjugation independently supplies its TRM.

For real stationary gradient flows, $\bfm{F}=\bfm{V}E_{\bfm{u}}$ with constant $\bfm{V}=\bfm{V}^T\succ0$, the nudge is $\epsilon\,\bfm{V}\mathbb{P}\psi_{\bfm{u}}$ and the adjoint is
\begin{equation}
    \bfm{a}
    =
    \bfm{V}^{-1}\lim_{\epsilon\to0}
    \frac{\bfm{u}^{*\epsilon}-\bfm{u}^*}{\epsilon}.
\end{equation}
Positive definiteness makes both free and nudged relaxation descend the tilted energy in the $\bfm{V}^{-1}$ metric. A stable operating point and continuation on the same solution branch are still required.

Appendices~\ref{a:onsager-second-order} and~\ref{a:onsager-first-order} derive the real second- and first-order constructions. Appendix~\ref{a:onsager-rc-network} gives the RC-network realization, where current and charge ports supply the weighting.

\section{Conclusions}\label{s:conclusions}

We have unified a broad range of formally exact physical gradient methods~\cite{scellier2017equilibrium,lopez2023self,pourcel2025learning,xue2024fully,hermans2015trainable,hughes2018training,pai2023experimentally,zhou2020situ,mididoddi2025threading,wanjura2024fully,guillamon2025situ,li2024training,li2025topological} using the adjoint method from PDE-constrained optimization. We have clarified that linear and nonlinear systems behave fundamentally differently and are governed by different sufficient conditions for obtaining the adjoint field on the same hardware. For linear systems, reciprocity is sufficient and the adjoint follows from a single finite-amplitude experiment, whereas nonlinear trajectory systems additionally require a time-reversal mirror, reciprocity of the linearized dynamics, and infinitesimal nudging. For stationary problems, the framework covers Equilibrium Propagation, Quantum Equilibrium Propagation, and in situ gradients in Helmholtz-type systems.

We have further shown that reciprocity is only the simplest instance of more general intertwining conditions. For linear systems, it suffices that the forward operator and its transpose are related by a known, invertible transformation that is constant during each experiment. This recovers fully forward mode training through spatial symmetries and permits exact physical gradients in a class of non-Hermitian, non-reciprocal systems. For nonlinear trajectories, the construction also requires a generalized time-reversal mirror to regenerate the reversed background and an independent intertwining condition that matches the tangent dynamics to the adjoint. The gradient is then obtained in the infinitesimal-nudge limit. Examples include $\mathcal{PT}$-symmetric systems and linear and nonlinear Hatano--Nelson chains~\cite{hatano1996localization}. Onsager reciprocity is the special case in which a constant physical mobility supplies the intertwiner, for both linear and nonlinear dynamics. It also covers stationary gradient flows near stable fixed points, where no time-reversal mirror is needed. The constructions allow explicit time dependence, with schedules replayed in reverse for trajectory adjoints.

Because on-device gradients are far easier to obtain in linear than in nonlinear hardware, parameter encoding (structural nonlinearity)~\cite{xia2024nonlinear,yildirim2024nonlinear} and source-position encoding~\cite{richardson2025nonlinear} are particularly attractive routes to physical computing: the device computes nonlinearly in its inputs, yet its adjoint field, and with it the gradient, remains accessible at finite amplitude and with damping present~\cite{wanjura2024fully, stroev2026programmable}. The non-autonomous generalization makes the temporal modulation of the internal forces itself a trainable resource and may thereby enable novel ways of computing with physical systems.

As discussed in Sec.~\ref{s:local-learning-rules}, when the parameters enter the forward operators linearly, the gradient overlaps are assembled entirely from measured trajectories and never reference the parameter values, so slow drift, aging, and calibration offsets are absorbed automatically. What must still be known is the structure of the model: which measured signals to pair, and where the trainable parameters sit in the equations of motion. How accurately that structure must be known and how the gradients degrade when it is wrong remain open questions.

\paragraph{Remark on concurrent work.} Shortly after our manuscript appeared on arXiv, three works addressing specific cases of physical backpropagation covered by our framework were posted: nonlinear stationary wave scattering with a single nonlinear element by Guillamon et al.~\cite{guillamon2026nonlinear}, $\mathcal{PT}$-assisted trajectory learning in nonlinear resonator chains by Li et al.~\cite{li2026situ}, and time-domain adjoint experiments in linear RLC circuits by Kwon et al.~\cite{kwon2026time}. We discuss their relation to our results in the corresponding sections.

\paragraph{Acknowledgments} We thank Dr. Sam Dillavou, Dr. Fatih Dinc and Andy Davis for helpful discussions and feedback on the manuscript. C.B. was supported by the Swiss National Science Foundation (SNSF) through a Postdoc.Mobility fellowship (P500PT 217673/1).
YG and HT were sponsored by AFOSR MURI Grant No. FA955022-1-0203 and by the Army Research Office under Grant No. W911NF-25-1-0261. The views and conclusions contained in this document are those of the authors and should not be interpreted as representing the official policies, either expressed or implied, of the Army Research Office or the U.S. Government. The U.S. Government is authorized to reproduce and distribute reprints for Government purposes notwithstanding any copyright notation herein.
This work was performed in part at Aspen Center for Physics, which is supported by National Science Foundation grant PHY-2210452.

\paragraph{Data availability} Notebooks implementing representative physical adjoint algorithms and numerical gradient checks are publicly available in Ref.~\cite{physicalAdjointNotebooks}.

\bibliographystyle{unsrt}
\bibliography{mybibfile}
\appendix
\phantomsection
\addcontentsline{toc}{section}{Appendix}
\stopcontents[main]
\startcontents[appendix]
\section*{Contents of the appendices}
\begingroup\hypersetup{linkcolor=sectioncolor}
\printcontents[appendix]{}{0}{\setcounter{tocdepth}{2}}
\endgroup
\section{Derivation of the adjoint state equation and the gradient formula}\label{apA}

In this appendix we derive the adjoint state equation and the gradient formula through the Lagrangian method of PDE-constrained optimization \cite{hinze2009optimization, Giles2000}, in the Lagrange-multiplier form that is standard for time-dependent inverse problems \cite{Plessix2006}.
The derivation is organized around a single object, the Lagrangian $\mathcal{L}[\bfm{u},\bfm{b},\boldsymbol{\lambda},\boldsymbol{\mu},\bfm{p}]$, viewed as a functional of \emph{independent} arguments: the state trajectory $\bfm{u}$, the multipliers $\bfm{b}$, $\boldsymbol{\lambda}$, $\boldsymbol{\mu}$, and the parameters $\bfm{p}$.
Every ingredient of the adjoint method is one of its partial derivatives:
\begin{itemize}
    \item stationarity with respect to the multipliers, $\mathcal{L}_{\bfm{b}} = 0$, $\mathcal{L}_{\boldsymbol{\lambda}} = 0$, $\mathcal{L}_{\boldsymbol{\mu}} = 0$, recovers the \emph{forward problem} (state equation and initial conditions);
    \item stationarity with respect to the state trajectory, $\mathcal{L}_{\bfm{u}} = 0$, defines the \emph{adjoint problem};
    \item at a point where all four conditions hold, the remaining partial derivative $\mathcal{L}_{\bfm{p}}$ equals the total derivative $d\chi/d\bfm{p}$ of the reduced cost, and evaluating it yields the \emph{gradient} as forward--adjoint overlap integrals.
\end{itemize}

\subsection{State equation}

Let $V$ be a real Hilbert space with inner product $\langle\cdot,\cdot\rangle$.
Here $V = \mathbb{R}^N$ (a discrete system with $N$ degrees of freedom),
\begin{equation}
    \langle \bfm{x},\, \bfm{y}\rangle = \sum_{n=1}^N x_n\, y_n = \bfm{x}^T\bfm{y}.
\end{equation}
Let $\bfm{u}$ be a path
\begin{equation}
    \bfm{u} \in C^2([0,T],\, V),
\end{equation}
so that each $\bfm{u}(t) \in V$.
The state equation is
\begin{equation}\label{apA:eq:state}
    \bfm{M}(\bfm{p}_M)\,\ddot{\bfm{u}}(t) + \bfm{D}(\bfm{p}_D)\,\dot{\bfm{u}}(t) + \bfm{F}[\bfm{u}(t),\bfm{p}_F,t] - \bfm{f}(t,\bfm{p}_f) = 0, \qquad t \in [0,T],
\end{equation}
with parametrized initial conditions $\bfm{u}(0) = \bfm{u}_0(\bfm{p}_{u_0})$, $\dot{\bfm{u}}(0) = \bfm{v}_0(\bfm{p}_{v_0})$.
Here $\bfm{M}(\bfm{p}_M), \bfm{D}(\bfm{p}_D): V \to V$ are linear operators, $\bfm{F}[\cdot,\bfm{p}_F,t]: V \to V$ is a possibly nonlinear operator, and $\bfm{f}:[0,T] \times \mathbb{R}^{M_f}\to V$.
We assume that $\bfm{M}(\bfm{p}_M)$ is invertible for all visited $\bfm{p}_M$.
Each term carries its own parameters $\bfm{p}_x \in \mathbb{R}^{M_x}$ for $x\in\{M,D,F,f,u_0,v_0\}$, collected into a single vector $\bfm{p} = [\bfm{p}_M,\bfm{p}_D,\bfm{p}_F,\bfm{p}_f,\bfm{p}_{u_0},\bfm{p}_{v_0}] \in \mathbb{R}^M$ with $M=\sum_x M_x$.

We use round brackets for explicit parameter dependence, e.g.\ $\bfm{M}(\bfm{p}_M)$, and square brackets for functional dependence, e.g.\ $\bfm{F}[\bfm{u}]$.

\subsection{Cost function}

Given a path $\bfm{u} \in C^2([0,T],\, V)$, we define the cost $\chi: C^2([0,T],\, V) \to \mathbb{R}$ by
\begin{equation}\label{apA:eq:cost}
    \chi[\bfm{u}] := \int_0^T\theta[\mathbb{P}\bfm{u}(t)]dt + \psi[\mathbb{P}\bfm{u}(T), \mathbb{P}\dot{\bfm{u}}(T)],
\end{equation}
where $\mathbb{P}$ is an orthogonal projector of the state onto output degrees of freedom, i.e.\ $\mathbb{P}=\mathbb{P}^T=\mathbb{P}^2$, and $\theta:V\to\mathbb{R}$ and $\psi:V \times V \to \mathbb{R}$ are possibly nonlinear functionals.

\subsection{Differentiability assumptions}

The dynamics may be non-autonomous, so $\bfm{F}$ carries an explicit time argument alongside its state and parameter arguments.
We require $\bfm{F}$ to be $C^2$ in the latter two and continuous in $t$, and $\bfm{f}$ to be $C^2$ in $\bfm{p}$. The derivatives below are taken at fixed $t$, which we suppress in this subsection.
Concretely, for each fixed $\bfm{p}$ the map $\bfm{u} \mapsto \bfm{F}[\bfm{u},\bfm{p}]$ has a Fr\'echet derivative $\bfm{F}_{\bfm{u}}: V \to V$, a bounded linear operator satisfying
\begin{equation}\label{apA:eq:Fu}
    \bfm{F}[\bfm{u}+\varepsilon\bfm{h},\,\bfm{p}] = \bfm{F}[\bfm{u},\bfm{p}] + \varepsilon\,\bfm{F}_{\bfm{u}}[\bfm{u},\bfm{p}]\,\bfm{h} + \mathcal{O}(\varepsilon^2) \qquad \forall\,\bfm{h} \in V,
\end{equation}
where $\bfm{F}_{\bfm{u}}[\bfm{u},\bfm{p}] \in \mathbb{R}^{N\times N}$ is the Jacobian with entries $\left(\frac{\partial F_i}{\partial u_j}\right)$.
Its adjoint $\bfm{F}_{\bfm{u}}^T: V \to V$ is defined by
\begin{equation}\label{apA:eq:Fu-adjoint}
    \langle \bfm{x},\, \bfm{F}_{\bfm{u}}\bfm{h}\rangle = \langle \bfm{F}_{\bfm{u}}^T\bfm{x},\, \bfm{h}\rangle \qquad \forall\,\bfm{x},\bfm{h} \in V.
\end{equation}
This identity is purely algebraic in $V$, so no boundary terms arise.

Analogously, for each fixed $\bfm{u}$ the map $\bfm{p} \mapsto \bfm{F}[\bfm{u},\bfm{p}]$ has a derivative $\bfm{F}_{\bfm{p}}(\bfm{u},\bfm{p}): \mathbb{R}^M \to V$ (the $N \times M$ matrix $\left(\frac{\partial F_i}{\partial p_j}\right)$), and for each fixed $t$ the map $\bfm{p} \mapsto \bfm{f}(t,\bfm{p})$ has a derivative $\bfm{f}_{\bfm{p}}(t,\bfm{p}): \mathbb{R}^M \to V$.
All these derivatives exist and depend continuously on their arguments.
By the Picard--Lindel\"of theorem, the state equation~\eqref{apA:eq:state} then has a unique solution $\bfm{u}(\cdot;\bfm{p})$ for each $\bfm{p}$, and this solution depends smoothly on $\bfm{p}$.

\newpage
\subsection{Constrained optimization problem}

The constrained optimization problem is the following.
\begin{problem}[Constrained optimization of a second-order ODE]\label{apA:prob:opt}
Find
\begin{equation}\label{apA:eq:opt}
    \min_{\bfm{p}\in\mathbb{R}^M} \chi[\bfm{u}] \qquad \text{subject to} \qquad
    \begin{cases}
    \bfm{M}(\bfm{p}_M)\,\ddot{\bfm{u}}(t) + \bfm{D}(\bfm{p}_D)\,\dot{\bfm{u}}(t) + \bfm{F}[\bfm{u}(t),\bfm{p}_F,t] = \bfm{f}(t,\bfm{p}_f), & t \in [0,T], \\[4pt]
    \bfm{u}(0) = \bfm{u}_0(\bfm{p}_{u_0}), \\[4pt]
    \dot{\bfm{u}}(0) = \bfm{v}_0(\bfm{p}_{v_0}).
    \end{cases}
\end{equation}
\end{problem}
The adjoint method computes the gradient of the cost function, $d\chi/d\bfm{p}$, at the cost of one additional ODE solve, independent of the number of parameters $M$.
We can then solve the optimization problem, for example, by gradient descent:
\begin{equation}\label{apA:eq:update}
    \bfm{p} \leftarrow \bfm{p} - \alpha\,\frac{d\chi}{d\bfm{p}},
\end{equation}
where $\alpha > 0$ is the learning rate.

\subsection{The Lagrangian}

We enforce the state equation with a multiplier path $\bfm{b} \in C^2([0,T], V)$ and the two initial conditions with multiplier vectors $\boldsymbol{\lambda}, \boldsymbol{\mu} \in V$, and define
\begin{multline}\label{apA:eq:lagrangian}
    \mathcal{L}[\bfm{u}, \bfm{b}, \boldsymbol{\lambda}, \boldsymbol{\mu}, \bfm{p}]
    := \int_0^T \langle \bfm{b}(t),\; \bfm{M}(\bfm{p}_M)\ddot{\bfm{u}}(t) + \bfm{D}(\bfm{p}_D)\dot{\bfm{u}}(t) + \bfm{F}[\bfm{u}(t),\bfm{p}_F,t] - \bfm{f}(t,\bfm{p}_f) \rangle\, dt \\
    + \langle \boldsymbol{\lambda},\; \bfm{u}(0) - \bfm{u}_0(\bfm{p}_{u_0}) \rangle
    + \langle \boldsymbol{\mu},\; \dot{\bfm{u}}(0) - \bfm{v}_0(\bfm{p}_{v_0}) \rangle
    + \chi[\bfm{u}].
\end{multline}
This is a functional $\mathcal{L}: C^2([0,T],V) \times C^2([0,T],V) \times V \times V \times \mathbb{R}^M \to \mathbb{R}$.

The five arguments of $\mathcal{L}$ are \emph{independent} variables: $\mathcal{L}$ is defined on the full, unconstrained product space.
In particular, $\bfm{u}$ is an arbitrary path and is \emph{not} assumed to solve the state equation, while $\bfm{p}$ enters only through the explicit dependencies $\bfm{M}(\bfm{p}_M)$, $\bfm{D}(\bfm{p}_D)$, $\bfm{F}[\cdot,\bfm{p}_F,\cdot]$, $\bfm{f}(\cdot,\bfm{p}_f)$, $\bfm{u}_0(\bfm{p}_{u_0})$, $\bfm{v}_0(\bfm{p}_{v_0})$.
We may therefore differentiate $\mathcal{L}$ with respect to each argument separately, holding the other four fixed.
These \emph{partial} (Fr\'echet) derivatives are denoted $\mathcal{L}_{\bfm{b}}$, $\mathcal{L}_{\boldsymbol{\lambda}}$, $\mathcal{L}_{\boldsymbol{\mu}}$, $\mathcal{L}_{\bfm{u}}$, and $\mathcal{L}_{\bfm{p}}$.
We compute these derivatives in turn.

\subsection{Stationarity with respect to the multipliers: the forward problem}\label{apA:ss:forward}

The Lagrangian is affine in $\bfm{b}$, $\boldsymbol{\lambda}$, and $\boldsymbol{\mu}$, so these three derivatives can be read off directly. Identifying them with elements of the respective spaces through the inner product,
\begin{equation}\label{apA:eq:La}
    \mathcal{L}_{\bfm{b}}(t) = \bfm{M}\ddot{\bfm{u}}(t) + \bfm{D}\dot{\bfm{u}}(t) + \bfm{F}[\bfm{u}(t),\bfm{p}_F,t] - \bfm{f}(t,\bfm{p}_f),
\end{equation}
and
\begin{equation}\label{apA:eq:Llm}
    \mathcal{L}_{\boldsymbol{\lambda}} = \bfm{u}(0) - \bfm{u}_0(\bfm{p}_{u_0}), \qquad
    \mathcal{L}_{\boldsymbol{\mu}} = \dot{\bfm{u}}(0) - \bfm{v}_0(\bfm{p}_{v_0}),
\end{equation}
so $\mathcal{L}_{\bfm{b}} = 0$ holds if and only if $\bfm{u}$ satisfies the state equation~\eqref{apA:eq:state} on $[0,T]$, and $\mathcal{L}_{\boldsymbol{\lambda}} = 0$, $\mathcal{L}_{\boldsymbol{\mu}} = 0$ if and only if $\bfm{u}$ satisfies the initial conditions.

Stationarity with respect to a multiplier thus recovers the constraint that the multiplier enforces, which is the generic mechanism of Lagrangian relaxation.
The three conditions together state that $\bfm{u} = \bfm{u}(\cdot;\bfm{p})$ is the forward solution.

\subsection{Stationarity turns the partial derivative into the gradient}\label{apA:ss:key-identity}

Fix $\bfm{p}$ and evaluate $\mathcal{L}$ at the forward solution $\bfm{u} = \bfm{u}(\cdot;\bfm{p})$.
The constraint terms in~\eqref{apA:eq:lagrangian} vanish, so
\begin{equation}\label{apA:eq:L-equals-chi}
    \chi(\bfm{p}) = \mathcal{L}[\bfm{u}(\cdot;\bfm{p}),\, \bfm{b},\, \boldsymbol{\lambda},\, \boldsymbol{\mu},\, \bfm{p}] \qquad \forall\, \bfm{b},\, \boldsymbol{\lambda},\, \boldsymbol{\mu}.
\end{equation}
Since the multipliers are arbitrary in~\eqref{apA:eq:L-equals-chi}, we may let them, too, depend on $\bfm{p}$: pick any assignments $\bfm{b}(\cdot;\bfm{p})$, $\boldsymbol{\lambda}(\bfm{p})$, $\boldsymbol{\mu}(\bfm{p})$ that depend differentiably on $\bfm{p}$ (the adjoint state constructed below does).
Differentiating~\eqref{apA:eq:L-equals-chi} with respect to $\bfm{p}$ by the chain rule, one term per argument of $\mathcal{L}$, gives
\begin{equation}\label{apA:eq:chain-rule}
    \frac{d\chi}{d\bfm{p}}
    = \mathcal{L}_{\bfm{u}}\frac{d\bfm{u}}{d\bfm{p}}
    + \mathcal{L}_{\bfm{b}}\frac{d\bfm{b}}{d\bfm{p}}
    + \mathcal{L}_{\boldsymbol{\lambda}}\frac{d\boldsymbol{\lambda}}{d\bfm{p}}
    + \mathcal{L}_{\boldsymbol{\mu}}\frac{d\boldsymbol{\mu}}{d\bfm{p}}
    + \mathcal{L}_{\bfm{p}},
\end{equation}
with every partial derivative evaluated at the point $(\bfm{u}(\cdot;\bfm{p}),\, \bfm{b}(\cdot;\bfm{p}),\, \boldsymbol{\lambda}(\bfm{p}),\, \boldsymbol{\mu}(\bfm{p}),\, \bfm{p})$.
The left-hand side is the \emph{total} derivative of the reduced cost. The right-hand side contains \emph{partial} derivatives of $\mathcal{L}$ multiplying the sensitivities of its arguments.

Now inspect~\eqref{apA:eq:chain-rule} term by term.
The three multiplier terms vanish because $\bfm{u}(\cdot;\bfm{p})$ is the forward solution: by Sec.~\ref{apA:ss:forward}, $\mathcal{L}_{\bfm{b}} = 0$, $\mathcal{L}_{\boldsymbol{\lambda}} = 0$, $\mathcal{L}_{\boldsymbol{\mu}} = 0$ at this point, whatever the sensitivities $d\bfm{b}/d\bfm{p}$, $d\boldsymbol{\lambda}/d\bfm{p}$, $d\boldsymbol{\mu}/d\bfm{p}$ may be.
The remaining obstruction is the first term.
Computing the sensitivity $d\bfm{u}/d\bfm{p}$ requires differentiating through the differential equation solver, with one linearized solve per parameter. The $M$ solves are in general prohibitively expensive.
But it enters~\eqref{apA:eq:chain-rule} only multiplied by $\mathcal{L}_{\bfm{u}}$, and the multipliers $\bfm{b}$, $\boldsymbol{\lambda}$, $\boldsymbol{\mu}$ are still free.
If we choose them such that
\begin{equation}\label{apA:eq:Lu-zero}
    \mathcal{L}_{\bfm{u}} = 0 \qquad \text{at } \bfm{u} = \bfm{u}(\cdot;\bfm{p}),
\end{equation}
then $d\bfm{u}/d\bfm{p}$ is never needed, and
\begin{equation}\label{apA:eq:key-identity}
    \frac{d\chi}{d\bfm{p}} = \mathcal{L}_{\bfm{p}}:
\end{equation}
the total derivative of the cost equals the partial derivative of the Lagrangian, which involves only the explicit parameter dependence of $\bfm{M}$, $\bfm{D}$, $\bfm{F}$, $\bfm{f}$, $\bfm{u}_0$, $\bfm{v}_0$ and is inexpensive to evaluate.
The condition~\eqref{apA:eq:Lu-zero} is a linear equation for $(\bfm{b}, \boldsymbol{\lambda}, \boldsymbol{\mu})$. We now show that it has exactly one solution, the \emph{adjoint state}.

\subsection{Stationarity with respect to the state: the adjoint problem}\label{apA:ss:adjoint}

Let $\bfm{h} \in C^2([0,T], V)$ be an arbitrary test path with no constraints at either endpoint, and compute
\begin{equation}
    \mathcal{L}_{\bfm{u}}\bfm{h} := \lim_{\varepsilon \to 0}\frac{\mathcal{L}[\bfm{u}+\varepsilon\bfm{h},\,\bfm{b},\,\boldsymbol{\lambda},\,\boldsymbol{\mu},\,\bfm{p}] - \mathcal{L}[\bfm{u},\,\bfm{b},\,\boldsymbol{\lambda},\,\boldsymbol{\mu},\,\bfm{p}]}{\varepsilon}.
\end{equation}
The Lagrangian~\eqref{apA:eq:lagrangian} has three groups of terms.
We vary each in turn.

\paragraph{(i) The constraint integral.}
Replacing $\bfm{u} \to \bfm{u} + \varepsilon\bfm{h}$ and using the linearity of $\bfm{M}$, $\bfm{D}$ and the differentiability~\eqref{apA:eq:Fu} of $\bfm{F}$ gives the following variation of the first term of~\eqref{apA:eq:lagrangian} (the drive $\bfm{f}$ has no $\bfm{u}$-dependence and does not contribute):
\begin{equation}
    \int_0^T \langle \bfm{b},\; \bfm{M}\ddot{\bfm{h}} + \bfm{D}\dot{\bfm{h}} + \bfm{F}_{\bfm{u}}\bfm{h}\rangle\,dt.
\end{equation}
To read off conditions on $\bfm{b}$, all terms must appear as $\langle\,\cdot\,,\bfm{h}\rangle$: we move the operators across the inner product using $\langle \bfm{x}, \bfm{M}\bfm{y}\rangle = \langle \bfm{M}^T\bfm{x}, \bfm{y}\rangle$ (and~\eqref{apA:eq:Fu-adjoint} for $\bfm{F}_{\bfm{u}}$), and move the time derivatives off $\bfm{h}$ and onto $\bfm{b}$ via integration by parts, twice for the mass term and once for the damping term:
\begin{align}
    \int_0^T \langle \bfm{M}^T\bfm{b},\,\ddot{\bfm{h}}\rangle\,dt
    &= \int_0^T \langle \bfm{M}^T\ddot{\bfm{b}},\,\bfm{h}\rangle\,dt
    + \Big[\langle \bfm{M}^T\bfm{b},\,\dot{\bfm{h}}\rangle - \langle \bfm{M}^T\dot{\bfm{b}},\,\bfm{h}\rangle\Big]_0^T, \\
    \int_0^T \langle \bfm{D}^T\bfm{b},\,\dot{\bfm{h}}\rangle\,dt
    &= -\int_0^T \langle \bfm{D}^T\dot{\bfm{b}},\,\bfm{h}\rangle\,dt
    + \Big[\langle \bfm{D}^T\bfm{b},\,\bfm{h}\rangle\Big]_0^T.
\end{align}
Combining these terms gives the variation of the constraint integral:
\begin{equation}\label{apA:eq:constraint-variation}
    \int_0^T \langle \bfm{M}^T\ddot{\bfm{b}} - \bfm{D}^T\dot{\bfm{b}} + \bfm{F}_{\bfm{u}}^T\bfm{b},\;\bfm{h}\rangle\,dt
    + \Big[\langle \bfm{D}^T\bfm{b} - \bfm{M}^T\dot{\bfm{b}},\;\bfm{h}\rangle + \langle \bfm{M}^T\bfm{b},\;\dot{\bfm{h}}\rangle\Big]_0^T.
\end{equation}

\paragraph{(ii) The initial-condition multiplier terms.}
Under $\bfm{u} \to \bfm{u} + \varepsilon\bfm{h}$, the variation is simply
\begin{equation}\label{apA:eq:ic-variation}
    \langle\boldsymbol{\lambda},\,\bfm{h}(0)\rangle + \langle\boldsymbol{\mu},\,\dot{\bfm{h}}(0)\rangle.
\end{equation}

\paragraph{(iii) The cost.}
The differentiability of $\theta$ and $\psi$, together with $\mathbb{P}^T = \mathbb{P}$, gives the cost variation:
\begin{equation}\label{apA:eq:cost-variation}
    \int_0^T \langle\mathbb{P}\theta_{\bfm{u}},\,\bfm{h}\rangle\,dt
    + \langle\mathbb{P}\psi_{\bfm{u}},\,\bfm{h}(T)\rangle
    + \langle\mathbb{P}\psi_{\dot{\bfm{u}}},\,\dot{\bfm{h}}(T)\rangle.
\end{equation}

\paragraph{The full variation.}
Adding Eqs.~\eqref{apA:eq:constraint-variation}, \eqref{apA:eq:ic-variation}, \eqref{apA:eq:cost-variation} and grouping terms by where $\bfm{h}$ and $\dot{\bfm{h}}$ are evaluated gives
\begin{multline}\label{apA:eq:full-variation}
    \mathcal{L}_{\bfm{u}}\bfm{h} = \int_0^T \langle \bfm{M}^T\ddot{\bfm{b}} - \bfm{D}^T\dot{\bfm{b}} + \bfm{F}_{\bfm{u}}^T\bfm{b} + \mathbb{P}\theta_{\bfm{u}},\;\bfm{h}\rangle\,dt \\
    + \langle \bfm{D}^T\bfm{b}(T) - \bfm{M}^T\dot{\bfm{b}}(T) + \mathbb{P}\psi_{\bfm{u}},\;\bfm{h}(T)\rangle
    + \langle \bfm{M}^T\bfm{b}(T) + \mathbb{P}\psi_{\dot{\bfm{u}}},\;\dot{\bfm{h}}(T)\rangle \\
    + \langle \bfm{M}^T\dot{\bfm{b}}(0) - \bfm{D}^T\bfm{b}(0) + \boldsymbol{\lambda},\;\bfm{h}(0)\rangle
    + \langle -\bfm{M}^T\bfm{b}(0) + \boldsymbol{\mu},\;\dot{\bfm{h}}(0)\rangle.
\end{multline}

\paragraph{Reading off the conditions.}
Since $\bfm{h}$ is arbitrary and unconstrained at both endpoints, the values $\bfm{h}(t)$, $\bfm{h}(0)$, $\dot{\bfm{h}}(0)$, $\bfm{h}(T)$, $\dot{\bfm{h}}(T)$ can be chosen independently, so $\mathcal{L}_{\bfm{u}} = 0$ requires each inner product in~\eqref{apA:eq:full-variation} to vanish separately.
The bulk term and the two terms at $t = T$ involve only $\bfm{b}$. Together they form a terminal-value problem that determines $\bfm{b}$ uniquely and defines the adjoint state.
\needspace{14\baselineskip}
\begin{definition}[Adjoint state for the second-order problem]\label{apA:def:adjoint}
The \emph{adjoint state} $\bfm{b}(t)\in V$ associated with Problem~\ref{apA:prob:opt} is the solution of the adjoint ODE
\begin{equation}\label{apA:eq:adjoint-ode}
    \bfm{M}^T\ddot{\bfm{b}}(t) - \bfm{D}^T\dot{\bfm{b}}(t) + \bfm{F}_{\bfm{u}}^T[\bfm{u}(t),\bfm{p}_F,t]\,\bfm{b}(t) + \mathbb{P}\theta_{\bfm{u}}[\mathbb{P}\bfm{u}(t)] = 0, \qquad t \in [0,T],
\end{equation}
a second-order ODE requiring two conditions, with the terminal conditions
\begin{align}
    \bfm{M}^T\bfm{b}(T) &= -\mathbb{P}\psi_{\dot{\bfm{u}}}[\mathbb{P}\bfm{u}(T),\mathbb{P}\dot{\bfm{u}}(T)], \label{apA:eq:tc1} \\
    \bfm{M}^T\dot{\bfm{b}}(T) &= \bfm{D}^T\bfm{b}(T) + \mathbb{P}\psi_{\bfm{u}}[\mathbb{P}\bfm{u}(T), \mathbb{P}\dot{\bfm{u}}(T)]. \label{apA:eq:tc2}
\end{align}
\end{definition}
We solve the adjoint ODE~\eqref{apA:eq:adjoint-ode} backward from $T$ to $0$ using these terminal conditions.
This produces $\bfm{b}(0)$ and $\dot{\bfm{b}}(0)$, which in turn determine the multipliers via the two terms at $t = 0$ in~\eqref{apA:eq:full-variation}:
\begin{align}
    \boldsymbol{\lambda} &= -\bfm{M}^T\dot{\bfm{b}}(0) + \bfm{D}^T\bfm{b}(0), \label{apA:eq:lambda} \\
    \boldsymbol{\mu} &= \bfm{M}^T\bfm{b}(0). \label{apA:eq:mu}
\end{align}
The adjoint ODE is linear in $\bfm{b}$ (the nonlinear forces enter only through their Jacobian, evaluated along the forward trajectory), its coefficients depend continuously on $\bfm{p}$, and~\eqref{apA:eq:lambda}--\eqref{apA:eq:mu} are explicit. Hence $(\bfm{b}, \boldsymbol{\lambda}, \boldsymbol{\mu})$ is the unique solution of $\mathcal{L}_{\bfm{u}} = 0$ and depends smoothly on $\bfm{p}$, as required in Sec.~\ref{apA:ss:key-identity}.

\subsection{The gradient}\label{apA:ss:gradient}

The gradient follows by evaluating the remaining partial derivative of $\mathcal{L}$, namely $\mathcal{L}_{\bfm{p}}$.
\begin{theorem}[Adjoint gradient for the second-order problem]\label{apA:thm:gradient}
Let $\bfm{u} = \bfm{u}(\cdot;\bfm{p})$ solve the forward problem of Problem~\ref{apA:prob:opt} and let $\bfm{b}$ be the adjoint state of Definition~\ref{apA:def:adjoint}. Then the gradient of the reduced cost with respect to $\bfm{p}$ is
\begin{equation}\label{apA:eq:gradient}
  \begin{aligned}
    \frac{d\chi}{d\bfm{p}}
    &= \int_0^T \!\bigl\langle \bfm{b}(t),\;
       \bfm{M}_{\bfm{p}_M}\ddot{\bfm{u}}(t)
       + \bfm{D}_{\bfm{p}_D}\dot{\bfm{u}}(t)
       + \bfm{F}_{\bfm{p}_F}[\bfm{u}(t),\bfm{p}_F,t]
       - \bfm{f}_{\bfm{p}_f}(t,\bfm{p})
       \bigr\rangle dt \\
    &\quad + \bigl\langle \bfm{M}^T\dot{\bfm{b}}(0) - \bfm{D}^T\bfm{b}(0),\,
       \frac{d\bfm{u}_0}{d\bfm{p}_{u_0}}\bigr\rangle
     - \bigl\langle \bfm{M}^T\bfm{b}(0),\,
       \frac{d\bfm{v}_0}{d\bfm{p}_{v_0}}\bigr\rangle.
  \end{aligned}
\end{equation}
\end{theorem}
\noindent\textit{Proof.}
With $(\bfm{b}, \boldsymbol{\lambda}, \boldsymbol{\mu})$ chosen as in Definition~\ref{apA:def:adjoint} and Eqs.~\eqref{apA:eq:lambda}--\eqref{apA:eq:mu}, we have $\mathcal{L}_{\bfm{u}} = 0$, and the key identity~\eqref{apA:eq:key-identity} gives $d\chi/d\bfm{p} = \mathcal{L}_{\bfm{p}}$.
The explicit $\bfm{p}$-dependence of $\mathcal{L}$ enters through $\bfm{M}(\bfm{p}_M)$, $\bfm{D}(\bfm{p}_D)$, $\bfm{F}[\bfm{u},\bfm{p}_F,t]$, $\bfm{f}(t,\bfm{p}_f)$ in the constraint integral and through $\bfm{u}_0(\bfm{p}_{u_0})$, $\bfm{v}_0(\bfm{p}_{v_0})$ in the initial-condition terms. The cost $\chi$ has no explicit parameter dependence.
Differentiating each term at fixed $\bfm{u}$, $\bfm{b}$, $\boldsymbol{\lambda}$, $\boldsymbol{\mu}$ gives
\begin{equation}\label{apA:eq:Lp}
    \mathcal{L}_{\bfm{p}} = \int_0^T \langle \bfm{b}(t),\; \bfm{M}_{\bfm{p}_M}\ddot{\bfm{u}}(t) + \bfm{D}_{\bfm{p}_D}\dot{\bfm{u}}(t) + \bfm{F}_{\bfm{p}_F}[\bfm{u}(t),\bfm{p}_F,t] - \bfm{f}_{\bfm{p}_f}(t,\bfm{p})\rangle\,dt
    - \langle\boldsymbol{\lambda},\,\frac{d\bfm{u}_0}{d\bfm{p}_{u_0}}\rangle
    - \langle\boldsymbol{\mu},\,\frac{d\bfm{v}_0}{d\bfm{p}_{v_0}}\rangle.
\end{equation}
Substituting the expressions~\eqref{apA:eq:lambda}--\eqref{apA:eq:mu} for $\boldsymbol{\lambda}$ and $\boldsymbol{\mu}$ gives~\eqref{apA:eq:gradient}. \qed

Every quantity on the right-hand side of~\eqref{apA:eq:gradient} is known: $\bfm{u}(t)$ from the forward solve, $\bfm{b}(t)$ from the backward solve, and the partial derivatives $\bfm{M}_{\bfm{p}_M}$, $\bfm{D}_{\bfm{p}_D}$, $\bfm{F}_{\bfm{p}_F}$, $\bfm{f}_{\bfm{p}_f}$, $d\bfm{u}_0/d\bfm{p}_{u_0}$, $d\bfm{v}_0/d\bfm{p}_{v_0}$ from the model.
The cost of evaluating~\eqref{apA:eq:gradient} is one forward solve plus one backward solve, independent of the number of parameters $M$.

If $\bfm{M}$ and $\bfm{D}$ do not depend on $\bfm{p}$ and the initial conditions are $\bfm{p}$-independent, Eq.~\eqref{apA:eq:gradient} reduces to
\begin{equation}\label{apA:eq:gradient-simple}
    \frac{d\chi}{d\bfm{p}} = \int_0^T \langle \bfm{b}(t),\; \bfm{F}_{\bfm{p}}(\bfm{u}(t),\bfm{p}) - \bfm{f}_{\bfm{p}}(t,\bfm{p})\rangle\,dt.
\end{equation}

\paragraph{Interpreting $\langle\bfm{b},\bfm{F}_{\bfm{p}}\rangle$.}
Recall that $\bfm{F}_{\bfm{p}}:\mathbb{R}^M \to V$ is a linear \emph{operator}, not an element of $V$, so the symbol $\langle\bfm{b},\bfm{F}_{\bfm{p}}\rangle$ is a shorthand for an $M$-component vector, to be read componentwise: for each parameter $p_j$, the partial derivative $\partial \bfm{F}/\partial p_j \in V$ is a vector, and the $j$-th component of the gradient is the scalar inner product
\begin{equation}
    \frac{d\chi}{dp_j} = \int_0^T \left\langle \bfm{b}(t),\; \frac{\partial \bfm{F}}{\partial p_j}[\bfm{u}(t),\bfm{p}] - \frac{\partial \bfm{f}}{\partial p_j}(t,\bfm{p})\right\rangle dt.
\end{equation}
Equivalently, using the adjoint operator $\bfm{F}_{\bfm{p}}^T: V \to \mathbb{R}^M$ defined by $\langle\bfm{b},\bfm{F}_{\bfm{p}}\bfm{q}\rangle_V = \langle \bfm{F}_{\bfm{p}}^T\bfm{b},\bfm{q}\rangle_{\mathbb{R}^M}$ for all $\bfm{q}\in\mathbb{R}^M$, the compact form is
\begin{equation}
    \frac{d\chi}{d\bfm{p}} = \int_0^T \big(\bfm{F}_{\bfm{p}}^T\bfm{b}(t) - \bfm{f}_{\bfm{p}}^T\bfm{b}(t)\big)\,dt \;\in\; \mathbb{R}^M.
\end{equation}
Hence $\bfm{F}_{\bfm{p}}^T\bfm{b}$ is the $M$-vector obtained by multiplying the transpose of the $N\times M$ Jacobian by the $N$-vector $\bfm{b}$. The same interpretation applies to the $\bfm{F}_{\bfm{p}_F}$, $\bfm{M}_{\bfm{p}_M}$, $\bfm{D}_{\bfm{p}_D}$, and $\bfm{f}_{\bfm{p}_f}$ terms in the full gradient~\eqref{apA:eq:gradient}.

\paragraph{Direct optimization of initial conditions.}
Suppose a subvector of $\bfm{p}$ encodes the initial state itself, i.e.\ $\bfm{u}_0(\bfm{p}_{u_0}) = \bfm{p}_{u_0}$ with $\bfm{p}_{u_0} \in V$. Then
$d\bfm{u}_0/d\bfm{p}_{u_0} = I_V$ and $d\bfm{v}_0/d\bfm{p}_{u_0} = 0$, so the $\bfm{u}_0$-block of Eq.~\eqref{apA:eq:gradient} identifies the gradient directly:
\begin{equation}\label{apA:eq:grad-u0}
    \frac{d\chi}{d\bfm{u}_0} = \bfm{M}^T\dot{\bfm{b}}(0) - \bfm{D}^T\bfm{b}(0).
\end{equation}
Analogously, if a subvector of $\bfm{p}$ equals $\bfm{v}_0$, then $d\bfm{v}_0/d\bfm{p}_{v_0} = I_V$, $d\bfm{u}_0/d\bfm{p}_{v_0} = 0$, and
\begin{equation}\label{apA:eq:grad-v0}
    \frac{d\chi}{d\bfm{v}_0} = -\bfm{M}^T\bfm{b}(0).
\end{equation}

\subsection{Integrating the adjoints forward in time}\label{apA:ss:forward-time}

The adjoint problem of Definition~\ref{apA:def:adjoint} is posed \emph{backward} in time: the data \eqref{apA:eq:tc1}--\eqref{apA:eq:tc2} are imposed at the final time and the ODE \eqref{apA:eq:adjoint-ode} is integrated from $T$ down to $0$.
Physical hardware, by contrast, runs forward.
The change of variable $t\mapsto T-t$ removes the mismatch.

Define the forward-time adjoint field
\begin{equation}\label{apA:fwd:a-def}
    \bfm{a}(t) := \bfm{b}(T-t),
    \qquad
    \dot{\bfm{a}}(t) = -\dot{\bfm{b}}(T-t),
    \qquad
    \ddot{\bfm{a}}(t) = \ddot{\bfm{b}}(T-t),
\end{equation}
so that odd derivatives flip sign under the reversal and even ones do not.
Evaluating the adjoint ODE~\eqref{apA:eq:adjoint-ode} at time $T-t$ and substituting~\eqref{apA:fwd:a-def} gives the forward-time adjoint equation
\begin{equation}\label{apA:fwd:adj}
    \bfm{M}^T\ddot{\bfm{a}}(t) + \bfm{D}^T\dot{\bfm{a}}(t) + \bfm{F}_{\bfm{u}}^T[\bfm{u}(T-t),\bfm{p}_F,T-t]\,\bfm{a}(t) + \mathbb{P}\theta_{\bfm{u}}[\mathbb{P}\bfm{u}(T-t)] = 0, \qquad t\in[0,T].
\end{equation}
Among the operators only the damping changes, $-\bfm{D}^T \to +\bfm{D}^T$, with mass and stiffness untouched. This sign flip is the entire algebraic effect of time reversal on the dissipative term, and it is the origin of the obstruction to a physical realization in dissipative systems studied in Appendix~\ref{a:second-order-nonlinear-physical}. The linearization and the source are now read along the reversed trajectory $\bfm{u}(T-t)$.

The terminal data are imposed at the reversed time $t=0$. With $\bfm{b}(T) = \bfm{a}(0)$ and $\dot{\bfm{b}}(T) = -\dot{\bfm{a}}(0)$, the terminal conditions~\eqref{apA:eq:tc1}--\eqref{apA:eq:tc2} become initial conditions,
\begin{align}
    \bfm{M}^T\bfm{a}(0) &= -\mathbb{P}\psi_{\dot{\bfm{u}}}[\mathbb{P}\bfm{u}(T),\mathbb{P}\dot{\bfm{u}}(T)], \label{apA:fwd:ic1}\\
    \bfm{M}^T\dot{\bfm{a}}(0) &= -\bfm{D}^T\bfm{a}(0) - \mathbb{P}\psi_{\bfm{u}}[\mathbb{P}\bfm{u}(T),\mathbb{P}\dot{\bfm{u}}(T)]. \label{apA:fwd:ic2}
\end{align}
Likewise $\bfm{b}(0) = \bfm{a}(T)$ and $\dot{\bfm{b}}(0) = -\dot{\bfm{a}}(T)$, so substituting $\bfm{b}(t) = \bfm{a}(T-t)$ throughout the gradient~\eqref{apA:eq:gradient} gives
\begin{multline}\label{apA:fwd:grad}
    \frac{d\chi}{d\bfm{p}} = \int_0^T \Bigl\langle \bfm{a}(T-t),\; \bfm{M}_{\bfm{p}_M}\ddot{\bfm{u}}(t) + \bfm{D}_{\bfm{p}_D}\dot{\bfm{u}}(t) + \bfm{F}_{\bfm{p}_F}[\bfm{u}(t),\bfm{p}_F,t] - \bfm{f}_{\bfm{p}_f}(t,\bfm{p})\Bigr\rangle\,dt \\
    - \Bigl\langle \bfm{M}^T\dot{\bfm{a}}(T) + \bfm{D}^T\bfm{a}(T),\; \frac{d\bfm{u}_0}{d\bfm{p}_{u_0}}\Bigr\rangle - \Bigl\langle \bfm{M}^T\bfm{a}(T),\; \frac{d\bfm{v}_0}{d\bfm{p}_{v_0}}\Bigr\rangle.
\end{multline}
Passing to forward time therefore does three things: it flips the sign of the damping operator, it reverses the trajectory argument $\bfm{u}(t)\mapsto\bfm{u}(T-t)$ feeding the linearization and the source, and it turns the terminal data at $t=T$ into initial data at $t=0$.

Eqs.~\eqref{apA:fwd:adj}--\eqref{apA:fwd:grad} are the adjoint problem and gradient formula quoted in the main text.

\section{Physical adjoint for nonlinear second-order systems: time reversal and the damping obstruction}\label{a:second-order-nonlinear-physical}

This appendix derives the nudged time-reversed dynamics quoted in Sec.~\ref{s:physical-adjoints}, in particular the gain-flipped perturbation equation~\eqref{eq:main-second-order-reversed-perturbation-gain}, and proves the nonlinear part of Theorem~\ref{thm:physical-adjoint}. Throughout we use the forward-time adjoint equation~\eqref{eq:adjoint-method-second-order-adjoint}, which propagates with a damping term $+\bfm{D}^T\dot{\bfm{a}}$, an internal-force linearization $\bfm{F}_{\bfm{u}}^T$ evaluated along the time-reversed trajectory $\bfm{u}(T-t)$, the source~\eqref{eq:adjoint-source}, and initial conditions~\eqref{eq:adjoint-method-second-order-adjoint-ic1}--\eqref{eq:adjoint-method-second-order-adjoint-ic2} set by the terminal-loss derivatives $\mathbb{P}\psi_{\bfm{u}}$ and $\mathbb{P}\psi_{\dot{\bfm{u}}}$.

\subsection{Time reversal requires removing or flipping the damping}\label{a:nonlinear-trm}

Unlike the linear case, in the nonlinear adjoint~\eqref{eq:adjoint-method-second-order-adjoint} the field $\bfm{a}$ propagates under $\bfm{F}_{\bfm{u}}^T$ evaluated along the \emph{reversed} trajectory $\bfm{u}(T-t)$. To access this linearization physically, we must first reproduce the reversed trajectory on the device. Define
\begin{equation}\label{a:eq:w-def}
    \bfm{w}(t) := \bfm{u}(T-t),
    \qquad
    \dot{\bfm{w}}(t) = -\dot{\bfm{u}}(T-t),
    \qquad
    \ddot{\bfm{w}}(t) = \ddot{\bfm{u}}(T-t).
\end{equation}
Evaluating the forward state equation~\eqref{eq:adjoint-method-second-order-state} at $T-t$ and using~\eqref{a:eq:w-def} gives the reversed state equation~\eqref{eq:main-second-order-reversed-state},
\begin{equation}\label{a:eq:reversed-state}
    \bfm{M}\ddot{\bfm{w}}(t) - \bfm{D}\dot{\bfm{w}}(t) + \bfm{F}\bigl[\bfm{w}(t),\bfm{p}_F,T-t\bigr] - \bfm{f}(T-t,\bfm{p}_f) = \bfm{0},
\end{equation}
with initial conditions
\begin{equation}\label{a:eq:reversed-ic}
    \bfm{w}(0) = \bfm{u}(T),
    \qquad
    \dot{\bfm{w}}(0) = -\dot{\bfm{u}}(T).
\end{equation}
The only difference from the forward equation~\eqref{eq:adjoint-method-second-order-state} is the sign of the damping term: time reversal sends $\bfm{D}\dot{\bfm{u}}\mapsto-\bfm{D}\dot{\bfm{w}}$, because the first derivative is odd under $t\mapsto T-t$ while the second derivative is even. Reproducing $\bfm{w}$ on the \emph{same} hardware therefore requires either $\bfm{D}=\bfm{0}$ or an experimental replacement $\bfm{D}\mapsto-\bfm{D}$ that turns damping into gain. We grant the latter possibility in the next subsection and show that it still fails to deliver the adjoint.

\subsection{The nudged response on the gain-flipped device}\label{a:nonlinear-nudge}

Suppose the gain flip $\bfm{D}\mapsto-\bfm{D}$ can be realized, so that~\eqref{a:eq:reversed-state} runs on the device and produces the reversed background $\bfm{w}$. On this same gain-flipped device we superimpose an infinitesimal nudge: a source in the direction of the adjoint source~\eqref{eq:adjoint-source} and a perturbation of the initial conditions in the directions of the adjoint terminal data~\eqref{eq:adjoint-method-second-order-adjoint-ic1}--\eqref{eq:adjoint-method-second-order-adjoint-ic2}. The nudged field $\bfm{w}^{\epsilon}$ obeys
\begin{equation}\label{a:eq:nudged-state}
    \bfm{M}\ddot{\bfm{w}}^{\epsilon}(t) - \bfm{D}\dot{\bfm{w}}^{\epsilon}(t) + \bfm{F}\bigl[\bfm{w}^{\epsilon}(t),\bfm{p}_F,T-t\bigr] - \bfm{f}(T-t,\bfm{p}_f) + \epsilon\,\mathbb{P}\theta_{\bfm{u}}\bigl[\mathbb{P}\bfm{u}(T-t)\bigr] = \bfm{0},
\end{equation}
with initial conditions
\begin{equation}\label{a:eq:nudged-ic}
    \bfm{w}^{\epsilon}(0) = \bfm{u}(T) + \epsilon\,\bfm{a}(0),
    \qquad
    \dot{\bfm{w}}^{\epsilon}(0) = -\dot{\bfm{u}}(T) + \epsilon\,\dot{\bfm{a}}(0),
    \qquad |\epsilon|\ll1,
\end{equation}
where $\bfm{a}(0)$ and $\dot{\bfm{a}}(0)$ are obtained by solving the adjoint initial conditions~\eqref{eq:adjoint-method-second-order-adjoint-ic1}--\eqref{eq:adjoint-method-second-order-adjoint-ic2}:
\begin{align}
    \bfm{a}(0) &= -\bfm{M}^{-1}\mathbb{P}\psi_{\dot{\bfm{u}}}, \label{a:eq:nudged-ic-a0}\\
    \dot{\bfm{a}}(0) &= -\bfm{M}^{-1}\bigl[\bfm{D}^T\bfm{a}(0) + \mathbb{P}\psi_{\bfm{u}}\bigr]
    = \bfm{M}^{-1}\bigl[\bfm{D}^T\bfm{M}^{-1}\mathbb{P}\psi_{\dot{\bfm{u}}} - \mathbb{P}\psi_{\bfm{u}}\bigr], \label{a:eq:nudged-ic-adot0}
\end{align}
where we set $\bfm{M}^T=\bfm{M}$.
The velocity nudge carries the damping contribution $\bfm{D}^T\bfm{a}(0)$ demanded by~\eqref{eq:adjoint-method-second-order-adjoint-ic2}, which drops out when $\bfm{D}=\bfm{0}$.
Write the differential response
\begin{equation}\label{a:eq:delta-w}
    \delta\bfm{w}(t) := \bfm{w}^{\epsilon}(t) - \bfm{w}(t) = \mathcal{O}(\epsilon).
\end{equation}
Expanding the internal force to first order,
\begin{equation}\label{a:eq:F-expansion}
    \bfm{F}\bigl[\bfm{w}^{\epsilon},\bfm{p}_F,T-t\bigr]
    =
    \bfm{F}\bigl[\bfm{w},\bfm{p}_F,T-t\bigr]
    +
    \bfm{F}_{\bfm{u}}\bigl[\bfm{w}(t),\bfm{p}_F,T-t\bigr]\,\delta\bfm{w}(t)
    +
    \mathcal{O}\bigl(\|\delta\bfm{w}\|^2\bigr),
\end{equation}
and subtracting the reversed state equation~\eqref{a:eq:reversed-state} from~\eqref{a:eq:nudged-state} gives the gain-flipped perturbation equation~\eqref{eq:main-second-order-reversed-perturbation-gain},
\begin{equation}\label{a:eq:perturbation-gain}
    \bfm{M}\delta\ddot{\bfm{w}}(t) - \bfm{D}\delta\dot{\bfm{w}}(t) + \bfm{F}_{\bfm{u}}\bigl[\bfm{w}(t),\bfm{p}_F,T-t\bigr]\delta\bfm{w}(t) + \epsilon\,\mathbb{P}\theta_{\bfm{u}}\bigl[\mathbb{P}\bfm{u}(T-t)\bigr] + \mathcal{O}(\epsilon^2) = \bfm{0},
\end{equation}
with initial conditions obtained by subtracting~\eqref{a:eq:reversed-ic} from~\eqref{a:eq:nudged-ic},
\begin{equation}\label{a:eq:perturbation-ic}
    \delta\bfm{w}(0) = \epsilon\,\bfm{a}(0),
    \qquad
    \delta\dot{\bfm{w}}(0) = \epsilon\,\dot{\bfm{a}}(0).
\end{equation}
Because $\bfm{w}(t)=\bfm{u}(T-t)$, the linearization $\bfm{F}_{\bfm{u}}[\bfm{w}(t),\bfm{p}_F,T-t]$ in~\eqref{a:eq:perturbation-gain} is precisely the time-reversed Jacobian $\bfm{F}_{\bfm{u}}[\bfm{u}(T-t),\bfm{p}_F,T-t]$ appearing in the adjoint equation~\eqref{eq:adjoint-method-second-order-adjoint}.

\subsection{Comparison with the adjoint equation: the damping obstruction}\label{a:nonlinear-comparison}

Introduce the rescaled response $\widetilde{\bfm{a}}(t):=\delta\bfm{w}(t)/\epsilon$, which is $\mathcal{O}(1)$. Dividing~\eqref{a:eq:perturbation-gain}--\eqref{a:eq:perturbation-ic} by $\epsilon$ and letting $\epsilon\to0$ gives
\begin{equation}\label{a:eq:rescaled-perturbation}
    \bfm{M}\ddot{\widetilde{\bfm{a}}}(t) - \bfm{D}\dot{\widetilde{\bfm{a}}}(t) + \bfm{F}_{\bfm{u}}\bigl[\bfm{u}(T-t),\bfm{p}_F,T-t\bigr]\widetilde{\bfm{a}}(t) + \mathbb{P}\theta_{\bfm{u}}\bigl[\mathbb{P}\bfm{u}(T-t)\bigr] = \bfm{0},
\end{equation}
with $\widetilde{\bfm{a}}(0) = \bfm{a}(0)$ and $\dot{\widetilde{\bfm{a}}}(0) = \dot{\bfm{a}}(0)$, given by~\eqref{a:eq:nudged-ic-a0}--\eqref{a:eq:nudged-ic-adot0}. We compare this term by term with the forward-time adjoint equation~\eqref{eq:adjoint-method-second-order-adjoint},
\begin{equation}\label{a:eq:adjoint-recall}
    \bfm{M}^T\ddot{\bfm{a}}(t) + \bfm{D}^T\dot{\bfm{a}}(t) + \bfm{F}_{\bfm{u}}^T\bigl[\bfm{u}(T-t),\bfm{p}_F,T-t\bigr]\bfm{a}(t) + \mathbb{P}\theta_{\bfm{u}}\bigl[\mathbb{P}\bfm{u}(T-t)\bigr] = \bfm{0},
\end{equation}
whose initial conditions~\eqref{eq:adjoint-method-second-order-adjoint-ic1}--\eqref{eq:adjoint-method-second-order-adjoint-ic2} the nudge~\eqref{a:eq:nudged-ic} reproduces by construction. Matching the propagators in~\eqref{a:eq:rescaled-perturbation} and~\eqref{a:eq:adjoint-recall} requires:
\begin{itemize}
    \item \textbf{Mass.} $\bfm{M}^T=\bfm{M}$, so the inertial operators agree.
    \item \textbf{Stiffness.} $\bfm{F}_{\bfm{u}}^T=\bfm{F}_{\bfm{u}}$ along the reversed trajectory, so the two internal-force linearizations agree, i.e.\ the tangent dynamics is reciprocal.
    \item \textbf{Damping.} The rescaled response carries $-\bfm{D}\dot{\widetilde{\bfm{a}}}$, whereas the adjoint carries $+\bfm{D}^T\dot{\bfm{a}}$. For reciprocal damping $\bfm{D}^T=\bfm{D}$ the two terms differ by $2\bfm{D}\dot{\bfm{a}}$, and they coincide \emph{only} when $\bfm{D}=\bfm{0}$.
\end{itemize}
Even after granting the gain flip needed to reproduce the reversed background, the damping enters the nudged perturbation with the sign $-\bfm{D}$, the opposite of the $+\bfm{D}^T$ demanded by the adjoint. A discrepancy in the propagator cannot be repaired by any choice of source or initial condition, so the same hardware cannot simultaneously carry the reversed background and propagate the correct adjoint perturbation around it unless the damping vanishes.

\subsection{The undamped reciprocal case}\label{a:nonlinear-undamped}

Set $\bfm{D}=\bfm{0}$ and assume $\bfm{M}^T=\bfm{M}$ together with $\bfm{F}_{\bfm{u}}^T=\bfm{F}_{\bfm{u}}$ along the trajectory. Then the reversed background~\eqref{a:eq:reversed-state} runs on the same hardware, and the rescaled perturbation equation~\eqref{a:eq:rescaled-perturbation} coincides with the adjoint equation~\eqref{a:eq:adjoint-recall}, whose initial conditions the nudge already carries. Here~\eqref{a:eq:nudged-ic-a0}--\eqref{a:eq:nudged-ic-adot0} reduce to $\bfm{a}(0)=-\bfm{M}^{-1}\mathbb{P}\psi_{\dot{\bfm{u}}}$ and $\dot{\bfm{a}}(0)=-\bfm{M}^{-1}\mathbb{P}\psi_{\bfm{u}}$, the nudge quoted in the main text. By uniqueness of solutions, $\widetilde{\bfm{a}}=\bfm{a}$, i.e.
\begin{equation}\label{a:eq:adjoint-from-difference}
    \bfm{a}(t) = \lim_{\epsilon\to0}\frac{\bfm{w}^{\epsilon}(t)-\bfm{w}(t)}{\epsilon},
    \qquad \bfm{w}(t)=\bfm{u}(T-t).
\end{equation}
The adjoint field is therefore the infinitesimal differential response of the nudged time-reversed trajectory, and the gradient follows from the overlap integrals~\eqref{eq:adjoint-method-gradient-pM}--\eqref{eq:adjoint-method-gradient-v0}. This proves the nonlinear part of Theorem~\ref{thm:physical-adjoint}. \qed

\section{Mass-proportional damping}\label{a:mass-proportional-damping}

The obstruction established in Appendix~\ref{a:second-order-nonlinear-physical} forbids damping in the nonlinear case because the nudged response on the gain-flipped device propagates with $-\bfm{D}\dot{\widetilde{\bfm{a}}}$, while the adjoint requires $+\bfm{D}^T\dot{\bfm{a}}$. The special case of \emph{mass-proportional} damping reconciles the two. The construction still flips damping to gain, so the adjoint propagates on a device other than the one that ran the forward pass, which places it outside the \emph{same-device} constructions of this article. The mechanism is instructive nonetheless.

Mass-proportional damping means
\begin{equation}\label{a:eq:mass-proportional}
    \bfm{D} = \gamma\,\bfm{M},
    \qquad
    \gamma\in\mathbb{R}\ \text{a known scalar}.
\end{equation}
Because the damping operator is now a scalar multiple of the inertia, a single scalar exponential weighting in time interchanges the damped and anti-damped equations while leaving $\bfm{M}$ and the stiffness untouched. We assume, as in Appendix~\ref{a:second-order-nonlinear-physical}, reciprocity along the reversed trajectory, $\bfm{M}^T=\bfm{M}$ and $\bfm{F}_{\bfm{u}}^T=\bfm{F}_{\bfm{u}}$. Note that~\eqref{a:eq:mass-proportional} makes $\bfm{D}$ automatically reciprocal, $\bfm{D}^T=\gamma\bfm{M}^T=\bfm{D}$.

\subsection{Exponential weighting interchanges damping and gain}\label{a:mp-weighting}

With $\bfm{D}=\gamma\bfm{M}$ the forward-time adjoint equation~\eqref{eq:adjoint-method-second-order-adjoint} reads
\begin{equation}\label{a:eq:mp-adjoint}
    \bfm{M}\ddot{\bfm{a}}(t) + \gamma\bfm{M}\dot{\bfm{a}}(t) + \bfm{F}_{\bfm{u}}\bigl[\bfm{u}(T-t),\bfm{p}_F,T-t\bigr]\bfm{a}(t) + \mathbb{P}\theta_{\bfm{u}}\bigl[\mathbb{P}\bfm{u}(T-t)\bigr] = \bfm{0},
\end{equation}
i.e.\ a \emph{damped} linear oscillator for $\bfm{a}$. The gain-flipped device of Appendix~\ref{a:second-order-nonlinear-physical}, by contrast, can only realize the \emph{anti-damped} perturbation dynamics
\begin{equation}\label{a:eq:mp-gain}
    \bfm{M}\ddot{\widetilde{\bfm{a}}}(t) - \gamma\bfm{M}\dot{\widetilde{\bfm{a}}}(t) + \bfm{F}_{\bfm{u}}\bigl[\bfm{u}(T-t),\bfm{p}_F,T-t\bigr]\widetilde{\bfm{a}}(t) + \boldsymbol{\sigma}(t) = \bfm{0},
\end{equation}
where $\boldsymbol{\sigma}$ is the source we are free to apply. Substitute the exponential ansatz
\begin{equation}\label{a:eq:mp-ansatz}
    \bfm{a}(t) = e^{-\gamma t}\,\widetilde{\bfm{a}}(t)
\end{equation}
into~\eqref{a:eq:mp-adjoint}, using $\dot{\bfm{a}} = e^{-\gamma t}(\dot{\widetilde{\bfm{a}}}-\gamma\widetilde{\bfm{a}})$ and $\ddot{\bfm{a}} = e^{-\gamma t}(\ddot{\widetilde{\bfm{a}}}-2\gamma\dot{\widetilde{\bfm{a}}}+\gamma^2\widetilde{\bfm{a}})$. The $\gamma^2\bfm{M}$ contributions cancel and the coefficient of $\dot{\widetilde{\bfm{a}}}$ flips from $+\gamma$ to $-\gamma$, leaving
\begin{equation}\label{a:eq:mp-substituted}
    e^{-\gamma t}\Bigl[\bfm{M}\ddot{\widetilde{\bfm{a}}} - \gamma\bfm{M}\dot{\widetilde{\bfm{a}}} + \bfm{F}_{\bfm{u}}\widetilde{\bfm{a}}\Bigr] + \mathbb{P}\theta_{\bfm{u}}\bigl[\mathbb{P}\bfm{u}(T-t)\bigr] = \bfm{0}.
\end{equation}
Multiplying by $e^{\gamma t}$ shows that $\widetilde{\bfm{a}}$ obeys the anti-damped equation~\eqref{a:eq:mp-gain}, provided the applied source is the exponentially weighted adjoint source
\begin{equation}\label{a:eq:mp-source}
    \boldsymbol{\sigma}(t) = e^{\gamma t}\,\mathbb{P}\theta_{\bfm{u}}\bigl[\mathbb{P}\bfm{u}(T-t)\bigr].
\end{equation}
The mass and stiffness operators in~\eqref{a:eq:mp-gain} are identical to those of the device. The gain flip reverses the sign of the damping, and the source carries the known scalar weight $e^{\gamma t}$. For an objective with only a terminal loss ($\theta\equiv0$), the source vanishes and no weighting of the drive is needed.

\subsection{Initial conditions, recovery, and the physical algorithm}\label{a:mp-algorithm}

The ansatz~\eqref{a:eq:mp-ansatz} also maps the adjoint initial conditions~\eqref{eq:adjoint-method-second-order-adjoint-ic1}--\eqref{eq:adjoint-method-second-order-adjoint-ic2} (with $\bfm{D}=\gamma\bfm{M}$) onto conditions of the undamped form for $\widetilde{\bfm{a}}$. From $\bfm{a}(0)=\widetilde{\bfm{a}}(0)$ and $\dot{\bfm{a}}(0)=\dot{\widetilde{\bfm{a}}}(0)-\gamma\widetilde{\bfm{a}}(0)$, the terms proportional to $\gamma$ cancel and
\begin{equation}\label{a:eq:mp-ic}
    \bfm{M}\widetilde{\bfm{a}}(0) = -\mathbb{P}\psi_{\dot{\bfm{u}}},
    \qquad
    \bfm{M}\dot{\widetilde{\bfm{a}}}(0) = -\mathbb{P}\psi_{\bfm{u}},
\end{equation}
i.e.\ the same nudge directions as in the undamped case of Appendix~\ref{a:nonlinear-undamped}. The physically measured field $\widetilde{\bfm{a}}$ is therefore obtained as in the undamped nonlinear construction, with the single modification that the adjoint source is applied with the weight $e^{\gamma t}$. The true adjoint is recovered by undoing the weighting,
\begin{equation}\label{a:eq:mp-recover}
    \bfm{a}(t) = e^{-\gamma t}\,\widetilde{\bfm{a}}(t)
    = e^{-\gamma t}\lim_{\epsilon\to0}\frac{\bfm{w}^{\epsilon}(t)-\bfm{w}(t)}{\epsilon},
    \qquad \bfm{w}(t)=\bfm{u}(T-t),
\end{equation}
before substitution into the gradient overlaps~\eqref{eq:adjoint-method-gradient-pM}--\eqref{eq:adjoint-method-gradient-v0}. Equivalently, the weight can be carried directly into the overlaps: since $\bfm{a}(T-t)=e^{-\gamma(T-t)}\widetilde{\bfm{a}}(T-t)$, each integrand acquires the known scalar factor $e^{-\gamma(T-t)}$, e.g.
\begin{equation}\label{a:eq:mp-weighted-gradient}
    \frac{d\chi}{d\bfm{p}_F}
    =
    \int_0^T e^{-\gamma(T-t)}\,
    \bigl\langle
        \widetilde{\bfm{a}}(T-t),\,
        \bfm{F}_{\bfm{p}_F}\bigl[\bfm{u}(t),\bfm{p}_F,t\bigr]
    \bigr\rangle\,dt,
\end{equation}
which realizes the exponential time weighting referenced in Sec.~\ref{s:physical-adjoints}. Concretely:
\begin{enumerate}
    \item Run the forward dynamics and record $\bfm{u}(t)$, then form $\bfm{w}(t)=\bfm{u}(T-t)$.
    \item Flip the damping to gain, $\bfm{D}=\gamma\bfm{M}\mapsto-\gamma\bfm{M}$, so the device reproduces the reversed background $\bfm{w}$.
    \item On the same gain-flipped device, apply the TRM and nudge together: initialize with $\bfm{w}^{\epsilon}(0)=\bfm{u}(T)-\epsilon\bfm{M}^{-1}\mathbb{P}\psi_{\dot{\bfm{u}}}$ and $\dot{\bfm{w}}^{\epsilon}(0)=-\dot{\bfm{u}}(T)-\epsilon\bfm{M}^{-1}\mathbb{P}\psi_{\bfm{u}}$, drive with $\bfm{f}(T-t)-\epsilon\,e^{\gamma t}\mathbb{P}\theta_{\bfm{u}}[\mathbb{P}\bfm{u}(T-t)]$, and reverse the explicit time dependence of $\bfm{F}$.
    \item Measure $\widetilde{\bfm{a}}(t)=\lim_{\epsilon\to0}[\bfm{w}^{\epsilon}(t)-\bfm{w}(t)]/\epsilon$ and form the adjoint $\bfm{a}(t)=e^{-\gamma t}\widetilde{\bfm{a}}(t)$.
    \item Evaluate the gradient from~\eqref{eq:adjoint-method-gradient-pM}--\eqref{eq:adjoint-method-gradient-v0}, equivalently~\eqref{a:eq:mp-weighted-gradient}.
\end{enumerate}

Two remarks. First, the construction requires $\gamma$ to be a known \emph{scalar}: only then is $e^{-\gamma t}$ a scalar weight that commutes with $\bfm{M}$ and $\bfm{F}_{\bfm{u}}$ and leaves the stiffness untouched in~\eqref{a:eq:mp-substituted}. A damping operator $\bfm{D}$ that is not a scalar multiple of $\bfm{M}$ would require a matrix exponential $e^{-\bfm{M}^{-1}\bfm{D}\,t}$ that does not, in general, commute with $\bfm{F}_{\bfm{u}}$, and the obstruction of Appendix~\ref{a:second-order-nonlinear-physical} returns. Second, as in the undamped case, this still requires the gain flip to reproduce the reversed background. Mass-proportional damping only makes the subsequent adjoint extraction exact, through the known exponential reweighting~\cite{pourcel2025lagrangian}.

\section{Onsager reciprocity}\label{a:onsager}

Section~\ref{s:symmetry-adjoints} identifies Onsager reciprocity as a special case of the general intertwining construction. Here we derive its second- and first-order real-valued forms. The physical field is the weighted adjoint $\bfm{d}=\bfm{V}\bfm{a}$. Whether the gradient requires explicit recovery of $\bfm{a}=\bfm{V}^{-1}\bfm{d}$ depends on how the trainable parameters enter the dynamics.

\subsection{Second-order real-valued systems}\label{a:onsager-second-order}

\subsubsection{Assumptions}\label{a:onsager-assumptions}

Consider the second-order system of Sec.~\ref{s:problem-statement}, with constant mass and damping operators. Let $\bfm{V}=\bfm{V}^T\succ0$ be a constant Onsager mobility. We assume that all three forward operators are $\bfm{V}$-reciprocal,
\begin{equation}\label{a:eq:onsager-operators}
    \bfm{V}\bfm{M}^T\bfm{V}^{-1}=\bfm{M},
    \qquad
    \bfm{V}\bfm{D}^T\bfm{V}^{-1}=\bfm{D},
    \qquad
    \bfm{V}\bfm{F}_{\bfm{u}}^T\bfm{V}^{-1}=\bfm{F}_{\bfm{u}},
\end{equation}
with the last condition evaluated along the (time-reversed) forward trajectory. In the common symmetric case $\bfm{M}^T=\bfm{M}$ and $\bfm{D}^T=\bfm{D}$, the first two conditions are the commutation relations $\bfm{M}\bfm{V}=\bfm{V}\bfm{M}$ and $\bfm{D}\bfm{V}=\bfm{V}\bfm{D}$.

The condition on $\bfm{F}_{\bfm{u}}$ is the linearized form of an \emph{Onsager-symmetric gradient flow}: if the internal force derives from a scalar energy through the mobility,
\begin{equation}\label{a:eq:onsager-force}
    \bfm{F}(\bfm{u},\bfm{p}_F)=\bfm{V}\,E_{\bfm{u}}(\bfm{u},\bfm{p}_F),
\end{equation}
then $\bfm{F}_{\bfm{u}}=\bfm{V}E_{\bfm{u}\bfm{u}}$ with the Hessian $E_{\bfm{u}\bfm{u}}$ symmetric, so $\bfm{V}\bfm{F}_{\bfm{u}}^T\bfm{V}^{-1}=\bfm{V}E_{\bfm{u}\bfm{u}}^T\bfm{V}^T\bfm{V}^{-1}=\bfm{V}E_{\bfm{u}\bfm{u}}=\bfm{F}_{\bfm{u}}$, as required. For a quadratic energy $E=\tfrac12\langle\bfm{u},\bfm{H}\bfm{u}\rangle$ with $\bfm{H}^T=\bfm{H}$ this is the linear Onsager-reciprocal case $\bfm{F}=\bfm{V}\bfm{H}\bfm{u}$. For nonquadratic $E$ the dynamics is nonlinear, but the mobility is still constant and reciprocal.

\subsubsection{The weighted adjoint}\label{a:onsager-substitution}

Recall the forward-time adjoint equation~\eqref{eq:adjoint-method-second-order-adjoint},
\begin{equation}\label{a:eq:onsager-adjoint-recall}
    \bfm{M}^T\ddot{\bfm{a}}(t) + \bfm{D}^T\dot{\bfm{a}}(t) + \bfm{F}_{\bfm{u}}^T\bigl[\bfm{u}(T-t),\bfm{p}_F,T-t\bigr]\bfm{a}(t) + \mathbb{P}\theta_{\bfm{u}}\bigl[\mathbb{P}\bfm{u}(T-t)\bigr] = \bfm{0}.
\end{equation}
Define the transformed field
\begin{equation}\label{a:eq:onsager-d-def}
    \bfm{d}(t):=\bfm{V}\bfm{a}(t),
    \qquad\text{equivalently}\qquad
    \bfm{a}(t)=\bfm{V}^{-1}\bfm{d}(t).
\end{equation}
Substituting $\bfm{a}=\bfm{V}^{-1}\bfm{d}$ into~\eqref{a:eq:onsager-adjoint-recall} and multiplying on the left by $\bfm{V}$ gives
\begin{equation}\label{a:eq:onsager-conjugated}
    \bigl(\bfm{V}\bfm{M}^T\bfm{V}^{-1}\bigr)\ddot{\bfm{d}}
    + \bigl(\bfm{V}\bfm{D}^T\bfm{V}^{-1}\bigr)\dot{\bfm{d}}
    + \bigl(\bfm{V}\bfm{F}_{\bfm{u}}^T\bfm{V}^{-1}\bigr)\bfm{d}
    + \bfm{V}\mathbb{P}\theta_{\bfm{u}}\bigl[\mathbb{P}\bfm{u}(T-t)\bigr] = \bfm{0}.
\end{equation}
By the $\bfm{V}$-reciprocity assumptions~\eqref{a:eq:onsager-operators}, the three conjugated operators collapse to the \emph{forward} operators $\bfm{M}$, $\bfm{D}$, $\bfm{F}_{\bfm{u}}$, leaving
\begin{equation}\label{a:eq:onsager-d-equation}
    \bfm{M}\ddot{\bfm{d}}(t) + \bfm{D}\dot{\bfm{d}}(t) + \bfm{F}_{\bfm{u}}\bigl[\bfm{u}(T-t),\bfm{p}_F,T-t\bigr]\bfm{d}(t) + \bfm{V}\mathbb{P}\theta_{\bfm{u}}\bigl[\mathbb{P}\bfm{u}(T-t)\bigr] = \bfm{0}.
\end{equation}
This is governed by the same linearized operator that the physical device realizes around the (time-reversed) forward trajectory, as in the Euclidean case $\bfm{V}=\bfm{I}$, but with the adjoint source now carrying the weight $\bfm{V}$.

The adjoint initial conditions~\eqref{eq:adjoint-method-second-order-adjoint-ic1}--\eqref{eq:adjoint-method-second-order-adjoint-ic2} transform the same way. Substituting $\bfm{a}=\bfm{V}^{-1}\bfm{d}$ and multiplying by $\bfm{V}$ gives
\begin{align}
    \bfm{M}\bfm{d}(0) &= -\,\bfm{V}\mathbb{P}\psi_{\dot{\bfm{u}}}\bigl[\mathbb{P}\bfm{u}(T),\mathbb{P}\dot{\bfm{u}}(T)\bigr], \label{a:eq:onsager-d-ic1}\\
    \bfm{M}\dot{\bfm{d}}(0) &= -\,\bfm{D}\bfm{d}(0) - \bfm{V}\mathbb{P}\psi_{\bfm{u}}\bigl[\mathbb{P}\bfm{u}(T),\mathbb{P}\dot{\bfm{u}}(T)\bigr], \label{a:eq:onsager-d-ic2}
\end{align}
again using~\eqref{a:eq:onsager-operators}. Thus $\bfm{d}$ solves the forward-operator problem~\eqref{a:eq:onsager-d-equation}--\eqref{a:eq:onsager-d-ic2}, with every adjoint source and initial-condition kick simply pre-multiplied by $\bfm{V}$.

\subsubsection{Physical realization and gradient recovery}\label{a:onsager-realization}

Eq.~\eqref{a:eq:onsager-d-equation} has the structure analyzed for $\bfm{V}=\bfm{I}$: it propagates under the forward $\bfm{M}$, $\bfm{D}$, $\bfm{F}_{\bfm{u}}$. The physical constructions therefore apply unchanged to $\bfm{d}$:
\begin{itemize}
    \item \textbf{Linear $\bfm{V}$-reciprocal} ($\bfm{F}=\bfm{V}\bfm{H}\bfm{u}$, with $\bfm{V}$-reciprocal $\bfm{M}$ and $\bfm{D}$): $\bfm{d}$ is obtained by a direct finite-amplitude run of the same hardware, driven by the $\bfm{V}$-weighted adjoint source and initialized by~\eqref{a:eq:onsager-d-ic1}--\eqref{a:eq:onsager-d-ic2}. As in the Euclidean case, $\bfm{V}$-reciprocal damping is permitted.
    \item \textbf{Nonlinear gradient flow} ($\bfm{F}=\bfm{V}E_{\bfm{u}}$, $\bfm{D}=\bfm{0}$): $\bfm{d}$ is obtained as the infinitesimal response of the nudged time-reversed trajectory (Appendix~\ref{a:second-order-nonlinear-physical}), now with the source and initial-condition nudges pre-multiplied by $\bfm{V}$. If instead $\bfm{D}=\gamma\bfm{M}$, the exponential reweighting of Appendix~\ref{a:mass-proportional-damping} also applies.
\end{itemize}
In both cases the measured physical field is $\bfm{d}=\bfm{V}\bfm{a}$. The Euclidean adjoint is
\begin{equation}\label{a:eq:onsager-recover}
    \bfm{a}(t)=\bfm{V}^{-1}\bfm{d}(t),
\end{equation}
so the force-parameter gradient can be written directly in terms of the measured field as
\begin{equation}\label{a:eq:onsager-weighted-gradient}
    \frac{d\chi}{d\bfm{p}_F}
    = \int_0^T \bigl\langle \bfm{V}^{-1}\bfm{d}(T-t),\,\bfm{F}_{\bfm{p}_F}\bigl[\bfm{u}(t),\bfm{p}_F,t\bigr]\bigr\rangle\,dt.
\end{equation}
If $\bfm{F}=\bfm{V}E_{\bfm{u}}$ and $\bfm{V}$ is independent of the parameters $\bfm{p}_F$ being differentiated, then $\bfm{F}_{\bfm{p}_F}=\bfm{V}E_{\bfm{u}\bfm{p}_F}$. Since $\bfm{V}$ is symmetric, the weighting cancels:
\begin{equation}\label{a:eq:onsager-gradient-cancellation}
\begin{split}
    \frac{d\chi}{d\bfm{p}_F}
    &=
    \int_0^T
    \bigl\langle
        \bfm{V}^{-1}\bfm{d}(T-t),\,
        \bfm{V}E_{\bfm{u}\bfm{p}_F}
        [\bfm{u}(t),\bfm{p}_F,t]
    \bigr\rangle\,dt\\
    &=
    \int_0^T
    \bigl\langle
        \bfm{d}(T-t),\,
        E_{\bfm{u}\bfm{p}_F}
        [\bfm{u}(t),\bfm{p}_F,t]
    \bigr\rangle\,dt.
\end{split}
\end{equation}
These gradients therefore do not require explicit recovery of $\bfm{a}$. The same cancellation applies to source parameters when $\bfm{f}=\bfm{V}\bfm{j}$ with fixed $\bfm{V}$, and to mass or damping parameters when their sensitivities carry the same left factor $\bfm{V}$. Other parameter sensitivities and initial-condition gradients retain the weights in the general formulas. If the differentiated parameter changes $\bfm{V}$ itself, its derivative must also be included. This distinction depends on the parameterization, not on the order of the dynamics.

This proves the second-order Onsager-reciprocal extension stated in Sec.~\ref{s:symmetry-adjoints} and generalizes Theorem~\ref{thm:physical-adjoint}.

\subsection{First-order real-valued systems}\label{a:onsager-first-order}

For the first-order system $\dot{\bfm{u}}+\bfm{F}[\bfm{u},\bfm{p}_F,t]-\bfm{f}=\bfm{0}$, assume $\bfm{V}\bfm{F}_{\bfm{u}}^T\bfm{V}^{-1}=\bfm{F}_{\bfm{u}}$ with constant $\bfm{V}=\bfm{V}^T\succ0$. Substituting $\bfm{d}=\bfm{V}\bfm{a}$ into the first-order adjoint~\eqref{a:eq:fo-adjoint-forward} and its initial condition gives
\begin{align}
    \dot{\bfm{d}}(t)
    +\bfm{F}_{\bfm{u}}
    [\bfm{u}(T-t),\bfm{p}_F,T-t]\bfm{d}(t)
    +\bfm{V}\mathbb{P}\theta_{\bfm{u}}
    [\mathbb{P}\bfm{u}(T-t)]
    &=\bfm{0},
    \label{a:eq:onsager-first-order-adjoint}\\
    \bfm{d}(0)
    &=-\bfm{V}\mathbb{P}\psi_{\bfm{u}}
    [\mathbb{P}\bfm{u}(T)].
    \label{a:eq:onsager-first-order-ic}
\end{align}
For linear $\bfm{F}=\bfm{K}(t)\bfm{u}$, the device propagates $\bfm{d}$ in one additional finite-amplitude run with applied drive $-\bfm{V}\mathbb{P}\theta_{\bfm{u}}$, initial condition~\eqref{a:eq:onsager-first-order-ic}, and any explicit schedule reversed. The force-gradient cancellation~\eqref{a:eq:onsager-gradient-cancellation} applies here as well.

For nonlinear trajectories, Onsager reciprocity does not remove the time-reversal obstruction of Sec.~\ref{ss:first-order-nonlinear}. At a linearly exponentially stable fixed point, the weighted nudge of Sec.~\ref{s:onsager-reciprocity} gives
\begin{equation}\label{a:eq:onsager-stationary-nudge}
    \bfm{F}[\bfm{u}^{*\epsilon},\bfm{p}_F]
    -\bfm{f}
    +\epsilon\bfm{V}\mathbb{P}\psi_{\bfm{u}}
    [\mathbb{P}\bfm{u}^{*\epsilon}]
    =\bfm{0}.
\end{equation}
Linearizing at the free fixed point and using $\bfm{V}$-reciprocity shows that
\begin{equation}\label{a:eq:onsager-stationary-recovery}
    \bfm{d}^*
    :=\lim_{\epsilon\to0}
    \frac{\bfm{u}^{*\epsilon}-\bfm{u}^*}{\epsilon}
    =\bfm{V}\bfm{a}^*,
\end{equation}
where $\bfm{a}^*$ is the static adjoint~\eqref{a:eq:fo-static-adjoint}. For $\bfm{F}=\bfm{V}E_{\bfm{u}}$, the static force-gradient overlap likewise reduces to $\langle\bfm{d}^*,E_{\bfm{u}\bfm{p}_F}\rangle$ when $\bfm{V}$ is independent of $\bfm{p}_F$.

For $\bfm{F}=\bfm{V}E_{\bfm{u}}$ with constant $\bfm{V}=\bfm{V}^T\succ0$ and constant drive $\bfm{f}$, the free and nudged relaxations obey $\dot{\bfm{u}}=-\bfm{V}(E^\epsilon)_{\bfm{u}}$ and descend the tilted energy in the $\bfm{V}^{-1}$ metric:
\begin{equation}\label{a:eq:onsager-metric-descent}
\begin{aligned}
    E^\epsilon(\bfm{u})
    &=E(\bfm{u})-\bfm{u}^T\bfm{V}^{-1}\bfm{f}
    +\epsilon\psi[\mathbb{P}\bfm{u}],\\
    \frac{dE^\epsilon}{dt}
    &=-(E^\epsilon)_{\bfm{u}}^T\bfm{V}(E^\epsilon)_{\bfm{u}}\leq0.
\end{aligned}
\end{equation}
A nondegenerate local minimum of $E^\epsilon$ is a linearly exponentially stable operating point. Both phases must still settle on the same stable branch.

\subsubsection{Example: resistor--capacitor networks with heterogeneous capacitances}\label{a:onsager-rc-network}

Let $\bfm{u}(t)$ denote the node voltages of a passive resistor network with external current injections $\bfm{j}(t)$. Kirchhoff's current law gives
\begin{equation}\label{a:eq:onsager-rc-state}
    \bfm{C}\dot{\bfm{u}}(t)
    +\bfm{G}\bfm{u}(t)
    -\bfm{j}(t)
    =\bfm{0},
    \qquad
    \bfm{C}=\bfm{C}^T\succ0,
    \qquad
    \bfm{G}=\bfm{G}^T,
\end{equation}
where $\bfm{C}$ is the diagonal matrix of node capacitances and $\bfm{G}$ is the conductance matrix. Define
\begin{equation}\label{a:eq:onsager-rc-normalized}
    \bfm{V}:=\bfm{C}^{-1},
    \qquad
    \bfm{K}:=\bfm{V}\bfm{G}.
\end{equation}
Then Eq.~\eqref{a:eq:onsager-rc-state} becomes $\dot{\bfm{u}}+\bfm{K}\bfm{u}-\bfm{V}\bfm{j}=\bfm{0}$. For heterogeneous capacitances, $\bfm{K}$ is generally not symmetric, but it is $\bfm{V}$-reciprocal:
\begin{equation}\label{a:eq:onsager-rc-reciprocity}
    \bfm{V}\bfm{K}^T\bfm{V}^{-1}
    =
    \bfm{V}\bfm{G}^T
    =
    \bfm{V}\bfm{G}
    =
    \bfm{K}.
\end{equation}

Write the distributed and terminal loss derivatives as
\begin{equation}
    \bfm{q}(t):=\mathbb{P}\theta_{\bfm{u}}
    \bigl[\mathbb{P}\bfm{u}(T-t)\bigr],
    \qquad
    \bfm{q}_T:=\mathbb{P}\psi_{\bfm{u}}
    \bigl[\mathbb{P}\bfm{u}(T)\bigr].
\end{equation}
The transformed adjoint $\bfm{d}=\bfm{V}\bfm{a}$ then obeys the circuit equations
\begin{equation}\label{a:eq:onsager-rc-adjoint}
    \bfm{C}\dot{\bfm{d}}(t)
    +\bfm{G}\bfm{d}(t)
    +\bfm{q}(t)
    =\bfm{0},
    \qquad
    \bfm{C}\bfm{d}(0)=-\bfm{q}_T.
\end{equation}
Thus the physical adjoint run uses the current $\bfm{j}^{(d)}(t)=-\bfm{q}(t)$ and an initial capacitor-charge kick $-\bfm{q}_T$. The circuit itself supplies the factors of $\bfm{V}$. The Euclidean adjoint is recovered as $\bfm{a}=\bfm{V}^{-1}\bfm{d}=\bfm{C}\bfm{d}$.

For the conductance gradients, hold the capacitances fixed and let $g_e$ be the conductance of an edge $e=(m,n)$. With the incidence vector $\boldsymbol{\ell}_e=\bfm{e}_m-\bfm{e}_n$ and $\bfm{G}_{g_e}=\boldsymbol{\ell}_e\boldsymbol{\ell}_e^T$, Eq.~\eqref{a:eq:onsager-gradient-cancellation} becomes
\begin{equation}\label{a:eq:onsager-rc-gradient}
\begin{split}
    \frac{d\chi}{dg_e}
    &=
    \int_0^T
    \bfm{a}(T-t)^T\bfm{V}\bfm{G}_{g_e}\bfm{u}(t)\,dt \\
    &=
    \int_0^T
    \bigl[d_m(T-t)-d_n(T-t)\bigr]
    \bigl[u_m(t)-u_n(t)\bigr]\,dt.
\end{split}
\end{equation}
Thus neither source preparation nor evaluation of the conductance gradients requires digital multiplication by $\bfm{V}$. Gradients with respect to the capacitances must be treated separately.

The same metric applies to reciprocal nonlinear circuit elements whose current--voltage laws derive from a co-content function $\mathcal{E}(\bfm{u})$~\cite{dillavou2024machine}:
\begin{equation}\label{a:eq:onsager-rc-nonlinear}
    \bfm{C}\dot{\bfm{u}}
    +\mathcal{E}_{\bfm{u}}(\bfm{u})
    -\bfm{j}
    =\bfm{0},
    \qquad
    \bfm{F}(\bfm{u})
    =\bfm{V}\mathcal{E}_{\bfm{u}}(\bfm{u}).
\end{equation}
Under the fixed-point assumptions above, the $\bfm{V}$-weighted nudge~\eqref{a:eq:onsager-stationary-nudge} is implemented by changing the applied current to $\bfm{j}^{\epsilon}=\bfm{j}-\epsilon\mathbb{P}\psi_{\bfm{u}}[\mathbb{P}\bfm{u}^{*,\epsilon}]$. This gives the Onsager-weighted Equilibrium Propagation construction. It does not provide a same-device adjoint for nonlinear first-order trajectories.

\section{Adjoint state equation and gradient for first-order systems}\label{a:first-order-adjoint}

We derive the first-order adjoint equation~\eqref{eq:main-first-order-adjoint} and the associated gradient formula quoted in Sec.~\ref{s:first-order}.
The derivation follows Appendix~\ref{apA} step for step: the Lagrangian is again a functional of independent arguments, and each of its partial derivatives supplies one ingredient of the method. Stationarity in the multipliers returns the forward problem, stationarity in the state defines the adjoint problem, and the remaining explicit parameter dependence is the gradient.
Two things are simpler here. The first-order problem carries a single initial condition, hence a single initial-condition multiplier, and its constraint integral needs only one integration by parts.
We reuse the differentiability assumptions of Appendix~\ref{apA}.

\subsection{Constrained optimization problem}\label{a:fo-ss-problem}

The forward dynamics is the first-order system~\eqref{eq:main-first-order-state},
\begin{equation}\label{a:eq:fo-state}
    \dot{\bfm{u}}(t) + \bfm{F}\bigl[\bfm{u}(t),\bfm{p}_F,t\bigr] - \bfm{f}(t,\bfm{p}_f) = \bfm{0},
    \qquad t\in[0,T],
    \qquad \bfm{u}(0)=\bfm{u}_0(\bfm{p}_{u_0}),
\end{equation}
with $\bfm{u}\in C^1([0,T],V)$.
The cost has trajectory and terminal parts. The terminal part depends on $\bfm{u}(T)$ alone, since the full state of a first-order system is the position itself:
\begin{equation}\label{a:eq:fo-cost}
    \chi[\bfm{u}] = \int_0^T \theta\bigl[\mathbb{P}\bfm{u}(t)\bigr]\,dt + \psi\bigl[\mathbb{P}\bfm{u}(T)\bigr].
\end{equation}
The cost evaluated on the forward solution defines the reduced cost $\chi(\bfm{p}) := \chi[\bfm{u}(\cdot;\bfm{p})]$, and the optimization problem is the following.
\begin{problem}[Constrained optimization of a first-order ODE]\label{a:prob:fo-opt}
Find
\begin{equation}\label{a:eq:fo-opt}
    \min_{\bfm{p}\in\mathbb{R}^M}\;\chi[\bfm{u}]
    \qquad\text{subject to}\qquad
    \begin{cases}
    \dot{\bfm{u}}(t) + \bfm{F}[\bfm{u}(t),\bfm{p}_F,t] - \bfm{f}(t,\bfm{p}_f) = \bfm{0}, & t\in[0,T],\\[2pt]
    \bfm{u}(0)=\bfm{u}_0(\bfm{p}_{u_0}).
    \end{cases}
\end{equation}
\end{problem}

\subsection{The Lagrangian}\label{a:fo-ss-lagrangian}

We enforce the state equation with a multiplier path $\bfm{b}\in C^1([0,T],V)$ and the initial condition with a multiplier vector $\boldsymbol{\lambda}\in V$, and define
\begin{equation}\label{a:eq:fo-lagrangian}
    \mathcal{L}[\bfm{u},\bfm{b},\boldsymbol{\lambda},\bfm{p}]
    := \int_0^T \bigl\langle \bfm{b}(t),\, \dot{\bfm{u}}(t) + \bfm{F}[\bfm{u}(t),\bfm{p}_F,t] - \bfm{f}(t,\bfm{p}_f) \bigr\rangle\,dt
    + \bigl\langle \boldsymbol{\lambda},\, \bfm{u}(0)-\bfm{u}_0(\bfm{p}_{u_0}) \bigr\rangle
    + \chi[\bfm{u}].
\end{equation}
As in Appendix~\ref{apA}, the four arguments of $\mathcal{L}$ are \emph{independent}: $\bfm{u}$ is an arbitrary path, \emph{not} assumed to solve the state equation, and $\bfm{p}$ enters only through the explicit dependencies $\bfm{F}[\cdot,\bfm{p}_F,\cdot]$, $\bfm{f}(\cdot,\bfm{p}_f)$, $\bfm{u}_0(\bfm{p}_{u_0})$.
We may therefore differentiate with respect to each argument separately, and the derivation consists of computing the four partial derivatives $\mathcal{L}_{\bfm{b}}$, $\mathcal{L}_{\boldsymbol{\lambda}}$, $\mathcal{L}_{\bfm{u}}$, $\mathcal{L}_{\bfm{p}}$ in turn.

\subsection{Stationarity with respect to the multipliers: the forward problem}\label{a:fo-ss-forward}

The Lagrangian is affine in $\bfm{b}$ and $\boldsymbol{\lambda}$, so both derivatives can be read off directly. Identifying them with elements of the respective spaces through the inner product,
\begin{equation}\label{a:eq:fo-La}
    \mathcal{L}_{\bfm{b}}(t) = \dot{\bfm{u}}(t) + \bfm{F}[\bfm{u}(t),\bfm{p}_F,t] - \bfm{f}(t,\bfm{p}_f),
    \qquad
    \mathcal{L}_{\boldsymbol{\lambda}} = \bfm{u}(0)-\bfm{u}_0(\bfm{p}_{u_0}),
\end{equation}
so $\mathcal{L}_{\bfm{b}} = 0$ holds if and only if $\bfm{u}$ satisfies the state equation~\eqref{a:eq:fo-state} on $[0,T]$, and $\mathcal{L}_{\boldsymbol{\lambda}} = 0$ if and only if $\bfm{u}$ satisfies the initial condition.
Stationarity with respect to a multiplier again recovers the constraint that the multiplier enforces. Together the two conditions state that $\bfm{u} = \bfm{u}(\cdot;\bfm{p})$ is the forward solution.

\subsection{Stationarity turns the partial derivative into the gradient}\label{a:fo-ss-key-identity}

Evaluated at the forward solution, the constraint terms in~\eqref{a:eq:fo-lagrangian} vanish, so $\chi(\bfm{p}) = \mathcal{L}[\bfm{u}(\cdot;\bfm{p}),\bfm{b},\boldsymbol{\lambda},\bfm{p}]$ for every $\bfm{b}$ and $\boldsymbol{\lambda}$.
Letting the multipliers also depend differentiably on $\bfm{p}$ and differentiating gives one term per argument of $\mathcal{L}$:
\begin{equation}\label{a:eq:fo-chain-rule}
    \frac{d\chi}{d\bfm{p}}
    = \mathcal{L}_{\bfm{u}}\frac{d\bfm{u}}{d\bfm{p}}
    + \mathcal{L}_{\bfm{b}}\frac{d\bfm{b}}{d\bfm{p}}
    + \mathcal{L}_{\boldsymbol{\lambda}}\frac{d\boldsymbol{\lambda}}{d\bfm{p}}
    + \mathcal{L}_{\bfm{p}}.
\end{equation}
The two multiplier terms vanish by Sec.~\ref{a:fo-ss-forward}, whatever the sensitivities $d\bfm{b}/d\bfm{p}$ and $d\boldsymbol{\lambda}/d\bfm{p}$ may be.
The expensive sensitivity $d\bfm{u}/d\bfm{p}$ is multiplied only by $\mathcal{L}_{\bfm{u}}$, while $\bfm{b}$ and $\boldsymbol{\lambda}$ are still free. Choosing them such that $\mathcal{L}_{\bfm{u}} = 0$ therefore leaves the identity
\begin{equation}\label{a:eq:fo-key-identity}
    \frac{d\chi}{d\bfm{p}} = \mathcal{L}_{\bfm{p}},
\end{equation}
in which only the cheap explicit parameter dependence of $\bfm{F}$, $\bfm{f}$ and $\bfm{u}_0$ appears.

\subsection{Stationarity with respect to the state: the adjoint problem}\label{a:fo-ss-adjoint}

Let $\bfm{h}\in C^1([0,T],V)$ be an arbitrary test path with no constraints at either endpoint. Varying $\bfm{u}\to\bfm{u}+\varepsilon\bfm{h}$, and noting that the drive $\bfm{f}$ has no $\bfm{u}$-dependence and does not contribute,
\begin{equation}\label{a:eq:fo-variation}
    \mathcal{L}_{\bfm{u}}\bfm{h}
    = \int_0^T \bigl\langle \bfm{b},\, \dot{\bfm{h}} + \bfm{F}_{\bfm{u}}\bfm{h}\bigr\rangle\,dt
    + \bigl\langle \boldsymbol{\lambda},\,\bfm{h}(0)\bigr\rangle
    + \int_0^T \bigl\langle \mathbb{P}\theta_{\bfm{u}},\,\bfm{h}\bigr\rangle\,dt
    + \bigl\langle \mathbb{P}\psi_{\bfm{u}},\,\bfm{h}(T)\bigr\rangle.
\end{equation}
The second-order case needed two integrations by parts; here, a single integration moves the derivative off $\bfm{h}$:
\begin{equation}\label{a:eq:fo-ibp}
    \int_0^T \langle \bfm{b},\dot{\bfm{h}}\rangle\,dt = -\int_0^T \langle \dot{\bfm{b}},\bfm{h}\rangle\,dt + \langle \bfm{b}(T),\bfm{h}(T)\rangle - \langle \bfm{b}(0),\bfm{h}(0)\rangle,
\end{equation}
and moving $\bfm{F}_{\bfm{u}}$ across the inner product gives $\langle \bfm{b},\bfm{F}_{\bfm{u}}\bfm{h}\rangle = \langle \bfm{F}_{\bfm{u}}^T\bfm{b},\bfm{h}\rangle$. Grouping by where $\bfm{h}$ is evaluated,
\begin{equation}\label{a:eq:fo-full-variation}
    \mathcal{L}_{\bfm{u}}\bfm{h} = \int_0^T \bigl\langle -\dot{\bfm{b}} + \bfm{F}_{\bfm{u}}^T\bfm{b} + \mathbb{P}\theta_{\bfm{u}},\,\bfm{h}\bigr\rangle\,dt 
    + \bigl\langle \bfm{b}(T) + \mathbb{P}\psi_{\bfm{u}},\,\bfm{h}(T)\bigr\rangle
    + \bigl\langle -\bfm{b}(0) + \boldsymbol{\lambda},\,\bfm{h}(0)\bigr\rangle.
\end{equation}
Since $\bfm{h}$ is arbitrary and unconstrained at both endpoints, $\mathcal{L}_{\bfm{u}} = 0$ requires each inner product to vanish separately.
The bulk term and the term at $t=T$ involve only $\bfm{b}$. Together they form a terminal-value problem that determines $\bfm{b}$ uniquely and defines the adjoint state.
\begin{definition}[Adjoint state for the first-order problem]\label{a:def:fo-adjoint}
The \emph{adjoint state} $\bfm{b}(t)\in V$ associated with Problem~\ref{a:prob:fo-opt} is the solution of the adjoint ODE
\begin{equation}\label{a:eq:fo-adjoint-backward}
    -\dot{\bfm{b}}(t) + \bfm{F}_{\bfm{u}}^T\bigl[\bfm{u}(t),\bfm{p}_F,t\bigr]\bfm{b}(t) + \mathbb{P}\theta_{\bfm{u}}\bigl[\mathbb{P}\bfm{u}(t)\bigr] = \bfm{0}, \qquad t\in[0,T],
\end{equation}
with the single required terminal condition
\begin{equation}\label{a:eq:fo-tc}
    \bfm{b}(T) = -\mathbb{P}\psi_{\bfm{u}}\bigl[\mathbb{P}\bfm{u}(T)\bigr].
\end{equation}
\end{definition}
We solve the adjoint ODE~\eqref{a:eq:fo-adjoint-backward} backward from $T$ to $0$ using this terminal condition.
This produces $\bfm{b}(0)$, which in turn determines the multiplier via the term at $t=0$ in~\eqref{a:eq:fo-full-variation},
\begin{equation}\label{a:eq:fo-lambda}
    \boldsymbol{\lambda} = \bfm{b}(0).
\end{equation}
The adjoint ODE is linear in $\bfm{b}$, its coefficients depend continuously on $\bfm{p}$, and~\eqref{a:eq:fo-lambda} is explicit. Hence $(\bfm{b},\boldsymbol{\lambda})$ is the unique solution of $\mathcal{L}_{\bfm{u}} = 0$ and depends smoothly on $\bfm{p}$, as required in Sec.~\ref{a:fo-ss-key-identity}.

\subsection{The gradient}\label{a:fo-ss-gradient}

Both remaining partial derivatives are now in place, and evaluating $\mathcal{L}_{\bfm{p}}$ gives the gradient.
\begin{theorem}[Adjoint gradient for the first-order problem]\label{a:thm:fo-gradient}
Let $\bfm{u} = \bfm{u}(\cdot;\bfm{p})$ solve the forward problem of Problem~\ref{a:prob:fo-opt} and let $\bfm{b}$ be the adjoint state of Definition~\ref{a:def:fo-adjoint}. Then the gradient of the reduced cost with respect to $\bfm{p}$ is
\begin{equation}\label{a:eq:fo-gradient-backward}
    \frac{d\chi}{d\bfm{p}} = \int_0^T \bigl\langle \bfm{b}(t),\, \bfm{F}_{\bfm{p}_F}[\bfm{u}(t),\bfm{p}_F,t] - \bfm{f}_{\bfm{p}_f}(t,\bfm{p}_f)\bigr\rangle\,dt - \bigl\langle \bfm{b}(0),\, \tfrac{d\bfm{u}_0}{d\bfm{p}_{u_0}}\bigr\rangle.
\end{equation}
\end{theorem}
\noindent\textit{Proof.}
With $(\bfm{b},\boldsymbol{\lambda})$ chosen as in Definition~\ref{a:def:fo-adjoint} and Eq.~\eqref{a:eq:fo-lambda}, we have $\mathcal{L}_{\bfm{u}} = 0$, and the key identity~\eqref{a:eq:fo-key-identity} gives $d\chi/d\bfm{p} = \mathcal{L}_{\bfm{p}}$.
The explicit $\bfm{p}$-dependence of $\mathcal{L}$ enters through $\bfm{F}[\bfm{u},\bfm{p}_F,t]$ and $\bfm{f}(t,\bfm{p}_f)$ in the constraint integral and through $\bfm{u}_0(\bfm{p}_{u_0})$ in the initial-condition term. The cost $\chi$ has none.
Differentiating each at fixed $\bfm{u}$, $\bfm{b}$, $\boldsymbol{\lambda}$ and substituting $\boldsymbol{\lambda} = \bfm{b}(0)$ gives~\eqref{a:eq:fo-gradient-backward}. \qed

As in the second-order case, every quantity on the right-hand side is known: $\bfm{u}(t)$ from the forward solve, $\bfm{b}(t)$ from the backward solve, and the parameter sensitivities from the model. The computational cost is one forward solve plus one backward solve, independent of the number of parameters $M$.

\paragraph{Direct optimization of the initial condition.}
Suppose a subvector of $\bfm{p}$ encodes the initial state itself, i.e.\ $\bfm{u}_0(\bfm{p}_{u_0}) = \bfm{p}_{u_0}$ with $\bfm{p}_{u_0}\in V$. Then $d\bfm{u}_0/d\bfm{p}_{u_0} = I_V$, and the $\bfm{u}_0$-block of Eq.~\eqref{a:eq:fo-gradient-backward} identifies the gradient directly:
\begin{equation}\label{a:eq:fo-grad-u0}
    \frac{d\chi}{d\bfm{u}_0} = -\bfm{b}(0),
\end{equation}
that is, minus the initial-condition multiplier, which is the $\bfm{M}=\bfm{0}$, $\bfm{D}=\bfm{I}$ case of Eq.~\eqref{apA:eq:grad-u0}.

\subsection{Integrating the adjoint forward in time}\label{a:fo-ss-forward-time}

The adjoint problem of Definition~\ref{a:def:fo-adjoint} is posed backward in time, while physical hardware runs forward. As in Sec.~\ref{apA:ss:forward-time}, the change of variable $t\mapsto T-t$ removes the mismatch.
Define the forward-time adjoint field
\begin{equation}\label{a:eq:fo-a-def}
    \bfm{a}(t) := \bfm{b}(T-t), \qquad \dot{\bfm{a}}(t) = -\dot{\bfm{b}}(T-t).
\end{equation}
Evaluating the adjoint ODE~\eqref{a:eq:fo-adjoint-backward} at time $T-t$ and substituting~\eqref{a:eq:fo-a-def} gives the forward-time adjoint equation
\begin{equation}\label{a:eq:fo-adjoint-forward}
    \dot{\bfm{a}}(t) + \bfm{F}_{\bfm{u}}^T\bigl[\bfm{u}(T-t),\bfm{p}_F,T-t\bigr]\bfm{a}(t) + \mathbb{P}\theta_{\bfm{u}}\bigl[\mathbb{P}\bfm{u}(T-t)\bigr] = \bfm{0}, \qquad t\in[0,T].
\end{equation}
Under the reversal the derivative term changes sign, $-\dot{\bfm{b}}\to+\dot{\bfm{a}}$. This is the $\bfm{M}=\bfm{0}$, $\bfm{D}=\bfm{I}$ instance of the flip $-\bfm{D}^T\to+\bfm{D}^T$ of Sec.~\ref{apA:ss:forward-time}, and it gives Eq.~\eqref{a:eq:fo-adjoint-forward} the same sign pattern as the state equation~\eqref{a:eq:fo-state}. Apart from this, the reversal only means that the linearization and the source are read along the reversed trajectory $\bfm{u}(T-t)$, with the explicit time dependence of $\bfm{F}$ replayed as $t\mapsto T-t$.
The terminal condition~\eqref{a:eq:fo-tc} becomes an initial condition at the reversed time $t=0$:
\begin{equation}\label{a:eq:fo-ic-forward}
    \bfm{a}(0) = -\mathbb{P}\psi_{\bfm{u}}\bigl[\mathbb{P}\bfm{u}(T)\bigr].
\end{equation}
Likewise $\bfm{b}(0) = \bfm{a}(T)$, so substituting $\bfm{b}(t) = \bfm{a}(T-t)$ throughout the gradient~\eqref{a:eq:fo-gradient-backward} gives
\begin{equation}\label{a:eq:fo-gradient-forward}
    \frac{d\chi}{d\bfm{p}} = \int_0^T \bigl\langle \bfm{a}(T-t),\, \bfm{F}_{\bfm{p}_F}[\bfm{u}(t),\bfm{p}_F,t] - \bfm{f}_{\bfm{p}_f}(t,\bfm{p}_f)\bigr\rangle\,dt - \bigl\langle \bfm{a}(T),\, \tfrac{d\bfm{u}_0}{d\bfm{p}_{u_0}}\bigr\rangle.
\end{equation}
Eqs.~\eqref{a:eq:fo-adjoint-forward} and~\eqref{a:eq:fo-ic-forward} are the adjoint equation~\eqref{eq:main-first-order-adjoint} with the source and initial condition~\eqref{eq:adj-src-IC-first-order} quoted in the main text.
They are likewise the first-order analogues of the second-order results of Appendix~\ref{apA}, and follow formally from them by the overdamped reduction $\bfm{M}=\bfm{0}$, $\bfm{D}=\bfm{I}$.

\section{Wirtinger derivation of the Schr\"odinger adjoint}\label{a:schrodinger-wirtinger}

We derive the Wirtinger adjoint state, gradient, and forward-time conjugated field used in Sec.~\ref{s:schrodinger} from an augmented Lagrangian in which $\bfm{u}$ and $\overline{\bfm{u}}$ are treated as independent variables~\cite{bosch2025adjoint}. The derivation parallels the real first-order case of Appendix~\ref{a:first-order-adjoint}. The only new ingredient is Wirtinger calculus \cite{koor2023short}.
Wirtinger variables are natural here because real-valued objectives of complex fields are generically non-holomorphic: the variation of the cost depends on both $\delta \bfm{u}$ and $\delta \overline{\bfm{u}}$.
This viewpoint is standard in complex optimization, where one treats a complex variable and its conjugate as formally independent variables to obtain stationarity conditions and descent directions \cite{Brandwood1983,Hjorungnes2007,Sorber2012}.
For time-dependent complex dynamics, the same variational structure appears in quantum optimal control, where the Schr\"odinger equation is imposed by a complex Lagrange multiplier and stationarity of the augmented functional yields a terminal-value adjoint equation \cite{Werschnik2007}. 
A rigorous PDE-constrained optimization formulation in complex Banach spaces, including semigroup Schr\"odinger equations and first- and second-order optimality conditions, was developed in \cite{Aronna2019}. 
Our derivation follows this complex Lagrangian/Wirtinger viewpoint, but is specialized to the adjoint field and forward-adjoint overlap gradients needed for the physical backpropagation constructions in Sec.~\ref{s:schrodinger}.

\subsection{Problem and conventions}\label{a:schr-problem-sub}

The state obeys the Schr\"odinger-type system~\eqref{eq:main-schrodinger-state},
\begin{equation}\label{a:eq:schr-state}
    i\dot{\bfm{u}}(t) + \bfm{F}\bigl[\bfm{u}(t),\overline{\bfm{u}}(t),\bfm{p}_F,t\bigr] - \bfm{f}(t,\bfm{p}_f) = \bfm{0},
    \qquad \bfm{u}(0)=\bfm{u}_0(\bfm{p}_{u_0}),
\end{equation}
with $\bfm{u}\in\mathbb{C}^N$, real parameters $\bfm{p}$, and a real-valued cost
\begin{equation}\label{a:eq:schr-cost}
    \chi[\bfm{u}] = \int_0^T \theta\bigl[\mathbb{P}\bfm{u}(t),\mathbb{P}\overline{\bfm{u}}(t)\bigr]\,dt + \psi\bigl[\mathbb{P}\bfm{u}(T),\mathbb{P}\overline{\bfm{u}}(T)\bigr].
\end{equation}
We use the Hermitian inner product $\langle\bfm{x},\bfm{y}\rangle_{\mathbb{C}}=\overline{\bfm{x}}^T\bfm{y}$, treat $\bfm{u}$ and $\overline{\bfm{u}}$ as independent, and write the two Wirtinger Jacobians $\bfm{F}_{\bfm{u}}$ and $\bfm{F}_{\overline{\bfm{u}}}$. Since $\theta$ and $\psi$ are real, $\theta_{\bfm{u}}=\overline{\theta_{\overline{\bfm{u}}}}$ and $\psi_{\bfm{u}}=\overline{\psi_{\overline{\bfm{u}}}}$, and, for an arbitrary variation $\bfm{h}$ of $\bfm{u}$, the first variation of the (real) cost is
\begin{equation}\label{a:eq:schr-loss-variation}
    \delta\chi = 2\operatorname{Re}\Bigl[\int_0^T\bigl\langle\mathbb{P}\theta_{\overline{\bfm{u}}},\bfm{h}\bigr\rangle_{\mathbb{C}}\,dt + \bigl\langle\mathbb{P}\psi_{\overline{\bfm{u}}},\bfm{h}(T)\bigr\rangle_{\mathbb{C}}\Bigr].
\end{equation}

\begin{problem}[Constrained optimization of a Schr\"odinger-type system]\label{a:prob:schr}
Find
\begin{equation}\label{a:eq:schr-opt}
    \min_{\bfm{p}}\;\chi[\bfm{u}] \qquad\text{subject to}\qquad
    \begin{cases}
    i\dot{\bfm{u}}(t) + \bfm{F}[\bfm{u}(t),\overline{\bfm{u}}(t),\bfm{p}_F,t] - \bfm{f}(t,\bfm{p}_f) = \bfm{0}, & t\in[0,T],\\[2pt]
    \bfm{u}(0)=\bfm{u}_0(\bfm{p}_{u_0}).
    \end{cases}
\end{equation}
\end{problem}

\subsection{Augmented Lagrangian and the Wirtinger adjoint}\label{a:schr-lagrangian-sub}

Introduce a complex adjoint field $\bfm{b}\in C^1([0,T],\mathbb{C}^N)$ and an initial-condition multiplier $\boldsymbol{\lambda}\in\mathbb{C}^N$, and pair the complex constraint against them through $2\operatorname{Re}$ so that the Lagrangian is real,
\begin{equation}\label{a:eq:schr-lagrangian}
    \mathcal{L} = \chi + 2\operatorname{Re}\int_0^T \bigl\langle \bfm{b}(t),\, i\dot{\bfm{u}} + \bfm{F} - \bfm{f}\bigr\rangle_{\mathbb{C}}\,dt + 2\operatorname{Re}\bigl\langle\boldsymbol{\lambda},\,\bfm{u}(0)-\bfm{u}_0\bigr\rangle_{\mathbb{C}}.
\end{equation}
The Lagrangian is affine in $\bfm{b}$ and $\boldsymbol{\lambda}$, which enter only through the real pairing, so those two derivatives can be read off as the residuals they multiply,
\begin{equation}\label{a:eq:schr-Lb}
    \mathcal{L}_{\bfm{b}}(t) = i\dot{\bfm{u}}(t) + \bfm{F}\bigl[\bfm{u}(t),\overline{\bfm{u}}(t),\bfm{p}_F,t\bigr] - \bfm{f}(t,\bfm{p}_f),
    \qquad
    \mathcal{L}_{\boldsymbol{\lambda}} = \bfm{u}(0)-\bfm{u}_0(\bfm{p}_{u_0}).
\end{equation}
A real pairing $2\operatorname{Re}\langle\,\cdot\,,\bfm{X}\rangle_{\mathbb{C}}$ that vanishes for every complex direction forces $\bfm{X}=\bfm{0}$, since a direction and its multiple by $i$ recover the real and imaginary parts separately. Hence $\mathcal{L}_{\bfm{b}} = 0$ holds if and only if $\bfm{u}$ satisfies the state equation~\eqref{a:eq:schr-state} on $[0,T]$, and $\mathcal{L}_{\boldsymbol{\lambda}} = 0$ if and only if it satisfies the initial condition. Stationarity in the multipliers returns the forward problem, as in the real case.

Evaluated on that forward solution, the two constraint pairings vanish and $\mathcal{L} = \chi(\bfm{p})$. Letting the multipliers also depend differentiably on $\bfm{p}$ and differentiating gives one term per argument of $\mathcal{L}$, with each product denoting the real pairing of a derivative with a sensitivity:
\begin{equation}\label{a:eq:schr-chain-rule}
    \frac{d\chi}{d\bfm{p}} = \mathcal{L}_{\bfm{u}}\frac{d\bfm{u}}{d\bfm{p}} + \mathcal{L}_{\bfm{b}}\frac{d\bfm{b}}{d\bfm{p}} + \mathcal{L}_{\boldsymbol{\lambda}}\frac{d\boldsymbol{\lambda}}{d\bfm{p}} + \mathcal{L}_{\bfm{p}}.
\end{equation}
Here $\mathcal{L}_{\bfm{u}}$ is the state derivative: its real pairing with a variation $\bfm{h}$ gives the first variation computed below. Because that pairing is taken under $2\operatorname{Re}$, it includes both the $\bfm{u}$ and the $\overline{\bfm{u}}$ contributions, and no separate condition in $\overline{\bfm{u}}$ is needed. The two multiplier terms vanish by the previous paragraph, whatever the sensitivities $d\bfm{b}/d\bfm{p}$ and $d\boldsymbol{\lambda}/d\bfm{p}$ may be, and the expensive sensitivity $d\bfm{u}/d\bfm{p}$ is multiplied only by $\mathcal{L}_{\bfm{u}}$. Choosing $\bfm{b}$ and $\boldsymbol{\lambda}$ so that $\mathcal{L}_{\bfm{u}} = 0$ therefore leaves the key identity
\begin{equation}\label{a:eq:schr-key-identity}
    \frac{d\chi}{d\bfm{p}} = \mathcal{L}_{\bfm{p}}.
\end{equation}

To impose stationarity of $\mathcal{L}$, we vary $\bfm{u}\to\bfm{u}+\varepsilon\bfm{h}$ (with the accompanying conjugate variation $\overline{\bfm{h}}$). The first variation $\delta\mathcal{L}$ is the coefficient of $\varepsilon$. The linearized residual is $i\dot{\bfm{h}} + \bfm{F}_{\bfm{u}}\bfm{h} + \bfm{F}_{\overline{\bfm{u}}}\overline{\bfm{h}}$. Combining this with the cost variation~\eqref{a:eq:schr-loss-variation} gives
\begin{equation}\label{a:eq:schr-variation}
    \delta\mathcal{L} = 2\operatorname{Re}\int_0^T\bigl\langle\bfm{b},\, i\dot{\bfm{h}} + \bfm{F}_{\bfm{u}}\bfm{h} + \bfm{F}_{\overline{\bfm{u}}}\overline{\bfm{h}}\bigr\rangle_{\mathbb{C}}\,dt + 2\operatorname{Re}\bigl\langle\boldsymbol{\lambda},\bfm{h}(0)\bigr\rangle_{\mathbb{C}} + \delta\chi.
\end{equation}
To read off the stationarity conditions we move every term onto $\bfm{h}$. The $\overline{\bfm{h}}$ contribution is folded onto $\bfm{h}$ under the real part,
\begin{equation}\label{a:eq:schr-B-identity}
    2\operatorname{Re}\bigl\langle\bfm{b},\bfm{B}\,\overline{\bfm{h}}\bigr\rangle_{\mathbb{C}} = 2\operatorname{Re}\bigl\langle\bfm{B}^T\overline{\bfm{b}},\bfm{h}\bigr\rangle_{\mathbb{C}},
\end{equation}
applied with $\bfm{B}=\bfm{F}_{\overline{\bfm{u}}}$, while the holomorphic term moves through the Hermitian adjoint, $2\operatorname{Re}\langle\bfm{b},\bfm{F}_{\bfm{u}}\bfm{h}\rangle_{\mathbb{C}}=2\operatorname{Re}\langle\overline{\bfm{F}_{\bfm{u}}}^T\bfm{b},\bfm{h}\rangle_{\mathbb{C}}$. The only time derivative is integrated by parts. Writing $\langle\bfm{b},i\dot{\bfm{h}}\rangle_{\mathbb{C}}=\langle-i\bfm{b},\dot{\bfm{h}}\rangle_{\mathbb{C}}$,
\begin{equation}\label{a:eq:schr-ibp}
    2\operatorname{Re}\int_0^T\langle\bfm{b},i\dot{\bfm{h}}\rangle_{\mathbb{C}}\,dt
    = 2\operatorname{Re}\int_0^T\langle i\dot{\bfm{b}},\bfm{h}\rangle_{\mathbb{C}}\,dt
    + 2\operatorname{Re}\Bigl[\langle-i\bfm{b}(T),\bfm{h}(T)\rangle_{\mathbb{C}} - \langle-i\bfm{b}(0),\bfm{h}(0)\rangle_{\mathbb{C}}\Bigr].
\end{equation}
Every surviving factor of $i$ originates from this step. Substituting~\eqref{a:eq:schr-loss-variation}--\eqref{a:eq:schr-ibp} into~\eqref{a:eq:schr-variation} and grouping by where $\bfm{h}$ is evaluated,
\begin{multline}\label{a:eq:schr-full-variation}
    \delta\mathcal{L} = 2\operatorname{Re}\int_0^T\bigl\langle i\dot{\bfm{b}} + \overline{\bfm{F}_{\bfm{u}}}^T\bfm{b} + \bfm{F}_{\overline{\bfm{u}}}^T\overline{\bfm{b}} + \mathbb{P}\theta_{\overline{\bfm{u}}},\,\bfm{h}\bigr\rangle_{\mathbb{C}}\,dt \\
    + 2\operatorname{Re}\bigl\langle -i\bfm{b}(T) + \mathbb{P}\psi_{\overline{\bfm{u}}},\,\bfm{h}(T)\bigr\rangle_{\mathbb{C}}
    + 2\operatorname{Re}\bigl\langle i\bfm{b}(0) + \boldsymbol{\lambda},\,\bfm{h}(0)\bigr\rangle_{\mathbb{C}}.
\end{multline}
Since $\bfm{h}$ is arbitrary and unconstrained at both endpoints, $\mathcal{L}_{\bfm{u}} = 0$ requires each bracket to vanish separately. The $t=0$ term fixes the multiplier $\boldsymbol{\lambda}=-i\bfm{b}(0)$, while the bulk and $t=T$ terms define the adjoint state.

\begingroup
\mdfsetup{nobreak=true}
\begin{definition}[Wirtinger adjoint state]\label{a:def:schr-adjoint}
The \emph{Wirtinger adjoint state} $\bfm{b}(t)\in\mathbb{C}^N$ associated with Problem~\ref{a:prob:schr} is the solution of the backward-time ODE
\begin{equation}\label{a:eq:schr-adjoint-backward}
    i\dot{\bfm{b}}(t) + \overline{\bfm{F}_{\bfm{u}}}^T\,\bfm{b}(t) + \bfm{F}_{\overline{\bfm{u}}}^T\,\overline{\bfm{b}}(t) + \mathbb{P}\theta_{\overline{\bfm{u}}}\bigl[\mathbb{P}\bfm{u}(t),\mathbb{P}\overline{\bfm{u}}(t)\bigr] = \bfm{0},
\end{equation}
with the Wirtinger Jacobians evaluated at $(\bfm{u}(t),\overline{\bfm{u}}(t))$, and the terminal condition
\begin{equation}\label{a:eq:schr-adjoint-tc}
    \bfm{b}(T) = -\,i\,\mathbb{P}\psi_{\overline{\bfm{u}}}\bigl[\mathbb{P}\bfm{u}(T),\mathbb{P}\overline{\bfm{u}}(T)\bigr].
\end{equation}
\end{definition}
\endgroup

\begingroup
\mdfsetup{nobreak=true}
\begin{theorem}[Wirtinger adjoint gradient]\label{a:thm:schr-gradient}
Let $\bfm{u}$ solve Problem~\ref{a:prob:schr} and $\bfm{b}$ be the adjoint state of Definition~\ref{a:def:schr-adjoint}. For real parameters $\bfm{p}$, the gradient of the cost is
\begin{equation}\label{a:eq:schr-gradient-a}
    \frac{d\chi}{d\bfm{p}} = 2\operatorname{Re}\int_0^T\bigl\langle\bfm{b}(t),\,\bfm{F}_{\bfm{p}_F}[\bfm{u}(t),\overline{\bfm{u}}(t),\bfm{p}_F,t]-\bfm{f}_{\bfm{p}_f}(t,\bfm{p}_f)\bigr\rangle_{\mathbb{C}}\,dt + 2\operatorname{Re}\bigl\langle i\,\bfm{b}(0),\,\bfm{u}_{0,\bfm{p}_{u_0}}\bigr\rangle_{\mathbb{C}}.
\end{equation}
\end{theorem}
\endgroup
\noindent\textit{Proof.} With $\bfm{b}$ and $\boldsymbol{\lambda}$ chosen as in Definition~\ref{a:def:schr-adjoint}, we have $\mathcal{L}_{\bfm{u}} = 0$, and the key identity~\eqref{a:eq:schr-key-identity} gives $d\chi/d\bfm{p}=\mathcal{L}_{\bfm{p}}$. The explicit parameter dependence of~\eqref{a:eq:schr-lagrangian} enters only through $\bfm{F}(\bfm{p}_F)$ and $\bfm{f}(\bfm{p}_f)$ in the constraint pairing and through $\bfm{u}_0(\bfm{p}_{u_0})$ in the initial-condition pairing, giving $\mathcal{L}_{\bfm{p}} = 2\operatorname{Re}\int_0^T\langle\bfm{b},\bfm{F}_{\bfm{p}_F}-\bfm{f}_{\bfm{p}_f}\rangle_{\mathbb{C}}\,dt - 2\operatorname{Re}\langle\boldsymbol{\lambda},\bfm{u}_{0,\bfm{p}_{u_0}}\rangle_{\mathbb{C}}$. Substituting $\boldsymbol{\lambda}=-i\bfm{b}(0)$ yields~\eqref{a:eq:schr-gradient-a}. The external drive enters the residual as $-\bfm{f}$, so its gradient carries the opposite sign to the internal force, and the initial-condition term carries the explicit factor $i$ inherited from $\boldsymbol{\lambda}=-i\bfm{b}(0)$. \qed

\subsection{Forward-time conjugated adjoint}\label{a:schr-forward-sub}

The adjoint of Definition~\ref{a:def:schr-adjoint} runs backward. For a physical realization, we use the conjugated time-reversed field $\bfm{c}(t):=\overline{\bfm{b}(T-t)}$. Evaluating~\eqref{a:eq:schr-adjoint-backward} at $T-t$, using $\dot{\bfm{b}}(T-t)=-\overline{\dot{\bfm{c}}(t)}$, and complex conjugating the entire equation gives the forward-time equation~\eqref{eq:main-schrodinger-adjoint},
\begin{equation}\label{a:eq:schr-c-forward}
\begin{split}
    i\dot{\bfm{c}}(t)
    &+ \bfm{F}_{\bfm{u}}^T\bigl[\bfm{u}(T-t),\overline{\bfm{u}}(T-t),\bfm{p}_F,T-t\bigr]\bfm{c}(t) \\
    &+ \overline{\bfm{F}_{\overline{\bfm{u}}}\bigl[\bfm{u}(T-t),\overline{\bfm{u}}(T-t),\bfm{p}_F,T-t\bigr]}^T\,\overline{\bfm{c}}(t)
    + \mathbb{P}\theta_{\bfm{u}}\bigl[\mathbb{P}\bfm{u}(T-t),\mathbb{P}\overline{\bfm{u}}(T-t)\bigr] = \bfm{0},
\end{split}
\end{equation}
where we used $\overline{\theta_{\overline{\bfm{u}}}}=\theta_{\bfm{u}}$, with initial condition
\begin{equation}\label{a:eq:schr-c-ic}
    \bfm{c}(0) = \overline{\bfm{b}(T)} = i\,\mathbb{P}\psi_{\bfm{u}}\bigl[\mathbb{P}\bfm{u}(T),\mathbb{P}\overline{\bfm{u}}(T)\bigr].
\end{equation}
Conjugating term by term flips the signs in front of the Wirtinger Jacobians, which is why the conjugated field $\bfm{c}$, rather than $\bfm{b}$, can be propagated on the same hardware. The raw adjoint is recovered by $\bfm{b}(t)=\overline{\bfm{c}(T-t)}$. Substituting into~\eqref{a:eq:schr-gradient-a} gives the gradient overlaps quoted in Sec.~\ref{s:schrodinger} in terms of the measured field,
\begin{align}
    \frac{d\chi}{d\bfm{p}_F} &= 2\operatorname{Re}\int_0^T\bigl\langle\overline{\bfm{c}(T-t)},\,\bfm{F}_{\bfm{p}_F}[\bfm{u}(t),\overline{\bfm{u}}(t),\bfm{p}_F,t]\bigr\rangle_{\mathbb{C}}\,dt,\label{a:eq:schr-grad-pF}\\
    \frac{d\chi}{d\bfm{p}_f} &= -\,2\operatorname{Re}\int_0^T\bigl\langle\overline{\bfm{c}(T-t)},\,\bfm{f}_{\bfm{p}_f}(t,\bfm{p}_f)\bigr\rangle_{\mathbb{C}}\,dt,\label{a:eq:schr-grad-pf}\\
    \frac{d\chi}{d\bfm{p}_{u_0}} &= 2\operatorname{Re}\bigl\langle i\,\overline{\bfm{c}(T)},\,\bfm{u}_{0,\bfm{p}_{u_0}}\bigr\rangle_{\mathbb{C}}.\label{a:eq:schr-grad-pu0}
\end{align}

\subsection{Proof of the nonlinear intertwined-adjoint construction}\label{a:nonlinear-intertwined-proof}

We prove the nonlinear part of Theorem~\ref{thm:intertwined-reciprocity}. Let $\tau$ denote time in the reversed experiment and set $r=T-\tau$. Use the loss data $\bfm{q}(\tau)$ and $\bfm{q}_T$ from Eq.~\eqref{eq:intertwiner-loss-data}, and abbreviate the forward-trajectory Jacobians by
\begin{equation}\label{eq:nonlinear-intertwiner-jacobians}
    \bfm{A}(r)
    :=
    \bfm{F}_{\bfm{u}}
    \bigl[
        \bfm{u}(r),\overline{\bfm{u}(r)},\bfm{p}_F,r
    \bigr],
    \qquad
    \bfm{B}(r)
    :=
    \bfm{F}_{\overline{\bfm{u}}}
    \bigl[
        \bfm{u}(r),\overline{\bfm{u}(r)},\bfm{p}_F,r
    \bigr].
\end{equation}
The forward-time conjugated adjoint~\eqref{eq:main-schrodinger-adjoint} is
\begin{equation}\label{eq:nonlinear-intertwiner-adjoint}
\begin{split}
    i\dot{\bfm{c}}
    +
    \bfm{A}(r)^T\bfm{c}
    +
    \overline{\bfm{B}(r)}^{\,T}\overline{\bfm{c}}
    +
    \bfm{q}
    &=
    \bfm{0},\\
    \bfm{c}(0)&=i\bfm{q}_T.
\end{split}
\end{equation}

Set $\bfm{d}=\mathcal{Q}\bfm{S}\bfm{c}$. Multiplying Eq.~\eqref{eq:nonlinear-intertwiner-adjoint} by $\mathcal{Q}\bfm{S}$ and using Eqs.~\eqref{eq:S-condition-1} and~\eqref{eq:S-condition-2} gives
\begin{equation}\label{eq:nonlinear-intertwiner-transformed}
\begin{split}
    i\dot{\bfm{d}}
    &+
    \mathcal{Q}\overline{\bfm{A}(r)}\mathcal{Q}^{-1}\bfm{d}
    +
    \mathcal{Q}\overline{\bfm{B}(r)}
    \overline{\mathcal{Q}}^{-1}\overline{\bfm{d}}
    +
    \mathcal{Q}\bfm{S}\bfm{q}
    =
    \bfm{0},\\
    \bfm{d}(0)&=i\mathcal{Q}\bfm{S}\bfm{q}_T.
\end{split}
\end{equation}
No time-reversal identity has been used in this step. The matrix $\bfm{S}$ acts only on the adjoint tangent equation.

Condition~\eqref{eq:Q-background-condition} now supplies the TRM. Define
\begin{equation}\label{eq:Q-reversed-background}
    \bfm{w}(\tau)
    :=
    \mathcal{Q}\overline{\bfm{u}(r)}.
\end{equation}
Evaluating the forward equation at $r$, complex conjugating it, and multiplying by $\mathcal{Q}$ gives
\begin{equation}
\begin{split}
    i\dot{\bfm{w}}(\tau)
    &+
    \bfm{F}\bigl[
        \bfm{w}(\tau),\overline{\bfm{w}(\tau)},\bfm{p}_F,r
    \bigr]
    -
    \mathcal{Q}\overline{\bfm{f}(r,\bfm{p}_f)}
    =
    \bfm{0},\\
    \bfm{w}(0)&=\mathcal{Q}\overline{\bfm{u}(T)}.
\end{split}
\end{equation}

The two Wirtinger variables are independent. Differentiate Eq.~\eqref{eq:Q-background-condition} first with respect to $\overline{\bfm{u}}$ and then with respect to $\bfm{u}$. Evaluated along the forward trajectory, the two derivatives are
\begin{equation}\label{eq:Q-tangent-pullback}
\begin{split}
    \bfm{F}_{\bfm{u}}
    \bigl[
        \bfm{w}(\tau),\overline{\bfm{w}(\tau)},\bfm{p}_F,r
    \bigr]\mathcal{Q}
    &=
    \mathcal{Q}\overline{\bfm{A}(r)},\\
    \bfm{F}_{\overline{\bfm{u}}}
    \bigl[
        \bfm{w}(\tau),\overline{\bfm{w}(\tau)},\bfm{p}_F,r
    \bigr]\overline{\mathcal{Q}}
    &=
    \mathcal{Q}\overline{\bfm{B}(r)}.
\end{split}
\end{equation}
The force covariance must therefore hold on an operating domain rather than only at isolated points of one trajectory.

Run the same system with
\begin{equation}\label{eq:nonlinear-intertwiner-nudge}
\begin{split}
    i\dot{\bfm{w}}^{\epsilon}(\tau)
    &+
    \bfm{F}\bigl[
        \bfm{w}^{\epsilon}(\tau),
        \overline{\bfm{w}^{\epsilon}(\tau)},
        \bfm{p}_F,r
    \bigr]\\
    &-
    \left[
        \mathcal{Q}\overline{\bfm{f}(r,\bfm{p}_f)}
        -
        \epsilon\mathcal{Q}\bfm{S}\bfm{q}(\tau)
    \right]
    =
    \bfm{0},\\
    \bfm{w}^{\epsilon}(0)
    &=
    \bfm{w}(0)
    +
    \epsilon i\mathcal{Q}\bfm{S}\bfm{q}_T.
\end{split}
\end{equation}
The applied nudge is $-\epsilon\mathcal{Q}\bfm{S}\bfm{q}$ because the external drive enters the residual with a minus sign.

Let $\delta\bfm{w}:=\bfm{w}^{\epsilon}-\bfm{w}$. Subtract the unnudged equation and expand $\bfm{F}$ to first order. Continuous Wirtinger differentiability implies differentiable dependence on $\epsilon$ on every interval where the perturbed trajectory remains in the operating domain. Using Eq.~\eqref{eq:Q-tangent-pullback},
\begin{equation}
\begin{split}
    i\delta\dot{\bfm{w}}
    &+
    \mathcal{Q}\overline{\bfm{A}(r)}\mathcal{Q}^{-1}
    \delta\bfm{w}
    +
    \mathcal{Q}\overline{\bfm{B}(r)}
    \overline{\mathcal{Q}}^{-1}
    \overline{\delta\bfm{w}}
    +
    \epsilon\mathcal{Q}\bfm{S}\bfm{q}
    +
    o(\epsilon)
    =
    \bfm{0},\\
    \delta\bfm{w}(0)
    &=
    \epsilon i\mathcal{Q}\bfm{S}\bfm{q}_T.
\end{split}
\end{equation}
Divide by $\epsilon$ and take the limit. The resulting initial-value problem is Eq.~\eqref{eq:nonlinear-intertwiner-transformed}, so uniqueness gives
\begin{equation}\label{eq:nonlinear-intertwiner-recovery}
    \bfm{d}(\tau)
    =
    \lim_{\epsilon\to0}
    \frac{\bfm{w}^{\epsilon}(\tau)-\bfm{w}(\tau)}{\epsilon}
    =
    \mathcal{Q}\bfm{S}\bfm{c}(\tau),
    \qquad
    \bfm{c}(\tau)
    =
    \bfm{S}^{-1}\mathcal{Q}^{-1}\bfm{d}(\tau).
\end{equation}
Thus both transformations are undone, in reverse order, when the adjoint is recovered from the measured response. \qed

\section{Second-order systems: gyroscopic \texorpdfstring{$\mathcal{PT}$}{PT} mirrors}\label{a:gyroscopic-pt}

For an undamped gyroscopic system, the background mirror and tangent matching can also be treated separately. Let $r:=T-\tau$ and consider
\begin{equation}\label{eq:gyroscopic-system}
    \bfm{M}\ddot{\bfm{u}}
    +
    \bfm{G}\dot{\bfm{u}}
    +
    \bfm{F}[\bfm{u},\bfm{p}_F,t]
    =
    \bfm{f}(t),
    \qquad
    \bfm{M}^T=\bfm{M},
    \qquad
    \bfm{G}^T=-\bfm{G}.
\end{equation}
Gyroscopic terms describe Coriolis forces, Lorentz forces on charged lattices, and the dynamics of gyroscopic metamaterials~\cite{nash2015topological}. Let $\bfm{\Pi}=\bfm{\Pi}^T=\bfm{\Pi}^{-1}$ be a real spatial mirror. Suppose
\begin{equation}\label{eq:gyroscopic-trm}
    \bfm{\Pi}\bfm{M}\bfm{\Pi}
    =
    \bfm{M},
    \qquad
    \bfm{\Pi}\bfm{G}\bfm{\Pi}
    =
    -\bfm{G},
    \qquad
    \bfm{F}[\bfm{\Pi}\bfm{u}(t),\bfm{p}_F,t]
    =
    \bfm{\Pi}\bfm{F}[\bfm{u}(t),\bfm{p}_F,t]
\end{equation}
along the forward trajectory. Then
\begin{equation}
    \bfm{w}(\tau)=\bfm{\Pi}\bfm{u}(r),
    \qquad
    \bfm{w}(0)=\bfm{\Pi}\bfm{u}(T),
    \qquad
    \dot{\bfm{w}}(0)=-\bfm{\Pi}\dot{\bfm{u}}(T)
\end{equation}
is generated on the same hardware by reversing the explicit schedule and applying the drive $\bfm{\Pi}\bfm{f}(r)$.

Let
\[
    \bfm{J}(t)
    :=
    \bfm{F}_{\bfm{u}}[\bfm{u}(t),\bfm{p}_F,t],
    \qquad
    \widetilde{\bfm{J}}(t)
    :=
    \bfm{F}_{\bfm{u}}[\bfm{\Pi}\bfm{u}(t),\bfm{p}_F,t].
\]
With no additional tangent intertwiner, the required matching condition is
\begin{equation}\label{eq:gyroscopic-tangent}
    \bfm{\Pi}\widetilde{\bfm{J}}(t)\bfm{\Pi}
    =
    \bfm{J}(t)^T.
\end{equation}
If the force covariance in Eq.~\eqref{eq:gyroscopic-trm} holds throughout an invariant operating domain, differentiating it gives
$\widetilde{\bfm{J}}=\bfm{\Pi}\bfm{J}\bfm{\Pi}$, and Eq.~\eqref{eq:gyroscopic-tangent} reduces to $\bfm{J}^T=\bfm{J}$.

Define the adjoint source locally by
\[
    \bfm{f}_{\mathrm{adj}}(\tau)
    =
    \mathbb{P}\theta_{\bfm{u}}
    \bigl[\mathbb{P}\bfm{u}(r)\bigr].
\]
The nudged drive is
\[
    \bfm{f}^{\epsilon}(\tau)
    =
    \bfm{\Pi}\bfm{f}(r)
    -
    \epsilon\bfm{\Pi}\bfm{f}_{\mathrm{adj}}(\tau).
\]
The endpoint perturbations are
\[
    \bfm{w}^{\epsilon}(0)
    =
    \bfm{w}(0)+\epsilon\bfm{\Pi}\bfm{a}(0),
    \qquad
    \dot{\bfm{w}}^{\epsilon}(0)
    =
    \dot{\bfm{w}}(0)+\epsilon\bfm{\Pi}\dot{\bfm{a}}(0),
\]
where the gyroscopic adjoint data satisfy
\begin{equation}
    \bfm{M}\bfm{a}(0)
    =
    -\mathbb{P}\psi_{\dot{\bfm{u}}},
    \qquad
    \bfm{M}\dot{\bfm{a}}(0)
    =
    \bfm{G}\bfm{a}(0)-\mathbb{P}\psi_{\bfm{u}}.
\end{equation}
The loss derivatives are evaluated at the terminal forward state. The differential response is
\begin{equation}
    \bfm{d}(\tau)
    =
    \lim_{\epsilon\to0}
    \frac{\bfm{w}^{\epsilon}(\tau)-\bfm{w}(\tau)}{\epsilon}
    =
    \bfm{\Pi}\bfm{a}(\tau),
    \qquad
    \bfm{a}=\bfm{\Pi}\bfm{d}.
\end{equation}
This is a separate extension of the second-order adjoint derivation, not an application of the nonlinear branch of Theorem~\ref{thm:physical-adjoint}.\footnote{For the $\bfm{\Pi}$-only construction, a parity-odd velocity coupling $\bfm{N}$ must satisfy both $\bfm{\Pi}\bfm{N}\bfm{\Pi}=-\bfm{N}$ to reverse the background and $\bfm{\Pi}\bfm{N}^T\bfm{\Pi}=\bfm{N}$ to match the adjoint. Together these imply $\bfm{N}^T=-\bfm{N}$. This excludes dissipative balanced gain and loss from the $\bfm{\Pi}$-only construction~\cite{bender2013observation}; an additional intertwiner may admit special cases.}

An additional tangent intertwiner can be included by setting $\bfm{d}=\bfm{\Pi}\bfm{S}\bfm{a}$, where $\bfm{S}$ is known, real, invertible, and constant during the experiment. Under force covariance throughout an invariant operating domain, the following conditions must hold for all $t$ along the forward trajectory and throughout the visited trainable family:
\[
    \bfm{S}\bfm{M}^T\bfm{S}^{-1}=\bfm{M},
    \qquad
    \bfm{S}\bfm{G}^T\bfm{S}^{-1}=-\bfm{G},
    \qquad
    \bfm{S}\bfm{J}^T\bfm{S}^{-1}=\bfm{J}.
\]
The corresponding drive and endpoint perturbations are
\[
    \bfm{f}^{\epsilon}(\tau)
    =
    \bfm{\Pi}\bfm{f}(r)
    -
    \epsilon\bfm{\Pi}\bfm{S}\bfm{f}_{\mathrm{adj}}(\tau),
\]
\[
    \bfm{w}^{\epsilon}(0)
    =
    \bfm{w}(0)+\epsilon\bfm{\Pi}\bfm{S}\bfm{a}(0),
    \qquad
    \dot{\bfm{w}}^{\epsilon}(0)
    =
    \dot{\bfm{w}}(0)+\epsilon\bfm{\Pi}\bfm{S}\dot{\bfm{a}}(0).
\]
The adjoint is recovered as $\bfm{a}=\bfm{S}^{-1}\bfm{\Pi}\bfm{d}$.

\section{Free-space optics}\label{app:free-space-optics}

\paragraph{Free-space diffractive optics.}
In the scalar, monochromatic, paraxial, and thin-element approximation, a free-space diffractive optical system is a linear Schr\"odinger-type system of the form considered in Sec.~\ref{s:schrodinger}, with the propagation coordinate $z$ playing the role of time~\cite{longhi2009quantum}. Starting from the scalar Helmholtz equation and writing the field as a slowly varying envelope,
\begin{equation}\label{a:eq:fso-helmholtz}
\left(\nabla_\perp^2+\partial_z^2+k_0^2 n^2(\bfm{r}_\perp,z)\right)E(\bfm{r}_\perp,z)=0,
\qquad
E(\bfm{r}_\perp,z)=u(\bfm{r}_\perp,z)e^{ikz},
\qquad
k=k_0n_0,
\end{equation}
the paraxial approximation neglects $\partial_z^2u$ relative to $2ik\partial_z u$, giving
\begin{equation}\label{a:eq:fso-paraxial}
i\partial_z u+
\left[
\frac{1}{2k}\nabla_\perp^2
+k_0\Delta n(\bfm{r}_\perp,z)
\right]u=0.
\end{equation}
An ideal SLM plane at $z=z_j$ is modeled as a thin local transmission
\begin{equation}\label{a:eq:fso-thin}
u(z_j^+)=m_j(\bfm{r}_\perp)u(z_j^-),
\qquad
m_j(\bfm{r}_\perp)=\rho_j(\bfm{r}_\perp)e^{i\phi_j(\bfm{r}_\perp)}
=e^{i\alpha_j(\bfm{r}_\perp)},
\end{equation}
or equivalently as a delta-function potential, the form underlying split-step beam propagation~\cite{feit1978light},
\begin{equation}\label{a:eq:fso-delta}
i\partial_z u+
\left[
\frac{1}{2k}\nabla_\perp^2
+\sum_j \alpha_j(\bfm{r}_\perp)\delta(z-z_j)
\right]u=0.
\end{equation}
After transverse discretization this has the form treated in Sec.~\ref{s:schrodinger},
\begin{equation}\label{a:eq:fso-state}
i\dot{\bfm{u}}(z)+\bfm{F}(\bfm{u},\overline{\bfm{u}},\bfm{p}_F)=0,
\qquad
\bfm{F}(\bfm{u},\overline{\bfm{u}},\bfm{p}_F)=\bfm{K}(z;\bfm{p}_K)\bfm{u},
\end{equation}
with
\begin{equation}\label{a:eq:fso-K}
\bfm{K}(z;\bfm{p}_K)
=
\frac{1}{2k}\bfm{L}_\perp
+
\sum_j \bfm{A}_j(\bfm{p}_j)\delta(z-z_j),
\qquad
\bfm{D}_j
=
e^{i\bfm{A}_j}
=
\operatorname{diag}(m_{j1},\ldots,m_{jN}).
\end{equation}
The internal force~\eqref{a:eq:fso-state} is linear and \emph{holomorphic}: it does not depend on $\overline{\bfm{u}}$, so $\bfm{F}_{\overline{\bfm{u}}}=\bfm{0}$ and the second Wirtinger Jacobian drops out of every formula in Appendix~\ref{a:schrodinger-wirtinger}. Equivalently, the input--output map is the usual free-space/SLM cascade~\cite{lin2018all}
\begin{equation}\label{a:eq:fso-cascade}
\bfm{U}
=
\bfm{P}_m\bfm{D}_m\bfm{P}_{m-1}\bfm{D}_{m-1}
\cdots
\bfm{P}_1\bfm{D}_1\bfm{P}_0,
\qquad
\bfm{P}_\ell
=
\exp\!\left(
i\frac{\Delta z_\ell}{2k}\bfm{L}_\perp
\right).
\end{equation}
For a reciprocal scalar discretization,
\begin{equation}\label{a:eq:fso-reciprocity}
\bfm{L}_\perp^T=\bfm{L}_\perp,
\qquad
\bfm{A}_j^T=\bfm{A}_j,
\qquad
\bfm{D}_j^T=\bfm{D}_j.
\end{equation}
Reciprocity therefore requires \emph{complex symmetry}, not Hermiticity or unitarity. In particular, lossy or linear-gain pixels with $|m_{ji}|\neq 1$ are admissible even though generally
\begin{equation}\label{a:eq:fso-nonunitary}
\overline{\bfm{D}}_j^T\bfm{D}_j\neq \bfm{I}.
\end{equation}

Accordingly, the physical adjoint experiment is the $\bfm{c}$-adjoint run of Appendix~\ref{a:schrodinger-wirtinger}, not a direct integration of the backward Wirtinger adjoint $\bfm{b}$. In the linear reciprocal case the $\bfm{c}$-field obeys
\begin{equation}\label{a:eq:fso-c-forward}
i\dot{\bfm{c}}(z)
+
\bfm{K}^T(L-z;\bfm{p}_K)\bfm{c}(z)
+
\mathbb P\theta_{\bfm{u}}
\!\left[
\mathbb P\bfm{u}(L-z),
\mathbb P\overline{\bfm{u}}(L-z)
\right]
=0,
\end{equation}
with initial condition
\begin{equation}\label{a:eq:fso-c-ic}
\bfm{c}(0)
=
i\,\mathbb P\psi_{\bfm{u}}
\!\left[
\mathbb P\bfm{u}(L),
\mathbb P\overline{\bfm{u}}(L)
\right].
\end{equation}
Since $\bfm{K}^T=\bfm{K}$, this is implemented by the same reciprocal optical system traversed in the reverse direction, with the adjoint source and initial condition applied at the output side. The operator $\bfm{K}^T(L-z)$ places the modulator planes at $L-z_j$, in the element order realized by reverse traversal. The backward adjoint is then recovered as
\begin{equation}\label{a:eq:fso-recover}
\bfm{b}(z)=\overline{\bfm{c}(L-z)}.
\end{equation}
For parameters that enter only through the optical operator, the gradient of Theorem~\ref{a:thm:schr-gradient} becomes
\begin{equation}\label{a:eq:fso-gradient}
\frac{d\chi}{dp}
=
2\operatorname{Re}
\int_0^L
\left\langle
\overline{\bfm{c}(L-z)},
\bfm{K}_p(z;\bfm{p}_K)\bfm{u}(z)
\right\rangle_{\mathbb C}
\,dz.
\end{equation}

\paragraph{Pixel gradients.}
Each SLM pixel carries two trainable parameters, a phase and a log-amplitude,
\begin{equation}\label{a:eq:fso-pixel-param}
    m_{ji}=e^{\gamma_{ji}+i\phi_{ji}}=\rho_{ji}e^{i\phi_{ji}},
    \qquad
    \rho_{ji}=e^{\gamma_{ji}},
    \qquad
    \alpha_{ji}=\phi_{ji}-i\gamma_{ji},
\end{equation}
so that $\gamma_{ji}<0$ describes an attenuating pixel, $\gamma_{ji}>0$ a linear-gain pixel, and a phase-only device is the special case $\gamma_{ji}\equiv0$. Both gradients are obtained from the same pair of runs.
We use the cascade to derive the gradient because it exhibits the pixel dependence directly~\cite{zhou2020situ}. Letting $\bfm{E}_i$ be diagonal with a single nonzero entry at pixel $i$,
\begin{equation}\label{a:eq:fso-dU}
\frac{\partial\bfm{U}}{\partial \phi_{ji}}
=
\bfm{P}_m\cdots\bfm{P}_j\,\bigl(i\,m_{ji}\bfm{E}_i\bigr)\,\bfm{P}_{j-1}\cdots\bfm{P}_0,
\end{equation}
and since $\bfm{E}_i\bfm{D}_j=\bfm{D}_j\bfm{E}_i$ the phase and amplitude gradients are the two parts of one complex overlap at plane $j$,
\begin{align}
\frac{\partial\chi}{\partial \phi_{ji}} &= 2\operatorname{Re}\bigl[\overline{b_i}\,u_i\bigr]_{z_j^+} , \label{a:eq:fso-grad-phi}\\
\frac{\partial\chi}{\partial \gamma_{ji}} &= 2\operatorname{Im}\bigl[\overline{b_i}\,u_i\bigr]_{z_j^+} . \label{a:eq:fso-grad-gamma}
\end{align}
Both fields must be evaluated on the \emph{same} side of plane $j$, here the output side $z_j^+$. The assignment of the real and imaginary parts follows from the factor $-i$ in $\bfm{b}(L)=-i\,\mathbb{P}\psi_{\overline{\bfm{u}}}$ of Definition~\ref{a:def:schr-adjoint}, which rotates the overlap by a quarter turn.

\section{Details for stationary problems}\label{app:helmholtz-scattering}

This appendix gives the derivations behind Sec.~\ref{s:stationary}. The setting is the stationary problem
\begin{equation}\label{app:eq:helmholtz-state}
    \bfm{F}\bigl[\bfm{u},\overline{\bfm{u}},\bfm{p}_F\bigr]-\bfm{f}(\bfm{p}_f)=\bfm{0},
\end{equation}
with an invertible tangent operator at the solution. This includes static linear elasticity, linear resistor networks, and frequency-domain scattering at a fixed drive frequency. We keep the complex-valued notation throughout. Real problems are the special case treated in the last subsection. The loss is a real-valued function of the measured field,
\begin{equation}\label{app:eq:helmholtz-cost}
    \chi=\psi\bigl[\mathbb{P}\bfm{u},\mathbb{P}\overline{\bfm{u}}\bigr].
\end{equation}

\begingroup
\mdfsetup{nobreak=true}
\begin{problem}[Constrained optimization of a stationary system]\label{app:prob:helmholtz}
Find
\begin{equation}\label{app:eq:helmholtz-opt}
    \min_{\bfm{p}}\;\psi\bigl[\mathbb{P}\bfm{u},\mathbb{P}\overline{\bfm{u}}\bigr]
    \qquad\text{subject to}\qquad
    \bfm{F}\bigl[\bfm{u},\overline{\bfm{u}},\bfm{p}_F\bigr]=\bfm{f}(\bfm{p}_f),
\end{equation}
over the parameters $\bfm{p}=[\bfm{p}_F,\bfm{p}_\bfm{f}]\in\mathbb{R}^M$.
\end{problem}
\endgroup

The derivation follows the preceding appendices: the field equation is imposed as a constraint with an adjoint multiplier. Stationarity with respect to the multiplier returns the field equation, while stationarity with respect to the field yields the adjoint equation. The remaining explicit parameter dependence gives the gradient as a forward--adjoint overlap. For frequency-domain problems this is the standard adjoint-variable sensitivity construction used for Maxwell and Helmholtz systems~\cite{VeronisDutton2004}.

\subsection{The stationary adjoint}\label{app:helmholtz-nonlinear}

There is no time, hence no integration by parts and no boundary term, and the single algebraic constraint carries a single multiplier. Pairing the residual through $2\operatorname{Re}$ gives
\begin{equation}\label{app:eq:lagrangian-helmholtz}
    \mathcal{L}[\bfm{u}, \bar{\bfm{u}}, \bfm{a}, \bar{\bfm{a}},\bfm{p}] = \chi + 2 \operatorname{Re} \bigl \langle \bfm{a}, \bfm{F}\bigl[\bfm{u},\overline{\bfm{u}},\bfm{p}_F\bigr] - \bfm{f}(\bfm{p}_f) \bigr\rangle_{\mathbb{C}},
\end{equation}
a real-valued function of independent arguments in which $\bfm{u}$ is an arbitrary field and is \emph{not} assumed to solve~\eqref{app:eq:helmholtz-state}. The partial derivatives of $\mathcal{L}$ below include the loss contribution.

Stationarity in the multiplier returns the constraint,
\begin{equation}\label{app:eq:helmholtz-La}
    \mathcal{L}_{\bfm{a}} = \bfm{F}\bigl[\bfm{u},\overline{\bfm{u}},\bfm{p}_F\bigr]-\bfm{f}(\bfm{p}_f),
\end{equation}
so $\mathcal{L}_{\bfm{a}}=0$ holds if and only if $\bfm{u}$ solves~\eqref{app:eq:helmholtz-state}, since a real pairing that vanishes in every complex direction forces the residual to vanish, as in Appendix~\ref{a:schrodinger-wirtinger}. At the solution, the constraint term vanishes and $\mathcal{L}=\chi(\bfm{p})$, so differentiating with one term per argument of $\mathcal{L}$ gives
\begin{equation}\label{app:ep:lagrangian-p-derivative-linear}
    \frac{d \chi}{d \bfm{p}} = \mathcal{L}_{\bfm{u}}\frac{d \bfm{u}}{d \bfm{p}} +  \mathcal{L}_{\bar{\bfm{u}}}\frac{d \bar{\bfm{u}}}{d \bfm{p}} + \mathcal{L}_{\bfm{a}}\frac{d \bfm{a}}{d \bfm{p}} + \mathcal{L}_{\bar{\bfm{a}}}\frac{d \bar{\bfm{a}}}{d \bfm{p}} + \mathcal{L}_{\bfm{p}}.
\end{equation}
The two multiplier terms vanish by~\eqref{app:eq:helmholtz-La}, whatever the sensitivities of $\bfm{a}$ may be. Computing $d\bfm{u}/d\bfm{p}$ would otherwise require one stationary solve per parameter. It appears only multiplied by $\mathcal{L}_{\bfm{u}}$ and its conjugate. Choosing $\bfm{a}$ so that $\mathcal{L}_{\overline{\bfm{u}}}=0$, which by conjugacy also gives $\mathcal{L}_{\bfm{u}}=0$, therefore leaves the key identity $d\chi/d\bfm{p}=\mathcal{L}_{\bfm{p}}$, and that condition is what defines the adjoint. Writing the real pairing out,
\begin{equation}
    2\operatorname{Re}\bigl\langle \bfm{a},\bfm{F}-\bfm{f}\bigr\rangle_{\mathbb{C}}
    =
    \overline{\bfm{a}}^{T}(\bfm{F}-\bfm{f})
    +
    \bfm{a}^{T}(\overline{\bfm{F}}-\overline{\bfm{f}}),
\end{equation}
and collecting the terms that carry $\overline{\bfm{u}}$, through the Wirtinger Jacobian $\bfm{F}_{\overline{\bfm{u}}}$ in the first term and $\overline{\bfm{F}}_{\overline{\bfm{u}}}=\overline{\bfm{F}_{\bfm{u}}}$ in the second,
\begin{equation}\label{app:eq:helmholtz-Lebar}
    \mathcal{L}_{\overline{\bfm{u}}} = \mathbb{P}\psi_{\overline{\bfm{u}}}\bigl[\mathbb{P}\bfm{u},\mathbb{P}\overline{\bfm{u}}\bigr] + \overline{\bfm{F}_{\bfm{u}}}^{\,T}\bfm{a} + \bfm{F}_{\overline{\bfm{u}}}^{T}\overline{\bfm{a}}.
\end{equation}
Setting it to zero gives the adjoint.

\begin{definition}[Stationary adjoint state]\label{app:def:helmholtz-adjoint}
The adjoint state $\bfm{a}\in\mathbb{C}^N$ associated with Problem~\ref{app:prob:helmholtz} is the solution of $\mathcal{L}_{\overline{\bfm{u}}}=0$, that is
\begin{equation}\label{app:eq:complex-adjoint-general}
    \overline{\bfm{F}_{\bfm{u}}}^{\,T}\bfm{a}
    +
    \bfm{F}_{\overline{\bfm{u}}}^{T}\overline{\bfm{a}}
    =
    -\,\mathbb{P}\psi_{\overline{\bfm{u}}}
    \bigl[\mathbb{P}\bfm{u},\mathbb{P}\overline{\bfm{u}}\bigr],
\end{equation}
with both Wirtinger Jacobians evaluated at the solution of~\eqref{app:eq:helmholtz-state}.
\end{definition}

\begingroup
\mdfsetup{nobreak=true}
\begin{theorem}[Stationary adjoint gradient]\label{app:thm:helmholtz-gradient}
Let $\bfm{u}$ solve Problem~\ref{app:prob:helmholtz} and $\bfm{a}$ be the adjoint state of Definition~\ref{app:def:helmholtz-adjoint}. Then, for each parameter $p_m$,
\begin{equation}\label{app:eq:adj-grad-general}
    \frac{d \chi}{d p_m} = 2 \operatorname{Re} \bigl\langle \bfm{a},\, \bfm{F}_{p_m}-\bfm{f}_{p_m} \bigr\rangle_{\mathbb{C}}.
\end{equation}
\end{theorem}
\endgroup
\noindent\textit{Proof.} With $\bfm{a}$ as in Definition~\ref{app:def:helmholtz-adjoint} we have $\mathcal{L}_{\overline{\bfm{u}}}=0$ and hence $\mathcal{L}_{\bfm{u}}=0$, so~\eqref{app:ep:lagrangian-p-derivative-linear} leaves $d\chi/d\bfm{p}=\mathcal{L}_{\bfm{p}}$, and the explicit parameter dependence of the Lagrangian~\eqref{app:eq:lagrangian-helmholtz} sits in $\bfm{F}$ and $\bfm{f}$ alone. \qed

\subsection{Linear systems}\label{app:helmholtz-linear}

For $\bfm{F}=\bfm{K}(\bfm{p}_\bfm{K})\bfm{u}$ with invertible $\bfm{K}$, the Wirtinger Jacobians are $\bfm{F}_{\bfm{u}}=\bfm{K}$ and $\bfm{F}_{\overline{\bfm{u}}}=\bfm{0}$, and the adjoint equation~\eqref{app:eq:complex-adjoint-general} becomes the Hermitian-adjoint system
\begin{equation}\label{app:eq:adj-state-helmholtz}
    \bfm{K}(\bfm{p}_\bfm{K})^{\dagger} \bfm{a} = -\,\mathbb{P} \psi_{\bfm{\bar{u}}}[\mathbb{P} \bfm{u}, \mathbb{P} \bar{\bfm{u}}],
\end{equation}
and the gradients of Theorem~\ref{app:thm:helmholtz-gradient} are
\begin{align}
    \frac{d \chi}{d \bfm{p}_{\bfm{K}}} &= \phantom{-}2 \operatorname{Re} \bigl\langle \bfm{a},\, \bfm{K}_{\bfm{p}_\bfm{K}} \bfm{u} \bigr\rangle_{\mathbb{C}}, \label{app:eq:adj-grad-pA}\\
    \frac{d \chi}{d \bfm{p}_\bfm{f}} &= -2 \operatorname{Re} \bigl\langle \bfm{a},\, \bfm{f}_{\bfm{p}_\bfm{f}} \bigr\rangle_{\mathbb{C}}. \label{app:eq:adj-grad-pb}
\end{align}
The conjugated equation~\eqref{app:eq:complex-adjoint-general} becomes
\begin{equation}\label{app:eq:adj-state-helmholtz-conjugated}
    \bfm{K}(\bfm{p}_\bfm{K})^{T}\bfm{c}
    =
    -\,\mathbb{P}\psi_{\bfm{u}}
    \bigl[\mathbb{P}\bfm{u},\mathbb{P}\overline{\bfm{u}}\bigr].
\end{equation}
For a reciprocal system, $\bfm{K}^T=\bfm{K}$, so $\bfm{c}$ is obtained by solving the \emph{same} forward problem with the source replaced by $-\mathbb{P}\psi_{\bfm{u}}$, applied at the measurement degrees of freedom. In terms of that physically generated field the gradients read
\begin{align}
    \frac{d \chi}{d \bfm{p}_{\bfm{K}}}
    &= \phantom{-}2\operatorname{Re}\bigl[\bfm{c}^{T}\bfm{K}_{\bfm{p}_\bfm{K}}\bfm{u}\bigr],
    \label{app:eq:adj-grad-real-pA}\\
    \frac{d \chi}{d \bfm{p}_\bfm{f}}
    &= -2\operatorname{Re}\bigl[\bfm{c}^{T}\bfm{f}_{\bfm{p}_\bfm{f}}\bigr].
    \label{app:eq:adj-grad-real-pb}
\end{align}
Eq.~\eqref{app:eq:adj-state-helmholtz-conjugated} is Eq.~\eqref{eq:main-helmholtz-adj-state-cc} of the main text, and~\eqref{app:eq:adj-grad-real-pA}--\eqref{app:eq:adj-grad-real-pb} are the gradients~\eqref{eq:stationary-gradient-F}--\eqref{eq:stationary-gradient-f} specialized to $\bfm{F}=\bfm{K}\bfm{u}$.

\subsection{Real fields: the static adjoint at a fixed point}\label{a:static-adjoint}

When $\bfm{F}$ is real and nonlinear, the trajectory-based construction is obstructed (Theorem~\ref{thm:first-order-physical} and Sec.~\ref{ss:first-order-nonlinear}). What remains accessible is the fixed point of the relaxation
\[
    \dot{\bfm{u}}(t) + \bfm{F}[\bfm{u}(t),\bfm{p}_F] - \bfm{f}(\bfm{p}_f) = \bfm{0},
\]
with a time-independent drive and a terminal-cost-only objective ($\theta\equiv0$). We assume that the relaxation converges to a linearly exponentially stable fixed point $\bfm{u}^*(\bfm{p})$. The loss $\psi[\mathbb{P}\bfm{u}^*]$ is measured after settling. The relevant parameters are those entering the fixed-point equation. Parameters affecting only the initial condition leave $\bfm{u}^*$ unchanged as long as the trajectory stays in the same basin of attraction. For real fields the conjugations in the derivation above drop and the factor $2\operatorname{Re}$ becomes the Euclidean inner product, so the adjoint equation~\eqref{app:eq:complex-adjoint-general} reduces to the following.

\begin{definition}[Static adjoint state]\label{a:def:static-adjoint}
The \emph{static adjoint} $\bfm{a}^*(\bfm{p})$ is the solution of
\begin{equation}\label{a:eq:fo-static-adjoint}
    \bfm{F}_{\bfm{u}}^T\bigl[\bfm{u}^*(\bfm{p}),\bfm{p}_F\bigr]\,\bfm{a}^* + \mathbb{P}\psi_{\bfm{u}}\bigl[\mathbb{P}\bfm{u}^*(\bfm{p})\bigr] = \bfm{0}.
\end{equation}
\end{definition}
\noindent No separate invertibility hypothesis is needed. Linearizing the forward dynamics about the fixed point gives $\delta\dot{\bfm{u}} = -\bfm{F}_{\bfm{u}}\delta\bfm{u}$, so linear exponential stability of $\bfm{u}^*$ places the spectrum of $\bfm{F}_{\bfm{u}}[\bfm{u}^*,\bfm{p}_F]$ in the open right half-plane, $\operatorname{Re}\lambda(\bfm{F}_{\bfm{u}})>0$. In particular $\bfm{F}_{\bfm{u}}$ is invertible, $\bfm{a}^*$ is the unique solution of~\eqref{a:eq:fo-static-adjoint}, and $\bfm{u}^*(\bfm{p})$, $\bfm{a}^*(\bfm{p})$ depend smoothly on $\bfm{p}$ by the implicit function theorem.

\begin{theorem}[Static adjoint gradient]\label{a:thm:static-gradient}
Let $\bfm{u}^*$ be the linearly exponentially stable fixed point above, so that $\bfm{F}_{\bfm{u}}[\bfm{u}^*,\bfm{p}_F]$ is invertible, and $\bfm{a}^*$ the static adjoint of Definition~\ref{a:def:static-adjoint}. Then
\begin{equation}\label{a:eq:fo-static-gradient}
    \frac{d\chi}{d\bfm{p}} = \bigl\langle \bfm{a}^*,\, \bfm{F}_{\bfm{p}_F}[\bfm{u}^*(\bfm{p}),\bfm{p}_F] - \bfm{f}_{\bfm{p}_f}(\bfm{p}_f)\bigr\rangle.
\end{equation}
\end{theorem}
\noindent\textit{Proof.} The real instance of Theorem~\ref{app:thm:helmholtz-gradient}. \qed

\paragraph{Relation to the time-dependent adjoint.}
The static adjoint is the time integral of the forward-time adjoint of Appendix~\ref{a:first-order-adjoint}. Take the forward system to be settled, $\bfm{u}(t)\equiv\bfm{u}^*$ on the whole window, which is the free phase of the protocol below. The linearization is then constant along the trajectory, and with $\theta\equiv0$ the running source vanishes, so~\eqref{a:eq:fo-adjoint-forward} is the homogeneous relaxation
\begin{equation}\label{a:eq:ep-static-relax}
    \dot{\bfm{a}}(t) + \bfm{F}_{\bfm{u}}^T[\bfm{u}^*,\bfm{p}_F]\,\bfm{a}(t) = \bfm{0},
    \qquad
    \bfm{a}(0) = -\mathbb{P}\psi_{\bfm{u}}\bigl[\mathbb{P}\bfm{u}^*\bigr],
    \qquad
    \bfm{a}(t) = e^{-\bfm{F}_{\bfm{u}}^T t}\,\bfm{a}(0).
\end{equation}
The spectrum of $\bfm{F}_{\bfm{u}}$ lies in the open right half-plane, so $\bfm{a}$ decays to zero and its late-time value carries no information. Its time integral does:
\begin{equation}\label{a:eq:ep-static-integral}
    \int_0^\infty \bfm{a}(t)\,dt = \bfm{F}_{\bfm{u}}^{-T}\bfm{a}(0) = -\bfm{F}_{\bfm{u}}^{-T}\,\mathbb{P}\psi_{\bfm{u}} = \bfm{a}^*,
\end{equation}
which is~\eqref{a:eq:fo-static-adjoint}. Over a window long compared with the decay time, truncating the integral to $[0,T]$ makes no difference. The two gradient formulas agree for the same reason: on the settled trajectory the second factor in the integrand of~\eqref{a:eq:fo-gradient-forward} is the constant $\bfm{F}_{\bfm{p}_F}[\bfm{u}^*,\bfm{p}_F] - \bfm{f}_{\bfm{p}_f}(\bfm{p}_f)$, so the time integral falls entirely on the adjoint and returns $\bfm{a}^*$, recovering~\eqref{a:eq:fo-static-gradient}. The initial-condition term drops out because $\bfm{a}(T)$ has decayed, matching the remark above that parameters entering only $\bfm{u}_0$ leave $\bfm{u}^*$ unchanged.

\paragraph{Physical realization: Equilibrium Propagation.}
The static adjoint~\eqref{a:eq:fo-static-adjoint} requires the transpose $\bfm{F}_{\bfm{u}}^T$, which is generally not realizable on the same device. Reciprocity removes it: for $\bfm{F}_{\bfm{u}}^T=\bfm{F}_{\bfm{u}}$, i.e.\ $\bfm{F}=E_{\bfm{u}}$, the static adjoint is the infinitesimal difference between the free fixed point and one nudged by the adjoint source, $\bfm{a}^*=\lim_{\epsilon\to0}(\bfm{u}^{*\epsilon}-\bfm{u}^*)/\epsilon$, the real-valued instance of Eq.~\eqref{eq:stationary-ep}. Both relaxations descend the energy $E$ tilted by the drive, with the second additionally tilted by $\epsilon\psi$. This gives the two-phase energy contrast of Scellier and Bengio~\cite{scellier2017equilibrium}. Substituting the difference of equilibria into Theorem~\ref{a:thm:static-gradient} gives the learning rule.

\subsection{Layered photonic networks}\label{app:helmholtz-layered-photonic}

Integrated photonic neural networks can be described as a sequence of
frequency-domain optical solves and elementwise activations. For
$\ell=1,\dots,L$, write
\begin{subequations}\label{app:eq:layer-state}
\begin{align}
    \bfm{K}_{\ell}(\bfm{p})\bfm{u}_{\ell}
    &=
    \bfm{B}_{\ell}\bfm{X}_{\ell-1},
    \label{app:eq:layer-field}\\
    \bfm{Z}_{\ell}
    &=
    \bfm{P}_{\ell}\bfm{u}_{\ell},
    \label{app:eq:layer-project}\\
    \bfm{X}_{\ell}
    &=
    f_{\ell}\bigl(\bfm{Z}_{\ell},\overline{\bfm{Z}}_{\ell}\bigr).
    \label{app:eq:layer-activation}
\end{align}
\end{subequations}
Here $\bfm{K}_{\ell}$ is the reciprocal frequency-domain operator of layer
$\ell$, $\bfm{B}_{\ell}$ injects the incoming modal amplitudes into the device,
and $\bfm{P}_{\ell}$ projects the internal field onto output modes. Eliminating
the field gives the usual layer transfer matrix
\begin{equation}\label{app:eq:layer-transfer}
    \bfm{Z}_{\ell}
    =
    \widehat{\bfm{W}}_{\ell}\bfm{X}_{\ell-1},
    \qquad
    \widehat{\bfm{W}}_{\ell}
    :=
    \bfm{P}_{\ell}\bfm{K}_{\ell}^{-1}\bfm{B}_{\ell}.
\end{equation}

Applying the Lagrangian construction to the three constraints in
Eq.~\eqref{app:eq:layer-state} gives three adjoint variables: the optical
adjoint field $\bfm{a}_{\ell}$, the pre-activation error
$\boldsymbol{\delta}_{\ell}$, and the output error
$\boldsymbol{\gamma}_{\ell}$. Here we pair each residual bilinearly, through $2\operatorname{Re}[\bfm{a}^T(\,\cdot\,)]$ rather than through the Hermitian $2\operatorname{Re}\langle\bfm{a},\,\cdot\,\rangle_{\mathbb{C}}$ used above. The two differ only by conjugating the multiplier. The bilinear choice produces transposed rather than Hermitian-adjoint operators, as realized by a reciprocal device. It is the layered counterpart of passing to $\bfm{c}=\overline{\bfm{a}}$ in Eq.~\eqref{app:eq:adj-state-helmholtz-conjugated}. The backward equations are
\begin{align}
    \bfm{K}_{\ell}^T\bfm{a}_{\ell}
    &=
    \bfm{P}_{\ell}^T\boldsymbol{\delta}_{\ell},
    \label{app:eq:layer-adj-field}\\
    \boldsymbol{\delta}_{\ell}
    &=
    (f_{\ell})_{\bfm{Z}}^T\boldsymbol{\gamma}_{\ell}
    +
    \overline{(f_{\ell})_{\overline{\bfm{Z}}}}^{\,T}
    \overline{\boldsymbol{\gamma}}_{\ell},
    \label{app:eq:layer-adj-activation}\\
    \boldsymbol{\gamma}_{\ell-1}
    &=
    \bfm{B}_{\ell}^T\bfm{a}_{\ell},
    \qquad
    \boldsymbol{\gamma}_{L}=-\psi_{\bfm{X}_{L}}.
    \label{app:eq:layer-adj-recursion}
\end{align}

Applied from $\ell=L$ down to $1$, these equations implement \emph{physical backpropagation} \cite{wagner1987multilayer, pai2023experimentally}: begin at the output error $\boldsymbol{\gamma}_L=-\psi_{\bfm{X}_L}$, pull it through the activation~\eqref{app:eq:layer-adj-activation} to obtain $\boldsymbol{\delta}_\ell$, solve the adjoint field equation~\eqref{app:eq:layer-adj-field} for $\bfm{a}_\ell$, and read off the preceding layer's output error $\boldsymbol{\gamma}_{\ell-1}=\bfm{B}_\ell^T\bfm{a}_\ell$.

If the optical layer is reciprocal, $\bfm{K}_{\ell}^T=\bfm{K}_{\ell}$, then
Eq.~\eqref{app:eq:layer-adj-field} is a physical solve in the same device,
excited from the output side with source
$\bfm{P}_{\ell}^T\boldsymbol{\delta}_{\ell}$. Combining
Eqs.~\eqref{app:eq:layer-adj-field} and~\eqref{app:eq:layer-adj-recursion}
gives
\begin{equation}\label{app:eq:layer-transpose-map}
    \boldsymbol{\gamma}_{\ell-1}
    =
    \bfm{B}_{\ell}^T\bfm{K}_{\ell}^{-T}
    \bfm{P}_{\ell}^T\boldsymbol{\delta}_{\ell}
    =
    \widehat{\bfm{W}}_{\ell}^T\boldsymbol{\delta}_{\ell}.
\end{equation}
Thus the physical adjoint solve implements the transpose of the forward optical
map.

The parameter gradient follows from the explicit parameter dependence of the
layer residuals:
\begin{equation}\label{app:eq:layer-gradient}
    \frac{d\chi}{dp_m}
    =
    2\,\operatorname{Re}
    \sum_{\ell=1}^{L}
    \left[
        \bfm{a}_{\ell}^T
        \bigl(
            (\bfm{K}_{\ell})_{p_m}\bfm{u}_{\ell}
            -
            (\bfm{B}_{\ell})_{p_m}\bfm{X}_{\ell-1}
        \bigr)
        -
        \boldsymbol{\delta}_{\ell}^T
        (\bfm{P}_{\ell})_{p_m}\bfm{u}_{\ell}
        -
        \boldsymbol{\gamma}_{\ell}^T
        (f_{\ell})_{p_m}
    \right].
\end{equation}
When the trainable parameters enter only through the optical operators
$\bfm{K}_{\ell}$, this reduces to the local field overlap
\begin{equation}\label{app:eq:layer-gradient-optical}
    \frac{d\chi}{dp_m}
    =
    2\,\operatorname{Re}
    \sum_{\ell=1}^{L}
    \bfm{a}_{\ell}^T
    (\bfm{K}_{\ell})_{p_m}
    \bfm{u}_{\ell}.
\end{equation}
For a local permittivity perturbation in a phase shifter region,
$(\bfm{K}_{\ell})_{\epsilon_m}$ is supported only on that region, so
Eq.~\eqref{app:eq:layer-gradient-optical} is measured by a local product of the
forward and adjoint optical fields.

\subsection{Equivalence with scattering-matrix learning rules}\label{app:helmholtz-green}

The learning rules derived for linear wave-scattering networks by Wanjura and Marquardt~\cite{wanjura2024fully} express each parameter gradient as a product of two measured scattering responses. We show that these two responses are the forward and adjoint fields of Appendix~\ref{app:helmholtz-linear}, so that their rule is the adjoint identity evaluated at individual ports.

Their neuron modes obey the coupled-mode and input--output equations
\begin{equation}\label{app:eq:green-def}
    \bfm{K}\bfm{u}=-i\boldsymbol{\Gamma}\bfm{q},
    \qquad
    \bfm{q}^{\mathrm{out}}=\bfm{S}\bfm{q},
    \qquad
    \bfm{S}=\bfm{I}+\boldsymbol{\Gamma}\mathcal{G}\boldsymbol{\Gamma},
\end{equation}
with $\bfm{q}$ the probe amplitudes, $\boldsymbol{\Gamma}=\operatorname{diag}(\sqrt{\kappa_1},\ldots,\sqrt{\kappa_N})$ the couplings to the probe waveguides, $\bfm{K}=\omega\bfm{I}-\bfm{H}$ and $\mathcal{G}=-i\bfm{K}^{-1}$. The dynamical matrix $\bfm{H}$ collects the detunings and couplings together with the decay rates $-\tfrac{i}{2}\boldsymbol{\Gamma}^{2}$, so $\bfm{K}$ is symmetric without being Hermitian: the network is reciprocal, $\bfm{K}^{T}=\bfm{K}$ and $\mathcal{G}^{T}=\mathcal{G}$, and the conjugated adjoint equation~\eqref{app:eq:adj-state-helmholtz-conjugated} applies.

\paragraph{Both fields are probe responses.}
Probing port $p$ with unit amplitude, $\bfm{q}=\bfm{e}_p$, the forward field is the $p$-th column of $\mathcal{G}$,
\begin{equation}\label{app:eq:green-forward}
    \bfm{u}=\sqrt{\kappa_p}\,\mathcal{G}\bfm{e}_p.
\end{equation}
The network output is a quadrature of the scattering amplitude recorded at port $r$, $y=\operatorname{Re}(e^{i\phi}S_{r,p})$, and $S_{r,p}=\delta_{r,p}+\sqrt{\kappa_r}\,u_r$, so with the real weight $\lambda:=\partial\chi/\partial y$ the loss derivative entering~\eqref{app:eq:adj-state-helmholtz-conjugated} is $\psi_{\bfm{u}}=\tfrac{\lambda}{2}e^{i\phi}\sqrt{\kappa_r}\,\bfm{e}_r$. Reciprocity then delivers the adjoint field as
\begin{equation}\label{app:eq:green-adjoint}
    \bfm{c}
    =-\bfm{K}^{-1}\psi_{\bfm{u}}
    =-\frac{i\lambda}{2}e^{i\phi}\sqrt{\kappa_r}\,\mathcal{G}\bfm{e}_r,
\end{equation}
the $r$-th column of the same $\mathcal{G}$. The adjoint field is the device's response to a probe at the readout port, so both factors of the gradient are scattering measurements on the same hardware.

\paragraph{The gradient is their product rule.}
Let $\theta_m$ be a trainable parameter of $\bfm{H}$, so that $\bfm{K}_{\theta_m}=-\bfm{H}_{\theta_m}$. Substituting~\eqref{app:eq:green-forward} and~\eqref{app:eq:green-adjoint} into the gradient~\eqref{app:eq:adj-grad-real-pA} and using $\mathcal{G}^{T}=\mathcal{G}$,
\begin{equation}\label{app:eq:green-gradient}
    \frac{d\chi}{d\theta_m}
    =2\operatorname{Re}\bigl[\bfm{c}^{T}\bfm{K}_{\theta_m}\bfm{u}\bigr]
    =\lambda\operatorname{Re}\!\left[
        e^{i\phi}\,i\sqrt{\kappa_p\kappa_r}\,
        \bigl(\mathcal{G}\bfm{H}_{\theta_m}\mathcal{G}\bigr)_{r,p}
      \right].
\end{equation}
The bracket is the sensitivity of the scattering matrix itself: differentiating $\bfm{S}=\bfm{I}+\boldsymbol{\Gamma}\mathcal{G}\boldsymbol{\Gamma}$ with $\partial_{\theta_m}\bfm{K}^{-1}=\bfm{K}^{-1}\bfm{H}_{\theta_m}\bfm{K}^{-1}$ gives
\begin{equation}\label{app:eq:green-dS}
    \frac{\partial S_{r,p}}{\partial\theta_m}
    =i\sqrt{\kappa_p\kappa_r}\,
      \bigl(\mathcal{G}\bfm{H}_{\theta_m}\mathcal{G}\bigr)_{r,p},
\end{equation}
so~\eqref{app:eq:green-gradient} is the chain rule $d\chi/d\theta_m=(\partial\chi/\partial y)(\partial y/\partial\theta_m)$, with contributions from several outputs summed.

\paragraph{Read at single ports.}
Inverting the input--output relation gives $\mathcal{G}_{m,n}=(S_{m,n}-\delta_{m,n})/\sqrt{\kappa_m\kappa_n}$, so the overlap in~\eqref{app:eq:green-dS} is a product of measured amplitudes. Writing $\bfm{E}_{j,\ell}$ for the matrix whose only nonzero entry is a unit entry at $(j,\ell)$, an on-site detuning has $\bfm{H}_{\Delta_j}=\bfm{E}_{j,j}$ and
\begin{equation}\label{app:eq:green-detuning}
    \frac{\partial S_{r,p}}{\partial\Delta_j}
    =\frac{i}{\kappa_j}\,
      (S_{r,j}-\delta_{r,j})(S_{j,p}-\delta_{j,p}),
\end{equation}
the response from the probe port to the perturbed site times the response from that site to the readout port, while a reciprocal coupling has $\bfm{H}_{J_{j,\ell}}=\bfm{E}_{j,\ell}+\bfm{E}_{\ell,j}$ and
\begin{equation}\label{app:eq:green-coupling}
    \frac{\partial S_{r,p}}{\partial J_{j,\ell}}
    =\frac{i}{\sqrt{\kappa_j\kappa_\ell}}
      \bigl[
        (S_{r,j}-\delta_{r,j})(S_{\ell,p}-\delta_{\ell,p})
        +(S_{r,\ell}-\delta_{r,\ell})(S_{j,p}-\delta_{j,p})
      \bigr],
\end{equation}
the sum over the two paths through the coupled pair. These are the learning rules of Ref.~\cite{wanjura2024fully}, written here in terms of the measured amplitudes $S_{m,n}-\delta_{m,n}$. Their rule is therefore an instance of the adjoint method, in which the drive and the measurement each address a single port and the backward pass is a second scattering experiment on the same device.

\subsection{Tangent-space adjoint for quantum equilibrium states}
\label{app:QEP-adjoint}

Quantum Equilibrium Propagation (QEP) obtains a loss gradient from the response
of ground-state observables to a Hamiltonian perturbation. The derivation below
identifies this response with the constrained stationary adjoint and relates
its tangent reciprocity to the Onsager symmetry used in
Ref.~\cite{Wanjura2025QuantumEquilibrium}.

\subsubsection{State equation and readout}

Consider a quantum system with Hilbert space $\mathbb{C}^N$ and Hamiltonian
\begin{equation}\label{eq:QEP-hamiltonian}
    \bfm{H}(\bfm{p})=\sum_j p_j\bfm{A}_j,
    \qquad p_j\in\mathbb{R},
    \qquad \overline{\bfm{A}_j}^{\,T}=\bfm{A}_j,
\end{equation}
where the operators $\bfm{A}_j$ are independent of the trainable parameters.
Inputs may enter through additional parameters held fixed during a training
step; these are suppressed here. Assume that $\bfm{H}$ has a nondegenerate
normalized ground state $\bfm{u}^*$ with energy $E_0$ and a positive gap
$\Delta=E_1-E_0$ to the first excited state. We use
$\langle\bfm{x},\bfm{z}\rangle_{\mathbb{C}}=\overline{\bfm{x}}^T\bfm{z}$.

The ground state satisfies the stationary equation
\begin{equation}\label{eq:QEP-stationary-state}
    \bfm{F}[\bfm{u},\overline{\bfm{u}},\bfm{p}]
    =\bfm{H}(\bfm{p})\bfm{u}-E(\bfm{u},\overline{\bfm{u}},\bfm{p})\bfm{u}
    =\bfm{0},
    \qquad
    E=\overline{\bfm{u}}^T\bfm{H}(\bfm{p})\bfm{u},
    \qquad \|\bfm{u}\|=1.
\end{equation}
On normalized states, $E$ is the Rayleigh quotient. Its dependence on the state
makes $\bfm{F}$ nonlinear even though $\bfm{H}$ is a linear operator. Ground-state
preparation selects the solution used for the computation.

The readout consists of expectation values of fixed Hermitian observables
$\bfm{O}_\ell$:
\begin{equation}\label{eq:QEP-readout}
    y_\ell(\bfm{u},\overline{\bfm{u}})
    =\langle\bfm{u},\bfm{O}_\ell\bfm{u}\rangle_{\mathbb{C}}\in\mathbb{R},
    \qquad
    \chi=\psi(\bfm{y}),
    \qquad \overline{\bfm{O}_\ell}^{\,T}=\bfm{O}_\ell.
\end{equation}
The real, differentiable loss $\psi$ has no explicit parameter dependence.
For a rank-one projector onto a measured degree of freedom, $y_\ell$ is the
probability of that measurement outcome. In the first step of the training, we evaluate the real-valued error signal
$\psi_{y_\ell}:=\partial\psi/\partial y_\ell$ at the free-state readout
$\bfm{y}(\bfm{u}^*,\overline{\bfm{u}^*})$.
The chain rule gives the adjoint source at the free state:
\begin{equation}\label{eq:QEP-adjoint-source}
    \bfm{f}_{\mathrm{adj}}
    =\chi_{\overline{\bfm{u}}}
    =\sum_\ell\psi_{y_\ell}\bfm{O}_\ell\bfm{u}^*
    =\bfm{B}\bfm{u}^*,
    \qquad
    \bfm{B}:=\sum_\ell\psi_{y_\ell}\bfm{O}_\ell.
\end{equation}
Thus the source is the action of a Hermitian operator on the state. Its
realization as a Hamiltonian perturbation requires the system to implement the
operators in $\bfm{B}$.

\subsubsection{Projected stationary adjoint}

We introduce the real Lagrangian
\begin{equation}
    \mathcal{L}
    =\chi+\overline{\bfm{a}}^T\bfm{F}
    +\bfm{a}^T\overline{\bfm{F}}.
\end{equation}
To differentiate, we extend the polynomial expression for $\bfm{F}$ to independent
ambient Wirtinger arguments, then restrict physical variations to the
normalized state with its arbitrary overall phase chosen consistently for nearby states.

Normalization imposes
$\operatorname{Re}\langle\bfm{u}^*,\delta\bfm{u}\rangle_{\mathbb{C}}=0$.
Choosing the phase locally so that
$\operatorname{Im}\langle\bfm{u}^*,\delta\bfm{u}\rangle_{\mathbb{C}}=0$ removes
the unobservable phase, leaving
\begin{equation}
    \mathcal{T}
    =\{\delta \bfm{u}\in\mathbb{C}^{N}:
       \langle\bfm{u}^*,\delta \bfm{u}\rangle_{\mathbb{C}}=0\},
    \qquad
    \mathbb{P}_{\mathcal{T}}=\bfm{I}-\bfm{u}^*\overline{\bfm{u}^*}^{\,T}.
\end{equation}
Consequently, the state and conjugate variations obey
\begin{equation}
    \delta\bfm{u}=\mathbb{P}_{\mathcal{T}}\delta\bfm{u},
    \qquad
    \delta\overline{\bfm{u}} = \overline{\delta \bfm{u}} = \overline{\mathbb{P}_{\mathcal{T}} \delta \bfm{u}}
    =\mathbb{P}_{\mathcal{T}}^T\delta\overline{\bfm{u}},
\end{equation}
where $\overline{\mathbb{P}_{\mathcal{T}}}^{\,T}=\mathbb{P}_{\mathcal{T}}$ and
$\overline{\mathbb{P}_{\mathcal{T}}}=\mathbb{P}_{\mathcal{T}}^T$.

Using the source~\eqref{eq:QEP-adjoint-source}, the constrained Wirtinger
variation is
\begin{align}
    \delta_{\overline{\bfm{u}}}\mathcal{L}
    &=\left(
        \bfm{f}_{\mathrm{adj}}+\overline{\bfm{F}_{\bfm{u}}}^{\,T}\bfm{a}
        +\bfm{F}_{\overline{\bfm{u}}}^T\overline{\bfm{a}}
      \right)^T\mathbb{P}_{\mathcal{T}}^T\delta\overline{\bfm{u}}
      \notag\\
    &=\left[\mathbb{P}_{\mathcal{T}}\left(
        \bfm{f}_{\mathrm{adj}}+\overline{\bfm{F}_{\bfm{u}}}^{\,T}\bfm{a}
        +\bfm{F}_{\overline{\bfm{u}}}^T\overline{\bfm{a}}
      \right)\right]^T\delta\overline{\bfm{u}}.
\end{align}
All Jacobians here and below are evaluated at the free state unless otherwise
specified. Requiring stationarity for every admissible variation gives
\begin{equation}\label{eq:QEP-adjoint}
    \mathbb{P}_{\mathcal{T}}\left(
        \overline{\bfm{F}_{\bfm{u}}}^{\,T}\bfm{a}
        +\bfm{F}_{\overline{\bfm{u}}}^T\overline{\bfm{a}}
        +\bfm{f}_{\mathrm{adj}}
    \right)=\bfm{0}.
\end{equation}
Because $\mathcal{L}$ is real, the conjugate stationarity condition follows as
well. The adjoint is chosen in $\mathcal{T}$, so
$\bfm{a}=\mathbb{P}_{\mathcal{T}}\bfm{a}$: the component along $\bfm{u}^*$ does
not contribute to its pairing with the tangent residual variation, as the
Jacobians below verify. This removes the multiplier redundancy; uniqueness also
requires invertibility on $\mathcal{T}$. Inserting this restriction and its
conjugate gives the equivalent form
\begin{equation}\label{eq:QEP-restricted-adjoint}
    \overline{\mathbb{P}_{\mathcal{T}}\bfm{F}_{\bfm{u}}\mathbb{P}_{\mathcal{T}}}^{\,T}\bfm{a}
    +\left(\mathbb{P}_{\mathcal{T}}\bfm{F}_{\overline{\bfm{u}}}\mathbb{P}_{\mathcal{T}}^T\right)^T\overline{\bfm{a}}
    +\mathbb{P}_{\mathcal{T}}\bfm{f}_{\mathrm{adj}}=\bfm{0}.
\end{equation}

\subsubsection{Projected Wirtinger Jacobians}

At fixed parameters, the independent variations of
$E=\overline{\bfm{u}}^T\bfm{H}\bfm{u}$ are
\begin{equation}
    \delta_{\bfm{u}}E=\overline{\bfm{u}}^T\bfm{H}\delta\bfm{u},
    \qquad
    \delta_{\overline{\bfm{u}}}E
    =(\bfm{H}\bfm{u})^T\delta\overline{\bfm{u}}.
\end{equation}
The product rule gives the ambient Jacobians
\begin{equation}
    \bfm{F}_{\bfm{u}}=\bfm{H}-E \bfm{I}-\bfm{u}\overline{\bfm{u}}^T\bfm{H},
    \qquad
    \bfm{F}_{\overline{\bfm{u}}}=-\bfm{u}(\bfm{H}\bfm{u})^T.
\end{equation}
Since $\bfm{H}\bfm{u}^*=E\bfm{u}^*$, Hermiticity also gives
$\overline{\bfm{u}^*}^{\,T}\bfm{H}=E\overline{\bfm{u}^*}^{\,T}$, so
\begin{equation}\label{app:eq:QEP-Wirtinger}
    \bfm{F}_{\bfm{u}}=\bfm{H}-E \bfm{I}-E\bfm{u}^*\overline{\bfm{u}^*}^{\,T},
    \qquad
    \bfm{F}_{\overline{\bfm{u}}}=-E\bfm{u}^*\bfm{u}^{*T}.
\end{equation}
The constrained linearization projects the output onto $\mathcal{T}$ and
restricts the two inputs to $\mathcal{T}$ and $\overline{\mathcal{T}}$,
respectively. Since $\mathbb{P}_{\mathcal{T}}\bfm{u}^*=0$, both rank-one terms
disappear:
\begin{equation}\label{eq:QEP-restricted-Jacobians}
    \begin{aligned}
        \mathbb{P}_{\mathcal{T}}\bfm{F}_{\bfm{u}}\mathbb{P}_{\mathcal{T}}
        &=\mathbb{P}_{\mathcal{T}}(\bfm{H}-E \bfm{I})\mathbb{P}_{\mathcal{T}},\\
        \mathbb{P}_{\mathcal{T}}\bfm{F}_{\overline{\bfm{u}}}\mathbb{P}_{\mathcal{T}}^T
        &=\bfm{0}.
    \end{aligned}
\end{equation}
The first projected Jacobian is Hermitian, because $\bfm{H}$ and
$\mathbb{P}_{\mathcal{T}}$ are Hermitian and $E\in\mathbb{R}$; the second is the
zero matrix, which is symmetric. Explicitly,
\begin{equation}\label{eq:QEP-projected-tangent-conditions}
    \begin{aligned}
        \overline{\mathbb{P}_{\mathcal{T}}\bfm{F}_{\bfm{u}}\mathbb{P}_{\mathcal{T}}}^{\,T}
        &=\mathbb{P}_{\mathcal{T}}\bfm{F}_{\bfm{u}}\mathbb{P}_{\mathcal{T}},\\
        \left(\mathbb{P}_{\mathcal{T}}\bfm{F}_{\overline{\bfm{u}}}\mathbb{P}_{\mathcal{T}}^T\right)^T
        &=\mathbb{P}_{\mathcal{T}}\bfm{F}_{\overline{\bfm{u}}}\mathbb{P}_{\mathcal{T}}^T.
    \end{aligned}
\end{equation}
These are the tangent reciprocity
conditions~\eqref{eq:stationary-nonlinear-cond} of
Theorem~\ref{thm:stationary}, restricted to the physical tangent space.
Substituting into~\eqref{eq:QEP-restricted-adjoint} yields
\begin{equation}\label{eq:reduced-adjoint-QEP}
    \mathbb{P}_{\mathcal{T}}(\bfm{H}-E \bfm{I})\bfm{a}
    +\mathbb{P}_{\mathcal{T}}\bfm{f}_{\mathrm{adj}}=\bfm{0},
    \qquad \bfm{a}\in\mathcal{T}.
\end{equation}

\subsubsection{Spectral-gap invertibility}

Let $\{\bfm{u}_n\}_{n=0}^{N-1}$ be an orthonormal eigenbasis, ordered
by energy:
\begin{equation}
    \bfm{H}\bfm{u}_n=E_n\bfm{u}_n,
    \qquad
    \bfm{u}_0=\bfm{u}^*,
    \qquad E_0=E.
\end{equation}
For a nondegenerate ground state, $\mathcal{T}$ is spanned by the excited
eigenvectors. Every $\delta \bfm{u}\in\mathcal{T}$ therefore expands as
$\delta \bfm{u}=\sum_{n\geq1}c_n\bfm{u}_n$, and
\begin{equation}
    \mathbb{P}_{\mathcal{T}}(\bfm{H}-E \bfm{I})\mathbb{P}_{\mathcal{T}}\delta\bfm{u}
    =\sum_{n\geq1}(E_n-E_0)c_n\bfm{u}_n.
\end{equation}
In particular, $\bfm{H}-E \bfm{I}$ preserves $\mathcal{T}$, so its restriction equals
that of the projected operator. With $\Delta:=E_1-E_0>0$, its spectrum and
inverse norm are
\begin{equation}\label{eq:QEP-gap-invertibility}
    \operatorname{spec}\!\left(\left.(\bfm{H}-E \bfm{I})\right|_{\mathcal{T}}\right)
    =\{E_n-E_0\}_{n\geq1},
    \qquad
    \left\|\left(\left.(\bfm{H}-E \bfm{I})\right|_{\mathcal{T}}\right)^{-1}\right\|
    =\Delta^{-1}.
\end{equation}
The inverse multiplies each $\bfm{u}_n$ by $(E_n-E_0)^{-1}$, and its induced
Hilbert space norm is the largest of these positive inverse eigenvalues. The
adjoint is then unique and can be written as
\begin{equation}
    \bfm{a}
    =-\sum_{n\geq1}
    \frac{\langle\bfm{u}_n,\bfm{f}_{\mathrm{adj}}\rangle_{\mathbb{C}}}
         {E_n-E_0}\bfm{u}_n,
    \qquad
    \|\bfm{a}\|\leq\Delta^{-1}\|\mathbb{P}_{\mathcal{T}}\bfm{f}_{\mathrm{adj}}\|.
\end{equation}
The gap controls the inverse norm; the condition number of the restricted
operator is $(E_{N-1}-E_0)/\Delta$. A physical implementation also requires ground-state
preparation for the free and nudged Hamiltonians. The spectral gap of
$\bfm{H}$ does not by itself establish the exponential stability of that
preparation dynamics on physical states.

\subsubsection{The Hamiltonian nudge generates the adjoint}

With the error signal components fixed at their free-state values, the nudged
Hamiltonian of Eq.~\eqref{eq:QEP-nudged-Hamiltonian} is
\[
    \bfm{H}_\epsilon=\bfm{H}+\epsilon\bfm{B}
    =\bfm{H}+\epsilon\sum_\ell\psi_{y_\ell}\bfm{O}_\ell.
\]
Let $\bfm{u}^{*,\epsilon}$ be its normalized ground state and write
$\langle\bfm{C}\rangle_\epsilon
:=\langle\bfm{u}^{*,\epsilon},\bfm{C}\bfm{u}^{*,\epsilon}\rangle_{\mathbb{C}}$
for an operator $\bfm{C}$. The nudged eigenvalue equation reads
\begin{equation}
    (\bfm{H}+\epsilon\bfm{B}-E_\epsilon\bfm{I})\bfm{u}^{*,\epsilon}
    =\bfm{0},
    \qquad E_\epsilon=\langle\bfm{H}_\epsilon\rangle_\epsilon.
\end{equation}
Since $E_\epsilon=\langle\bfm{H}\rangle_\epsilon
+\epsilon\langle\bfm{B}\rangle_\epsilon$, this is equivalently
\begin{equation}\label{app:eq:QEP-nudged-residual}
    \bfm{F}[\bfm{u}^{*,\epsilon},\overline{\bfm{u}^{*,\epsilon}},\bfm{p}]
    +\epsilon(\bfm{B}-\langle\bfm{B}\rangle_\epsilon\bfm{I})
    \bfm{u}^{*,\epsilon}=\bfm{0}.
\end{equation}
The expectation value subtraction accounts for the change in the ground state
energy and projects the nudge onto the normalized state's tangent space.

The nondegenerate ground state and positive gap ensure a smooth branch for
sufficiently small $\epsilon$. We choose its phase smoothly so that
$\langle\bfm{u}^*,\bfm{u}^{*,\epsilon}\rangle_{\mathbb{C}}$ is real and
positive. This is possible locally because the overlap equals one at
$\epsilon=0$. Differentiating the normalization gives
$\operatorname{Re}\langle\bfm{u}^*,\frac{d\bfm{u}^{*,\epsilon}}{d\epsilon}\rangle_{\mathbb{C}}=0$, and the
phase choice gives
$\operatorname{Im}\langle\bfm{u}^*,\frac{d\bfm{u}^{*,\epsilon}}{d\epsilon}\rangle_{\mathbb{C}}=0$, where
\begin{equation}\label{app:eq:QEP-nudge-expansion}
    \bfm{u}^{*,\epsilon}=\bfm{u}^*+\epsilon\frac{d\bfm{u}^{*,\epsilon}}{d\epsilon}+O(\epsilon^2),
    \qquad
    \bfm{v}=\left.\frac{d\bfm{u}^{*,\epsilon}}{d\epsilon}\right|_{\epsilon=0}
    \in\mathcal{T}.
\end{equation}
Subtracting the free stationary equation from
Eq.~\eqref{app:eq:QEP-nudged-residual} and retaining first-order terms yields
\begin{equation}
    \bfm{F}_{\bfm{u}}\frac{d\bfm{u}^{*,\epsilon}}{d\epsilon}
    +\bfm{F}_{\overline{\bfm{u}}}\overline{\frac{d\bfm{u}^{*,\epsilon}}{d\epsilon}}
    +(\bfm{B}-\langle\bfm{B}\rangle_0\bfm{I})\bfm{u}^*
    =\bfm{0}.
\end{equation}
Projecting onto $\mathcal{T}$ and using
Eq.~\eqref{eq:QEP-restricted-Jacobians} gives
\begin{equation}
    \mathbb{P}_{\mathcal{T}}(\bfm{H}-E_0\bfm{I})
    \mathbb{P}_{\mathcal{T}}\frac{d\bfm{u}^{*,\epsilon}}{d\epsilon}
    +\mathbb{P}_{\mathcal{T}}\bfm{f}_{\mathrm{adj}}=\bfm{0}.
\end{equation}
This is the reduced adjoint equation~\eqref{eq:reduced-adjoint-QEP}.
Invertibility on $\mathcal{T}$ therefore gives
\begin{equation}\label{eq:QEP-adj-tangent-formula}
    \bfm{a}
    =\left.\frac{d\bfm{u}^{*,\epsilon}}{d\epsilon}\right|_{\epsilon=0}
    =\lim_{\epsilon\to0}
    \frac{\bfm{u}^{*,\epsilon}-\bfm{u}^*}{\epsilon}.
\end{equation}
The plus sign in the Hamiltonian perturbation agrees with the adjoint-source
sign in the stationary Lagrangian.

\subsubsection{The gradient overlap is an observable response}

At the constrained stationary state, the adjoint cancels the state variation
of $\mathcal{L}$, so $d\chi/dp_j=\partial\mathcal{L}/\partial p_j$.
The loss has no explicit parameter dependence, and differentiating the
residual at fixed state gives
\begin{equation}
    \bfm{F}_{p_j}
    =(\bfm{A}_j-\langle\bfm{A}_j\rangle_0\bfm{I})\bfm{u}^*.
\end{equation}
Hence, the stationary adjoint gradient is
\begin{equation}\label{eq:QEP-gradient-overlap}
    \frac{d\chi}{dp_j}
    =2\operatorname{Re}\left\langle
        \bfm{a},(\bfm{A}_j-\langle\bfm{A}_j\rangle_0\bfm{I})\bfm{u}^*
      \right\rangle_{\mathbb{C}}
    =2\operatorname{Re}\langle\bfm{a},\bfm{A}_j\bfm{u}^*\rangle_{\mathbb{C}}.
\end{equation}
The second equality uses $\langle\bfm{a},\bfm{u}^*\rangle_{\mathbb{C}}=0$.
Hermiticity of $\bfm{A}_j$ and Eq.~\eqref{eq:QEP-adj-tangent-formula} give
\begin{align}
    \left.\frac{d\langle\bfm{A}_j\rangle_\epsilon}{d\epsilon}\right|_{\epsilon=0}
    &=\langle\left.\frac{d\bfm{u}^{*,\epsilon}}{d\epsilon}\right|_{\epsilon=0},\bfm{A}_j\bfm{u}^*\rangle_{\mathbb{C}}
      +\langle\bfm{u}^*,\bfm{A}_j\left.\frac{d\bfm{u}^{*,\epsilon}}{d\epsilon}\right|_{\epsilon=0}\rangle_{\mathbb{C}}\notag\\
    &=2\operatorname{Re}\langle\bfm{a},\bfm{A}_j\bfm{u}^*\rangle_{\mathbb{C}}
      =\frac{d\chi}{dp_j}.
      \label{app:eq:QEP-observable-response}
\end{align}
This proves the observable gradient formula~\eqref{eq:QEP-gradient-observable}
and the QEP rule of Ref.~\cite{Wanjura2025QuantumEquilibrium}:
\begin{equation}\label{eq:QEP-update-rule}
    \frac{d\chi}{dp_j}
    =\left.\frac{d\langle\bfm{A}_j\rangle_\epsilon}{d\epsilon}\right|_{\epsilon=0}
    =\frac{\langle\bfm{A}_j\rangle_\epsilon-\langle\bfm{A}_j\rangle_0}{\epsilon}
     +O(\epsilon).
\end{equation}
The experiment measures the change in $\langle\bfm{A}_j\rangle$, so it does
not require phase-resolved measurements of either the state or the adjoint. The
gradient readout is local whenever $\bfm{A}_j$ is local: a two-site coupling
requires the correlator of the two sites it couples.

\subsubsection{Onsager reciprocity}\label{app:QEP-onsager}

The static susceptibility symmetry was derived in the Supplementary Information of Ref.~\cite{Wanjura2025QuantumEquilibrium}. Here, we recast this derivation in terms of the two projected Wirtinger conditions.
Since
$\mathbb{P}_{\mathcal{T}}\bfm{F}_{\overline{\bfm{u}}}\mathbb{P}_{\mathcal{T}}^T=0$,
the tangent response has no conjugate-variation term and is governed solely by
$\mathbb{P}_{\mathcal{T}}\bfm{F}_{\bfm{u}}\mathbb{P}_{\mathcal{T}}$. We first define the reduced 
resolvent operator:
\begin{equation}\label{app:eq:QEP-resolvent}
    G:=-\left[\left.
        \bigl(\mathbb{P}_{\mathcal{T}}\bfm{F}_{\bfm{u}}\mathbb{P}_{\mathcal{T}}\bigr)
      \right|_{\mathcal{T}}\right]^{-1}
    =\left[\left.(E_0 \bfm{I}-\bfm{H})\right|_{\mathcal{T}}\right]^{-1}
    =\overline{G}^{\,T},
\end{equation}
where the existence of the inverse follows separately from the positive
gap~\eqref{eq:QEP-gap-invertibility} and Hermiticity follows from that of the tangent operator $\mathbb{P}_{\mathcal{T}}\bfm{F}_{\bfm{u}}\mathbb{P}_{\mathcal{T}}$.
The Onsager reciprocity condition used in Ref.~\cite{Wanjura2025QuantumEquilibrium} is
\begin{equation}
    \frac{\partial\langle\bfm{A}_\ell\rangle}{\partial p_j}
    =
    \frac{\partial\langle\bfm{A}_j\rangle}{\partial p_\ell},
\end{equation}
where the expectation values are evaluated in equilibrium. To derive this reciprocity from the projected tangent conditions~\eqref{eq:stationary-nonlinear-cond}, we choose the phase locally so that
$\partial\bfm{u}^*/\partial p_j\in\mathcal{T}$, and differentiate the ground-state
equation $\bfm{H}\bfm{u}^*=E_0\bfm{u}^*$ with respect to $p_j$. Since
$\partial\bfm{H}/\partial p_j=\bfm{A}_j$, this gives
\begin{equation}
    \bfm{A}_j\bfm{u}^*
    +\bfm{H}\frac{\partial\bfm{u}^*}{\partial p_j}
    =\frac{\partial E_0}{\partial p_j}\bfm{u}^*
    +E_0\frac{\partial\bfm{u}^*}{\partial p_j}.
\end{equation}
Projecting onto $\mathcal{T}$ removes the energy derivative term because
$\mathbb{P}_{\mathcal{T}}\bfm{u}^*=0$. Using
$\mathbb{P}_{\mathcal{T}}\partial\bfm{u}^*/\partial p_j
=\partial\bfm{u}^*/\partial p_j$, we obtain
\begin{equation}
    \mathbb{P}_{\mathcal{T}}(E_0\bfm{I}-\bfm{H})\mathbb{P}_{\mathcal{T}}
    \frac{\partial\bfm{u}^*}{\partial p_j}
    =\mathbb{P}_{\mathcal{T}}\bfm{A}_j\bfm{u}^*.
\end{equation}
On $\mathcal{T}$, the operator on the left has the inverse $G$ defined
in~\eqref{app:eq:QEP-resolvent}. Applying this inverse therefore yields
\begin{equation}\label{app:eq:QEP-state-derivative}
    \frac{\partial\bfm{u}^*}{\partial p_j}
    =G\mathbb{P}_{\mathcal{T}}\bfm{A}_j\bfm{u}^*.
\end{equation}
For the parameter-independent Hermitian operators $\bfm{A}_j$, differentiating
$\langle\bfm{A}_j\rangle_0=\langle\bfm{u}^*,\bfm{A}_j\bfm{u}^*\rangle_{\mathbb{C}}$
and using~\eqref{app:eq:QEP-state-derivative} gives
\begin{align}
    \frac{\partial\langle\bfm{A}_j\rangle_0}{\partial p_k}
    &=\left\langle\frac{\partial\bfm{u}^*}{\partial p_k},
        \bfm{A}_j\bfm{u}^*\right\rangle_{\mathbb{C}}
      +\left\langle\bfm{u}^*,
        \bfm{A}_j\frac{\partial\bfm{u}^*}{\partial p_k}\right\rangle_{\mathbb{C}}
      \notag\\
    &=2\operatorname{Re}\left\langle\bfm{A}_j\bfm{u}^*,
        \frac{\partial\bfm{u}^*}{\partial p_k}\right\rangle_{\mathbb{C}}
      \notag\\
    &=2\operatorname{Re}\left\langle
        \mathbb{P}_{\mathcal{T}}\bfm{A}_j\bfm{u}^*,
        G\mathbb{P}_{\mathcal{T}}\bfm{A}_k\bfm{u}^*
      \right\rangle_{\mathbb{C}}.
      \label{app:eq:QEP-onsager-susceptibility}
\end{align}
The projector in the first argument may be inserted because the state
derivative lies in $\mathcal{T}$. Hermiticity of $G$ on this space then allows
it to be moved between the two arguments. Taking the real part and using
conjugate symmetry of the inner product yields
\begin{align}
    \frac{\partial\langle\bfm{A}_j\rangle_0}{\partial p_k}
    &=2\operatorname{Re}\left\langle
        G\mathbb{P}_{\mathcal{T}}\bfm{A}_j\bfm{u}^*,
        \mathbb{P}_{\mathcal{T}}\bfm{A}_k\bfm{u}^*
      \right\rangle_{\mathbb{C}}
      \notag\\
    &=2\operatorname{Re}\left\langle
        \mathbb{P}_{\mathcal{T}}\bfm{A}_k\bfm{u}^*,
        G\mathbb{P}_{\mathcal{T}}\bfm{A}_j\bfm{u}^*
      \right\rangle_{\mathbb{C}}
      \notag\\
    &=\frac{\partial\langle\bfm{A}_k\rangle_0}{\partial p_j}.
      \label{app:eq:QEP-onsager-resolvent}
\end{align}

\subsubsection{Finite nudges and measurement noise}\label{app:QEP-finite-nudge}

A sufficient condition for staying in the perturbative regime is
$|\epsilon|\|\bfm{B}\|\ll\Delta$, where the norm is the induced Hilbert space
operator norm. The free and nudged preparations must follow the same smooth
ground-state branch. The inverse-norm bound~\eqref{eq:QEP-gap-invertibility}
allows a larger response as the gap closes, so smaller nudges may be needed
near a phase boundary. Centered nudging reduces the finite-difference bias:
\begin{equation}\label{app:eq:QEP-central-update}
	\frac{d\chi}{dp_j}
	=\frac{\langle\bfm{A}_j\rangle_{\epsilon}
		-\langle\bfm{A}_j\rangle_{-\epsilon}}{2\epsilon}
	+O(\epsilon^2).
\end{equation}
Both nudges use the same error signal from the free state. This requires
two nudged preparation settings in addition to the free preparation used to
determine those coefficients.

For fixed error signal, independent estimates of each expectation value
from $M$ measurements with single shot variance of order $\sigma^2$ give a
gradient estimator variance of order $\sigma^2/(\epsilon^2M)$. Decreasing
$|\epsilon|$ reduces the finite-nudge bias but increases this variance.

\section{How to test a gradient: the Taylor test}\label{a:taylor-test}
The \emph{Taylor test} provides a simple, model-agnostic check of an adjoint implementation and is widely used in PDE-constrained optimization~\cite{ghelichkhan2024automatic,farrell2013automated}.

Let $\bfm{g} := d\chi/d\bfm{p}$ denote the gradient produced by the adjoint procedure, and let $\bfm{v}\in\mathbb{R}^M$ be a random direction (e.g.\ a Gaussian vector). For a small step size $h>0$, define the zeroth- and first-order remainders
\begin{align}
    R_0(h) &\;:=\; \bigl|\,\chi(\bfm{p}+h\bfm{v}) \;-\; \chi(\bfm{p})\,\bigr|, \label{eq:R0} \\
    R_1(h) &\;:=\; \bigl|\,\chi(\bfm{p}+h\bfm{v}) \;-\; \chi(\bfm{p}) \;-\; h\,\langle \bfm{g},\bfm{v}\rangle\,\bigr|. \label{eq:R1}
\end{align}
A Taylor expansion of $\chi$ about $\bfm{p}$ gives
\begin{equation}\label{eq:taylor-rates}
    R_0(h) \;=\; h\,|\langle \bfm{g},\bfm{v}\rangle| \;+\; \mathcal{O}(h^2), \qquad
    R_1(h) \;=\; \mathcal{O}(h^2).
\end{equation}
The diagnostic is a log--log plot of $R_0(h)$ and $R_1(h)$ as $h$ is swept through several decades:
\begin{itemize}
    \item $R_0$ should decrease with slope $1$ in $\log h$.
    \item $R_1$ should decrease with slope $2$, \emph{provided} $\bfm{g}$ is the true gradient.
    \item If $\bfm{g}$ is wrong, $R_1$ instead decreases with slope $1$, and the curve typically transitions from a slope of $1$ at large $h$ to a noise-dominated floor at small $h$, producing the characteristic ``hockey-stick'' shape.
\end{itemize}
For very small $h$, $R_1$ becomes dominated by floating-point round-off in the forward evaluation. The Taylor test should therefore be assessed using the slope over the regime where $R_1$ still decreases monotonically.

Figure~\ref{fig:taylor-test} shows the test for the nonlinear Schr\"odinger example in the accompanying notebooks. Here $R_0(h)$ decreases with slope $1$ as required by Eq.~\eqref{eq:taylor-rates}. For the correct, physically obtained gradient, $R_1(h)$ decreases with slope $2$ over eight decades in $h$ before saturating at the round-off floor of the forward solve. For a gradient computed with the tangent reciprocity condition (A1) deliberately violated, $R_1(h)$ decreases with only slope $1$, remaining parallel to $R_0(h)$ rather than curving down to match the correct-gradient trend. Each accompanying notebook reports this diagnostic for its adjoint implementation.

\begin{figure}[h]
    \centering
    \includegraphics[width=0.5\textwidth]{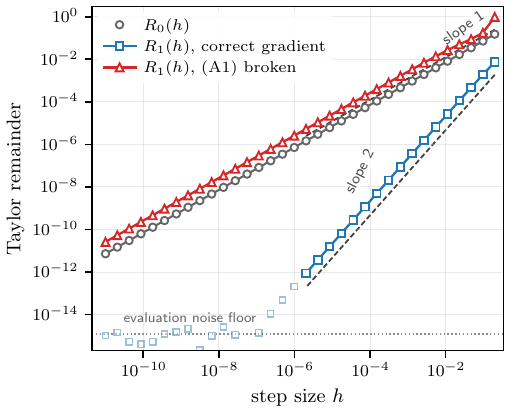}
    \caption{
    Taylor test for the nonlinear Schr\"odinger example, following Eqs.~\eqref{eq:R0}--\eqref{eq:taylor-rates}. Circles: $R_0(h)$, decreasing with slope $1$. Blue squares: $R_1(h)$ for the correct, physically obtained gradient, decreasing with slope $2$ until it saturates at the round-off floor of the forward evaluation (faint markers). Red triangles: $R_1(h)$ for a gradient computed with condition (A1) broken, which decreases with only slope $1$ and never curves down to the correct gradient trend. Dashed lines are slope-$1$ and slope-$2$ references.
    }
    \label{fig:taylor-test}
\end{figure}

\end{document}